\documentclass[lettersize,journal]{IEEEtran}
\usepackage{cite}
\usepackage{amsfonts}
\usepackage{amsmath,amsthm,amssymb}
\usepackage{algorithm}
\usepackage{algpseudocode}
\usepackage{array}
\usepackage[caption=false,font=normalsize,labelfont=sf,textfont=sf]{subfig}
\usepackage{textcomp}
\usepackage{stfloats}
\usepackage{stmaryrd,txfonts}
\usepackage{url}
\usepackage{booktabs}
\usepackage{verbatim}
\usepackage{graphicx}
\usepackage{tikz}
\usepackage{threeparttable}
\usepackage[colorlinks=true,linkcolor=blue,citecolor=blue,urlcolor=blue]{hyperref}

\makeatletter
\def\@cite#1#2{\textcolor{blue}{[{#1\if@tempswa , #2\fi}]}}
\makeatother
\let\oldeqref\eqref
\renewcommand{\eqref}[1]{\textcolor{blue}{\oldeqref{#1}}}

\definecolor{hyperBlue}{rgb}{0.81175232,0.84312439,1}
\definecolor{mintFill}{rgb}{0.69802856,0.90979004,0.87449646}
\definecolor{orangeFill}{rgb}{1,0.50979614,0.07843018}
\definecolor{outlineInk}{rgb}{0.06274414,0.09411621,0.26274109}
\definecolor{textInk}{rgb}{0.09803772,0.09803772,0.09803772}

\newtheorem{theorem}{Theorem}

\newtheorem{lemma}[theorem]{Lemma}
\newtheorem{definition}[theorem]{Definition}%[section]
\newtheorem{corollary}[theorem]{Corollary}%[section]
\newtheorem{conj}[theorem]{Conjecture}%[section]

\newtheorem{proposition}[theorem]{Proposition}

\newtheorem{remark}[theorem]{Remark}

\newcommand{\MD}{\ensuremath{\textsc{MinDistance}}}
\newcommand{\MDt}{\ensuremath{\textsc{MinDistance}_3}}
\newcommand{\MDJ}{\ensuremath{\textsc{MinDistance}_J}}
\newcommand{\MDtt}{\ensuremath{\textsc{MinDistance}_{(3,3)}}}
\newcommand{\MDtK}{\ensuremath{\textsc{MinDistance}_{(3,K)}}}
\newcommand{\MDJK}{\ensuremath{\textsc{MinDistance}_{(J,K)}}}

\newcommand{\MT}{\ensuremath{\textsc{MinTrap}}}
\newcommand{\MTt}{\ensuremath{\textsc{MinTrap}_3}}
\newcommand{\MTJ}{\ensuremath{\textsc{MinTrap}_J}}

\newcommand{\MTJK}{\ensuremath{\textsc{MinTrap}_{(J,K)}}}

\newcommand{\MEC}{\ensuremath{\textsc{MinEvenCover}}}

\newcommand{\MECJ}{\ensuremath{\textsc{MinEvenCover}_J}}

\newcommand{\MECJK}{\ensuremath{\textsc{MinEvenCover}_{(J,K)}}}

\newcommand{\MLECJK}{\ensuremath{\textsc{MinLinearEvenCover}_{(J,K)}}}

\makeatletter
\newcommand{\LongState}[1]{%
	\State \parbox[t]{\dimexpr\linewidth-\ALG@tlm\relax}{#1\strut}%
}
\newcommand{\LongStatex}[1]{%
	\Statex \hskip\ALG@tlm
	\parbox[t]{\dimexpr\linewidth-\ALG@tlm\relax}{#1\strut}%
}
\makeatother

\begin{document}

\title{On the Intractability of the Minimum Distance Problem for Regular LDPC Codes}

\author{Chenyuan Jia,
	Qingqing Peng,
	Ke Liu,
	Guanghui Wang,
	Guiying Yan
	\thanks{This work is partially supported by the National Key R\&D Program of China, (2023YFA1009602) (Corresponding author: Guanghui Wang).} 
	\thanks{C. Jia, Q. Peng and G. Wang are with the School of Mathematics, Shandong University, Shandong, China (e-mail: chenyuanjia@mail.sdu.edu.cn; pqing@mail.sdu.edu.cn; ghwang@sdu.edu.cn).}
	\thanks{G. Yan is with the Academy of Mathematics and Systems Science, CAS, University of Chinese Academy of Sciences, Beijing, 100190 China (e-mail: yangy@amss.ac.cn).}
	\thanks{K. Liu is with the Hangzhou Research Center, Huawei Technologies Co. Ltd., Hangzhou, China (e-mail: liuke79@huawei.com).}
}

%\markboth{Journal of \LaTeX\ Class Files,~Vol.~1, No.~2, December~2023}%

%\IEEEpubid{0000--0000~\copyright~2023 IEEE}

\maketitle

\begin{abstract}
	The minimum distance problem (MDP) for low-density parity-check (LDPC) codes is a central problem in coding theory and is closely related to the analysis of low-weight codewords and error-floor behavior. Although the unrestricted MDP is computationally intractable, its complexity under degree constraints that commonly occur in LDPC code design has remained less clear. In this paper, we study the MDP for left regular and biregular Tanner graphs. For every fixed $J\geq3$, we prove that the standard at-most-weight problem is $\mathrm{NP}$-complete for $J$-left regular Tanner graphs and that its exact-weight variant is $\mathrm{W}[1]$-complete when parameterized by the prescribed weight. For biregular Tanner graphs, we prove $\mathrm{NP}$-completeness for $(3,K)$-regular instances for every fixed $K\geq 3$ by replacing degree-two auxiliary completion blocks with a single-port high-girth gadget. A nonzero relative support inside this gadget induces an essentially cubic graph, so the Moore bound gives an exponential lower bound in the girth and allows a polynomial-size Karp reduction. Combining this right-degree amplification with a replica-and-global-check left-degree amplification yields $\mathrm{NP}$-completeness for $(J,K)$-regular Tanner graphs for every fixed $J,K\geq 3$. The reductions are based on a degree-preserving transformation framework consisting of hyperedge decomposition, check node splitting, and controlled variable replication. These transformations relate different degree distributions while preserving explicit maps among nonzero codewords, even covers, and nonempty $(a,0)$-trapping sets. The results delineate the computational limits of computing minimum distance exactly under natural regularity constraints.
\end{abstract}

\begin{IEEEkeywords}
Regular LDPC codes, minimum distance, exact-weight codeword problem, NP-completeness, W[1]-completeness.
\end{IEEEkeywords}

\section{Introduction} \label{sec:intro}

\IEEEPARstart{L}{OW-DENSITY} parity-check (LDPC) codes \cite{Gallager62} are fundamental in modern communication systems because they achieve capacity-approaching performance under \emph{belief-propagation} (BP) decoding. Their high-SNR performance, however, may be limited by the \emph{error-floor} phenomenon, where the \emph{bit error rate} (BER) decreases much more slowly than predicted by the waterfall-region behavior \cite{Richardson03}. Minimum distance is a key global parameter in this regime: a small minimum distance implies the existence of low-weight codewords that can dominate high-SNR performance, especially under maximum-likelihood or near-maximum-likelihood decoding. Under iterative BP decoding, small trapping or absorbing structures also play a central role. For this reason, computing and bounding the minimum distance of LDPC codes remains a basic problem in coding theory and code design \cite{Downey99, Arora97, Vardy97}.

The \emph{minimum distance problem} (MDP) has several equivalent combinatorial formulations. In the Tanner-graph representation, the support of a nonzero codeword is a nonempty set of variable nodes such that every adjacent check node is incident with an even number of selected variables. This is precisely a nonempty $(a,0)$-trapping set. In the corresponding incidence-hypergraph representation, variables become hyperedges and parity checks become vertices; codewords then correspond to \emph{even covers}, namely subfamilies of hyperedges in which every vertex has an even degree \cite{Feige08, Hsieh23}. These equivalent viewpoints connect minimum distance, even covers, and trapping sets through the same parity constraint.

The unrestricted MDP is known to be computationally hard. Classical results establish the $\mathrm{NP}$-completeness of maximum-likelihood decoding and related weight problems \cite{Berlekamp78}, the $\mathrm{NP}$-completeness of minimum distance computation \cite{Vardy97}, and strong inapproximability results for minimum distance \cite{Dumer03}. For trapping set-related structures, Dehghan and Banihashemi \cite{Dehghan19, Dehghan20} established In parameterized complexity, Arvind et al.~\cite{Arvind16} proved that the homogeneous exact-weight linear-equation problem is $\mathrm{W}[1]$-hard even when every variable occurs exactly three times. For the standard at-most-weight formulation, Bhattacharyya et al.~\cite{Bhattacharyya21} subsequently established $\mathrm{W}[1]$-hardness under randomized parameterized reductions, together with parameterized inapproximability, for the minimum distance problem (equivalently, the even-set problem). What remains less well understood is how this hardness behaves under the uniform-degree structures that arise naturally in LDPC Tanner graphs, such as left regular and $(J,K)$-regular configurations.

In this paper, we systematically study the hardness of computing the minimum distance in left regular and biregular LDPC Tanner graphs. For every fixed $J\geq3$, we prove that the standard at-most-weight problem is $\mathrm{NP}$-complete on $J$-left regular Tanner graphs and that its exact-weight variant is $\mathrm{W}[1]$-complete when parameterized by the prescribed weight. For the biregular setting, we prove $\mathrm{NP}$-completeness for $(3,K)$-regular bipartite graphs for every fixed $K\geq 3$. The new ingredient is a single-port completion gadget for the right-degree amplification: its variable nodes have degree three, its internal check nodes have degree $K$, and every nonempty relative support has size larger than the decision threshold. This relative-distance bound follows because such a support induces an essentially cubic graph, to which the Moore bound applies. The resulting Karp reduction, combined with the replica-and-global-check left-degree amplification, yields $\mathrm{NP}$-completeness for $(J,K)$-regular graphs for every fixed $J,K\geq 3$. The reductions rely on a degree-preserving transformation framework based on hyperedge decomposition, check node splitting, and controlled variable replication. Each transformation is designed to preserve explicit maps between feasible parity structures so that minimum-cardinality solutions can be tracked across different degree distributions.

Overall, the paper shows that natural regularity constraints do not, by themselves, remove the intrinsic difficulty of computing the minimum distance exactly. By presenting the minimum distance, even-cover, and trapping-set formulations in a unified framework, the results clarify the complexity boundary for regular LDPC code analysis and provide guidance on where exact algorithms are unlikely to be practical.

\section{Overview}\label{sec:overview}

This section provides a roadmap for the paper and highlights how the main hardness results are connected. At a high level, the argument proceeds in three layers. First, in Section~\ref{sec:prelim}, we formalize the equivalence among nonzero codewords, nonempty $(a,0)$-trapping sets, and nonempty even covers in the associated incidence hypergraph. This unified viewpoint allows us to move freely between coding-theoretic and combinatorial formulations of the MDP.

Second, Section~\ref{sec:inter} develops three local transformation primitives that preserve the relevant parity structures through explicit solution maps: hyperedge decomposition, check node splitting, and controlled variable replication. These transformations form the technical core of the paper. They allow us to adjust the left degree, right degree, and uniformity of the underlying incidence structures while keeping precise control over feasible solutions and their cardinalities.

Third, these primitives are assembled into the Karp-reduction chain shown in Fig.~\ref{fig:reduction-overview}. For the standard at-most-weight problem, Section~\ref{sec:left} reduces the unrestricted minimum even cover problem to the $3$-uniform case and then lifts the result to every fixed left degree $J\geq3$. The parameterized exact-weight result has a separate starting point: the sparse homogeneous instances of Arvind et al.~\cite{Arvind16} directly yield $3$-left regular instances, and the cardinality-preserving global-check construction lifts them to every fixed $J>3$. Next, Section~\ref{sec:regular} uses check-node splitting and variable replication to obtain classical hardness for the $(3,3)$-regular case. The subsequent right-degree amplification uses the single-port high-girth gadget and gives a Karp reduction to the $(3,K)$-regular case; the final left-degree amplification gives a Karp reduction to the $(J,K)$-regular case. Thus the classical results arise from a coherent reduction framework, while the exact-weight parameterized result is stated and proved separately wherever its source problem differs.

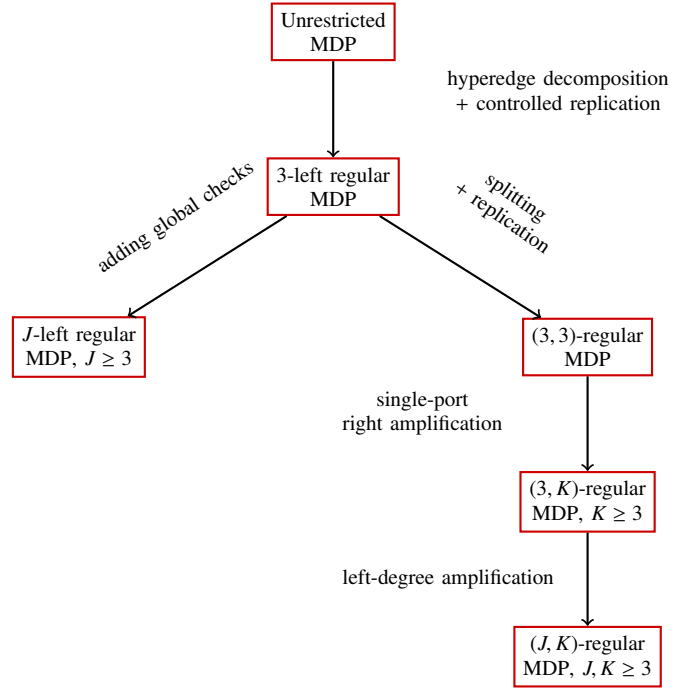
\begin{figure}[h]
	\centering
	\begin{tikzpicture}[x=0.95cm,y=0.95cm,scale=0.77]
		\node (mec) [draw=red!80!black, thick, align=center] at (0,4.8) {\footnotesize Unrestricted\\[-1mm] \footnotesize \textsc{MDP}};
		\node (md3) [draw=red!80!black, thick, align=center] at (0,2.0) {\footnotesize $3$-left regular\\[-1mm] \footnotesize \textsc{MDP}};
		\node (mdj) [draw=red!80!black, thick, align=center] at (-4.6,-0.9) {\footnotesize $J$-left regular\\[-1mm] \footnotesize \textsc{MDP}, $J\ge 3$};
		\node (reg33) [draw=red!80!black, thick, align=center] at (4.6,-0.9) {\footnotesize $(3,3)$-regular\\[-1mm] \footnotesize \textsc{MDP}};
		\node (reg3k) [draw=red!80!black, thick, align=center] at (4.6,-3.7) {\footnotesize $(3,K)$-regular\\[-1mm] \footnotesize \textsc{MDP}, $K\ge 3$};
		\node (regjk) [draw=red!80!black, thick, align=center] at (4.6,-6.5) {\footnotesize $(J,K)$-regular\\[-1mm] \footnotesize \textsc{MDP}, $J,K\ge 3$};
		
		\draw[->, thick] (mec) -- (md3);
		\draw[->, thick] (md3) -- (mdj);
		\draw[->, thick] (md3) -- (reg33);
		\draw[->, thick] (reg33) -- (reg3k);
		\draw[->, thick] (reg3k) -- (regjk);
		
		\node [right, align=center] at (1.9,3.7) {\footnotesize hyperedge decomposition\\[-1mm] \footnotesize + controlled replication};
		
		\node [above, align=center, rotate=34] at (-2.65,1.15) {\footnotesize adding global checks};
		
		\node [above, align=center, rotate=-35] at (2.85,1.15) {\footnotesize splitting\\[-1mm] \footnotesize + replication};
		
		\node [right, align=center] at (0,-2.10) {\footnotesize single-port\\[-1mm] \footnotesize right amplification};
		
		\node [right, align=center] at (0,-5.10) {\footnotesize left-degree amplification};
	\end{tikzpicture}
	\caption{Karp-reduction flow for the standard at-most-weight MDP under left regular and biregular degree constraints. The exact-weight $\mathrm{W}[1]$-completeness result for left regular graphs starts instead from the sparse homogeneous instances of Arvind et al.~\cite{Arvind16}.}
	\label{fig:reduction-overview}
\end{figure}

\definecolor{c_blue}{RGB}{204, 204, 255}
\definecolor{c_cyan}{RGB}{178, 235, 226}
\definecolor{c_bgblue}{RGB}{232, 234, 246}
\section{Preliminaries} \label{sec:prelim}

\subsection{Graphs and Hypergraphs}

A \emph{graph} is a pair $G=(V,E)$, where $V$ is a set of vertices and $E$ is a set of edges. The graph is \emph{undirected} if each edge connects two vertices without orientation. A graph is \emph{simple} if it has no loops and no multiple edges. Thus, for a simple undirected graph, each edge can be written as an unordered pair $\{u,v\}$ with $u,v\in V$ and $u\neq v$. Unless otherwise specified, all graphs referred to in the following text are simple undirected graphs. Two vertices $u,v\in V(G)$ are said to be \emph{adjacent} in $G$ if $\{u,v\}\in E(G)$. In this case, we also say that $u$ is adjacent to $v$, or that $u$ is a \emph{neighbor} of $v$. For a vertex $v\in V(G)$, its \emph{degree} in $G$, denoted by $\deg_G(v)$, is the number of vertices adjacent to $v$. A graph $G$ is called \emph{$k$-regular} if $\deg_G(v)=k$ for every $v\in V(G)$.

For a graph $G=(V,E)$ and a subset $U\subseteq V$, the \emph{induced subgraph} of $G$ on $U$, denoted by $G[U]$, is the graph with vertex set $U$ and edge set
\[
E(G[U])=\{\{u,v\}\in E:u,v\in U\}.
\]
A \emph{path} of length $\ell$ in $G$ is a sequence of distinct vertices
\[
v_0,v_1,\ldots,v_{\ell}
\]
such that $\{v_{i-1},v_i\}\in E$ for every $1\leq i\leq \ell$. A \emph{cycle} of length $\ell$ is a sequence
\[
v_0,v_1,\ldots,v_{\ell-1},v_0,
\]
where $v_0,v_1,\ldots,v_{\ell-1}$ are distinct, $\ell\geq 3$, and
\[
\{v_{i-1},v_i\}\in E \quad \text{for } 1\leq i\leq \ell-1,
\qquad
\{v_{\ell-1},v_0\}\in E.
\]
Equivalently, a cycle is obtained from a path by adding an edge between its two endpoints. The \emph{girth} of $G$, denoted by $g(G)$, is the length of the shortest cycle in $G$. If $G$ contains no cycle, then its girth is defined to be infinite.

We shall also use hypergraphs. A hypergraph generalizes an ordinary graph: in an ordinary graph, each edge is a 2-element subset of the vertex set, while in a hypergraph, each hyperedge is an arbitrary nonempty subset of the vertex set.

\begin{definition}[Hypergraph]\label{def:hypergraph}
A \emph{hypergraph} is denoted by
\[
\mathcal{H}=(V,\mathcal{E}),
\qquad 
\mathcal{E}=(e_{\alpha})_{\alpha\in I},
\]
where $V$ is a finite set of vertices, $I$ is a finite index set, and each
$e_{\alpha}$ is a nonempty subset of $V$, called a \emph{hyperedge}.
\end{definition}
% See Fig.~\ref{fig:even-cover} for an illustration of a hypergraph and its hyperedges.

Fig.~\ref{fig:even-cover} illustrates an example of a hypergraph. For instance, \(e_1=\{v_1,v_2,v_3\}\) is a hyperedge consisting of the vertices \(v_1\), \(v_2\), and \(v_3\). We use an indexed family of hyperedges rather than a set of subsets: if
$\alpha\neq \beta$ but $e_{\alpha}=e_{\beta}$, then $e_{\alpha}$ and
$e_{\beta}$ are still regarded as two distinct hyperedges. This convention is useful for representing multiple hyperedges.

A vertex $v\in V$ is said to be \emph{incident} with a hyperedge $e_{\alpha}$
if $v\in e_{\alpha}$. The degree of $v$ in $\mathcal{H}$ is
\[
\deg_{\mathcal{H}}(v)
=
|\{\alpha\in I:v\in e_{\alpha}\}|.
\]
The hypergraph is $J$-\emph{uniform} if every hyperedge has cardinality $J$,
that is, $|e_{\alpha}|=J$ for every $\alpha\in I$. It is $K$-\emph{regular}
if every vertex is incident with exactly $K$ hyperedges, that is,
$\deg_{\mathcal{H}}(v)=K$ for every $v\in V$. For \(B\subseteq I\), we use \(\mathcal{H}_B\) to denote the \emph{sub-hypergraph} consisting of the hyperedges indexed by \(B\). It is \emph{linear} if any two
distinct indexed hyperedges intersect in at most one vertex:
\[
|e_{\alpha}\cap e_{\beta}|\leq 1,
\qquad \forall \alpha\neq \beta .
\]

\begin{definition}[Even cover]\label{def:even-cover}

For a subset \(B\subseteq I\) of hyperedge indices, we say that
\(B\) is an \emph{even cover} of \(\mathcal{H}\) if every vertex is incident
with an even number of hyperedges indexed by \(B\), that is,
\[
\left|\{\alpha\in B:v\in e_{\alpha}\}\right|\equiv 0\pmod 2,
\qquad \forall v\in V .
\]
Vertices incident with no selected hyperedges are allowed, since zero is
considered even.

\end{definition}

Fig.~\ref{fig:even-cover} illustrates the notion of an even cover. The red hyperedges correspond to the selected index set $B$, while the black hyperedge is not selected. In the induced sub-hypergraph, every vertex is incident with an even number of selected hyperedges. Therefore, the red hyperedges form an even cover.

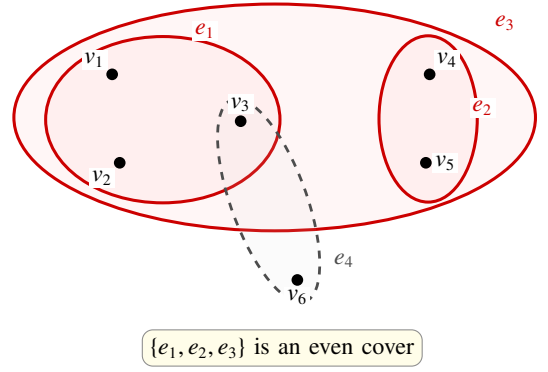
\begin{figure}[h]
\centering
\begin{tikzpicture}[
    x=1cm,
    y=1cm,
    every node/.style={font=\small},
    vertex/.style={
        circle,
        fill=black,
        inner sep=1.55pt
    },
    vlabel/.style={
        font=\small,
        fill=white,
        inner sep=1pt
    },
    selected/.style={
        line width=1.15pt,
        draw=red!80!black,
        fill=red!18,
        fill opacity=0.18
    },
    nonselected/.style={
        draw=black!70,
        dashed,
        line width=1.05pt,
        fill=black!12,
        fill opacity=0.10
    },
    edge label/.style={
        font=\small,
        fill=white,
        inner sep=1pt
    },
    note/.style={
        draw=black!45,
        rounded corners,
        fill=yellow!10,
        inner sep=3pt,
        font=\small
    }
]

% ------------------------------------------------
% Vertex coordinates
% ------------------------------------------------
\coordinate (v1) at (-2.45,  0.62);
\coordinate (v2) at (-2.35, -0.55);
\coordinate (v3) at (-0.75,  0.00);
\coordinate (v4) at ( 1.75,  0.62);
\coordinate (v5) at ( 1.70, -0.55);
\coordinate (v6) at ( 0.00, -2.10);

% ------------------------------------------------
% Hyperedges
% Draw hyperedges first so the vertices remain visible and clearly inside.
% ------------------------------------------------

% e_3 = {v_1,v_2,v_3,v_4,v_5}, selected in Y
\draw[selected]
    (-0.30,0.05) ellipse [x radius=3.45, y radius=1.50];

% e_1 = {v_1,v_2,v_3}, selected in Y
\draw[selected]
    (-1.78,0.02) ellipse [x radius=1.55, y radius=1.10];

% e_4 = {v_3,v_6}, not selected in Y
\draw[nonselected, rotate around={-70:(-0.38,-1.05)}]
    (-0.38,-1.05) ellipse [x radius=1.42, y radius=0.50];

% e_2 = {v_4,v_5}, selected in Y
\draw[selected]
    (1.73,0.03) ellipse [x radius=0.65, y radius=1.10];

% ------------------------------------------------
% Hyperedge labels
% ------------------------------------------------
\node[edge label, text=red!80!black]    at (-1.22,1.20) {$e_1$};
\node[edge label, text=red!80!black]    at ( 2.45,0.18) {$e_2$};
\node[edge label, text=red!80!black]    at ( 2.75,1.30) {$e_3$};
\node[edge label, text=black!70]        at ( 0.62,-1.84) {$e_4$};

% ------------------------------------------------
% Vertices and vertex labels
% ------------------------------------------------
\node[vertex, label={[vlabel]above left:$v_1$}]  at (v1) {};
\node[vertex, label={[vlabel]below left:$v_2$}]  at (v2) {};
\node[vertex, label={[vlabel]above:$v_3$}]       at (v3) {};
\node[vertex, label={[vlabel]above right:$v_4$}] at (v4) {};
\node[vertex, label={[vlabel]right:$v_5$}]       at (v5) {};
\node[vertex, label={[vlabel]below:$v_6$}]       at (v6) {};

% ------------------------------------------------
% Annotation
% ------------------------------------------------
\node[note] at (-0.20,-3)
    {$\{e_1,e_2,e_3\}$ is an even cover};

\end{tikzpicture}
\caption{
A hypergraph \(\mathcal H=(V,\mathcal E)\) illustrating an even cover.
The vertex set is
\(V=\{v_1,v_2,v_3,v_4,v_5,v_6\}\), and the hyperedge set is
\(\mathcal E=\{e_1,e_2,e_3,e_4\}\), where
\(e_1=\{v_1,v_2,v_3\}\),
\(e_2=\{v_4,v_5\}\),
\(e_3=\{v_1,v_2,v_3,v_4,v_5\}\), and
\(e_4=\{v_3,v_6\}\).
The selected hyperedges
\(\{e_1,e_2,e_3\}\) form an even cover:
with respect to \(B=\{1,2,3\}\), we have
\(\deg_{\mathcal{H}_B}(v_i)=2\) for \(i=1,\dots,5\), while
\(\deg_{\mathcal{H}_B}(v_6)=0\).
Thus every vertex has even degree in the selected sub-hypergraph.
The dashed black hyperedge \(e_4\) belongs to \(\mathcal E\setminus (e_{\alpha})_{\alpha\in B}\) and is not selected.
}
\label{fig:even-cover}
\end{figure}

\subsection{LDPC Codes and Trapping Sets}
A \emph{binary linear code} \(\mathcal{C}\) can be specified by a parity-check matrix \(H\in\mathbb{F}_2^{m\times n}\), where
\[
\mathcal{C}=\{\boldsymbol{x}\in\mathbb{F}_2^n: H\boldsymbol{x}^{\top}=\boldsymbol{0}\}.
\]
Equivalently, the codewords of \(\mathcal{C}\) are precisely the binary vectors satisfying all parity-check equations specified by the rows of \(H\). A \emph{low-density parity-check} (LDPC) code is a binary linear code that admits a sparse parity-check matrix \(H\), whose number of nonzero entries in $\mathcal{O}(n)$.

The matrix $H$ induces a Tanner graph $\mathcal{G}=(L,R,E)$, where the variable nodes $L$ correspond to the codeword symbols, the check nodes $R$ correspond to the parity-check equations, and $\{l_j,r_i\}\in E$ if and only if $H_{i,j}=1$. The Tanner graph is a simple bipartite graph with bipartition $(L,R)$. It is $J$-\emph{left regular} if every variable node has degree $J$, and $K$-\emph{right regular} if every check node has degree $K$. 

\begin{definition}[Nonzero codewords]\label{def:nc}
	A vector $\boldsymbol{c}\in\mathbb{F}_2^n$ is a nonzero codeword of $\mathcal{C}$ if $\boldsymbol{c}\neq\boldsymbol{0}$ and $H\boldsymbol{c}^{\top}=\boldsymbol{0}$. The Hamming weight of \(\boldsymbol{c}\), denoted by \(\operatorname{wt}(\boldsymbol{c})\), is the number of nonzero entries in \(\boldsymbol{c}\).
\end{definition}

\begin{definition}[$(a,b)$-trapping sets]\label{def:trap}
Let $\mathcal{G}=(L,R,E)$ be a Tanner graph. For $S\subseteq L$, let
\[
N(S)=\{r\in R:\text{ there exists } v\in S \text{ such that } \{v,r\}\in E\}
\]
be the set of neighboring check nodes of $S$, and let
\[
\mathcal{G}_S = \mathcal{G}[S\cup N(S)]
\]
be the subgraph induced by $S$ together with its neighboring check nodes.
The set $S$ is an $(a,b)$-trapping set if $|S|=a$ (i.e., $a$ is the cardinality of the $(a,b)$-trapping set) and exactly $b$ check nodes
in $\mathcal{G}_S$ have odd degree in $\mathcal{G}_S$.
\end{definition}

Since $(a,0)$-trapping sets occur frequently in the reduction arguments, 
we shall use the term $(a,0)$-TS to refer to a nonempty 
$(a,0)$-trapping set, unless stated otherwise.

\begin{proposition}[Codeword—$(a,0)$-TS correspondence]\label{prop:trap-codeword}
	Let $\mathcal{G}=(L,R,E)$ be the Tanner graph of a binary linear code $\mathcal{C}$ 
	with parity-check matrix $H$. For a subset $S\subseteq L$, define its 
	characteristic vector $\boldsymbol{c}_S\in\mathbb{F}_2^n$ by
	\[
	(\boldsymbol{c}_S)_j=
	\begin{cases}
	1, & \text{if } v_j\in S,\\
	0, & \text{if } v_j\notin S.
	\end{cases}
	\]
	Then a nonempty set $S\subseteq L$ is a $(a,0)$-TS if and only if 
	$\boldsymbol{c}_S$ is a nonzero codeword of $\mathcal{C}$. Consequently, 
	the minimum distance of $\mathcal{C}$ is equal to the minimum cardinality 
	of a nonempty $(a,0)$-TS in $\mathcal{G}$.
\end{proposition}

\begin{proof}
	For a subset $S\subseteq L$, consider its characteristic vector 
	$\boldsymbol{c}_S$. For the row $\boldsymbol{h}_i$ of $H$ corresponding 
	to the check node $r_i$, we have
	\[
	\boldsymbol{h}_i\boldsymbol{c}_S^{\top}
	=
	\sum_{j=1}^n H_{i,j}(\boldsymbol{c}_S)_j
	=
	\sum_{v_j\in S} H_{i,j}
	\pmod 2.
	\]
	The integer sum $\sum_{v_j\in S}H_{i,j}$ counts exactly the number of 
	neighbors of check node $r_i$ contained in $S$, that is, the degree of $r_i$ in the 
	induced subgraph $\mathcal{G}_S$ if $r_i\in N(S)$, and zero otherwise. Therefore, $S$ is an $(a,0)$-TS if and only if
	\[
	\sum_{v_j\in S} H_{i,j} \equiv 0 \pmod 2
	\]
	for every check node $r_i$. By the displayed identity above, this is equivalent to
	\[
	H\boldsymbol{c}_S^{\top}=\boldsymbol{0}.
	\]
	Since $S$ is 
	nonempty, $\boldsymbol{c}_S\neq \boldsymbol{0}$, and its Hamming weight is 
	$|S|$. Minimizing $|S|$ over all nonempty $(a,0)$-TSs is therefore 
	equivalent to minimizing the Hamming weight of a nonzero codeword.
\end{proof}

\subsection{Hypergraph Representation}

We next describe a construction of a hypergraph from the Tanner graph. Given
$\mathcal{G}=(L,R,E)$ be a Tanner graph, where
$R=\{r_1,\ldots,r_m\}$ is the set of check nodes and
$L=\{v_1,\ldots,v_n\}$ is the set of variable nodes.
Define a hypergraph $\mathcal{H}=(V,\mathcal{E})$ by taking
$V=\{r_1,\ldots,r_m\}$ to be the set of parity checks and by associating
with each variable node $v_j$ the indexed hyperedge
\[
e_j=\{r_i:(v_j,r_i)\in E\}.
\]

In this construction, the
check nodes of the Tanner graph become the vertices of the hypergraph, while
each variable node determines one indexed hyperedge through its incident
check nodes.
This construction is illustrated in Fig.~\ref{fig:hypergraph_example}. Under this representation, $J$-left regularity of Tanner graph corresponds to $J$-uniformity of
the hypergraph, and $K$-right regularity of Tanner graph corresponds to $K$-regularity of the
hypergraph.

% 可以写清楚G中码字和H中偶覆盖
% 已改
Through this construction, codewords of the code defined by the Tanner graph are in bijective correspondence with even covers of the associated hypergraph. This correspondence is formally stated and proved in Proposition~\ref{prop:code-even}~\cite{Hsieh23,Hsieh25}.

\begin{figure}[h]
	\centering
	\begin{tikzpicture}[scale=0.63]
		
		\tikzset{
			titlebox/.style={
				draw=black!50,
				rounded corners=2pt,
				fill=gray!10,
				inner sep=3pt,
				font=\normalsize
			}
		}
		
		% Left: Tanner graph title
		\node[titlebox] at (0,4.3) {Tanner graph};
		
		\node[circle, draw, fill=c_blue, minimum size=0.75cm] (C) at (0,3) {};
		
		\foreach \x in {1,2,3,4} {
			\node[rectangle, draw, fill=c_cyan, minimum size=0.7cm, font=\large] 
			(Sq\x) at (\x*1.5 - 3.75, 0) {\x};
			\draw (C) -- (Sq\x.north);
		}
		\node at (4.5,1.6) {\huge $\Longleftrightarrow$};
		% Right: Hypergraph title
		\node[titlebox] at (9,5.6) {Hypergraph};
		
		\draw[draw=black, fill=c_bgblue] (9,1.35) ellipse (1.8cm and 3.5cm);
		
		\foreach \y/\i in {3.5/1, 2.0/2, 0.5/3, -1.0/4} {
			\node[circle, draw, fill=c_cyan, minimum size=0.75cm, font=\large] 
			at (9,\y) {\i};
		}
		
	\end{tikzpicture}
	\caption{Hypergraph representation of a Tanner graph variable node and its adjacent check nodes.}
	\label{fig:hypergraph_example}
\end{figure}
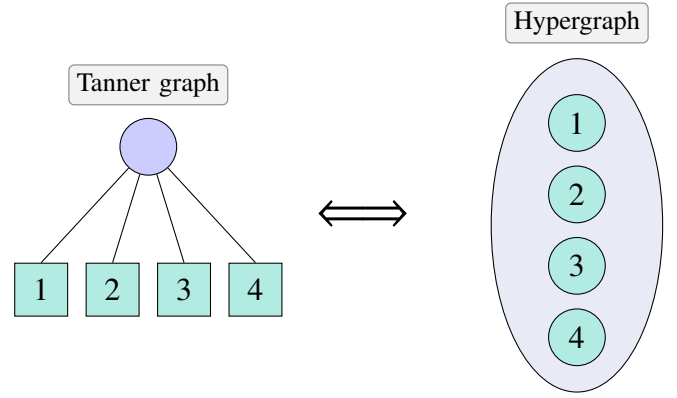

\begin{proposition}[Codeword—even cover correspondence]\label{prop:code-even}
	Let $\mathcal{C}$ be a binary linear code with parity-check matrix $H$, and let $\mathcal{H}=(V,\mathcal{E})$ be the hypergraph constructed from the Tanner graph of $H$ as above. For any $\boldsymbol{x}\in\mathbb{F}_2^n$, define
	\[
	Y_{\boldsymbol{x}}\coloneqq \{e_j: j\in\{1,\ldots,n\}, x_j=1, e_j\in \mathcal{E}\}.
	\]
	Then
	\[
	H\boldsymbol{x}^{\top}=\boldsymbol{0}
	\quad\Longleftrightarrow\quad
	Y_{\boldsymbol{x}}\text{ is an even cover of }\mathcal{H}.
	\]
	Consequently, nonzero codewords of $\mathcal{C}$ are in a bijection with nonempty even covers of $\mathcal{H}$ that preserves the Hamming weight (equivalently, the cardinality of $Y_{\boldsymbol{x}}$).
    % weight-preserving 偶覆盖是用weight表示吗？感觉换一个词更好
    % 已改
\end{proposition}

\begin{proof}
	The $i$th symbol of $H\boldsymbol{x}^{\top}$ is
	\[
	\sum_{j=1}^n H_{i,j}x_j \pmod 2,
	\]
	which is exactly the number, modulo two, of selected hyperedges containing the vertex $r_i$. Hence, the parity-check equations hold if and only if every hypergraph vertex has an even degree in the selected subfamily of hyperedges. The Hamming weight of $\boldsymbol{x}$ equals the number of selected hyperedges, giving the claimed weight-preserving bijection.
\end{proof}

\subsection{Complexity-Theoretic Background}

\subsubsection*{$\mathrm{NP}$-completeness}
The class $\mathrm{NP}$ consists of decision problems for which a ``yes'' instance admits a certificate verifiable in polynomial time. A problem is $\mathrm{NP}$-complete if it belongs to $\mathrm{NP}$ and every problem in $\mathrm{NP}$ is polynomial-time reducible to it~\cite{Garey79}.

\subsubsection*{Parameterized complexity}
Parameterized complexity studies decision problems together with an auxiliary parameter~\cite{Flum06}. A parameterized problem is a language $L\subseteq \Sigma^*\times\mathbb{N}$, whose instances are pairs $(x,w)$, where $w$ is the parameter.

\begin{definition}[Fixed-parameter tractability]\label{def:fpt}
A parameterized problem $L$ is \emph{fixed-parameter tractable} (\textsc{FPT}) if there exists an algorithm deciding whether $(x,w)\in L$ in time $f(w)|x|^{\mathcal{O}(1)}$, where $f$ is a computable function depending only on $w$.
\end{definition}

The class \textsc{FPT} contains all fixed-parameter tractable problems. An \textsc{FPT}-reduction from a parameterized problem $L_1$ to another parameterized problem $L_2$ maps an instance $(x,w)$ to an instance $(x',w')$ in time $f(w)|x|^{\mathcal{O}(1)}$ such that $(x,w)\in L_1$ if and only if $(x',w')\in L_2$, and $w'\leq g(w)$ for some computable function $g$~\cite{Erik26}.

The $\mathrm{W}$-hierarchy provides a standard framework for identifying parameterized problems that are unlikely to be fixed-parameter tractable~\cite{Rodney13}. Its first level is anchored by the following canonical complete problem.

\begin{definition}[Nondeterministic Turing Machine Acceptance \textsc{(Nondet TM Acceptance)}]\label{def:ntm-acceptance}
The parameterized problem \textsc{Nondet TM Acceptance} is defined as follows.
\begin{itemize}
    \item \textbf{Input:} A nondeterministic Turing machine $M$ and a positive integer $k$.
    \item \textbf{Question:} Does $M$ have an accepting computation path of length at most $k$ on the empty input?
    \item \textbf{Parameter:} $k$.
\end{itemize}
Using the empty input is the standard \emph{short-computation convention}; equivalent formulations that supply a fixed input instead can be recovered by encoding that input directly into the machine description.
\end{definition}

\begin{definition}[{$\mathrm{W}[1]$}]\label{def:w1-class}
A parameterized problem $L$ belongs to the class $\mathrm{W}[1]$ if and only if $L$ admits an \textsc{FPT}-reduction to \textsc{Nondet TM Acceptance}.
\end{definition}

\begin{definition}[{$\mathrm{W}[1]$}-completeness]\label{def:w1}
A parameterized problem is $\mathrm{W}[1]$-\emph{hard} if every problem in $\mathrm{W}[1]$ admits an \textsc{FPT}-reduction to it. It is $\mathrm{W}[1]$-\emph{complete} if it is $\mathrm{W}[1]$-hard and also belongs to $\mathrm{W}[1]$.
\end{definition}

\begin{theorem}[Standard intractability implications]\label{thm:hierarchy}
Under the usual complexity-theoretic assumptions, the following implications hold.
\begin{itemize}
\item If $\mathrm{P}\neq\mathrm{NP}$, then no $\mathrm{NP}$-complete problem admits a polynomial-time algorithm.
\item If $\mathrm{FPT}\neq\mathrm{W}[1]$, then no $\mathrm{W}[1]$-complete problem is fixed-parameter tractable.
\end{itemize}
\end{theorem}

\subsection{Problem Formulations}
The MDP for binary linear codes has three equivalent perspectives: parity-check matrices, Tanner graphs, and incidence hypergraphs. We use these formulations interchangeably in the reductions. In all variants below, the parameter is the target size $w$.

\paragraph{General instances.}
We first consider the unrestricted setting.

\begin{itemize}
	\item \MD: Given a binary matrix $H\in\mathbb{F}_2^{m\times n}$ and an integer $w>0$, decide whether there exists a nonzero vector $x\in\mathbb{F}_2^n$ such that
	\[
	\operatorname{wt}(x)\leq w
	\quad\text{and}\quad
	Hx^{\top}=\boldsymbol{0}.
	\]
	
	\item \MT: Given a bipartite graph $\mathcal{G}=(L,R,E)$ and an integer $w>0$, decide whether there exists a nonempty set $S\subseteq L$ with $|S|\leq w$ such that
	\[
	|N(r)\cap S|\equiv 0\pmod 2
	\qquad\text{for every } r\in R.
	\]

    \item \MEC: Given a hypergraph $\mathcal{H}=(V,\mathcal{E})$ with
    $\mathcal{E}=(e_\alpha)_{\alpha\in I}$ and an integer $w>0$,
    decide whether there exists a nonempty set $B\subseteq I$ with
    $|B|\leq w$ such that every vertex is incident with an even number of
    hyperedges indexed by $B$, that is,
    \[
    \left|\{\alpha\in B:v\in e_\alpha\}\right|\equiv 0\pmod 2
    \qquad\text{for every } v\in V .
    \]
\end{itemize}

The equivalence of these formulations follows from Proposition~\ref{prop:trap-codeword}, Proposition~\ref{prop:code-even}, and the standard incidence transformations between matrices, Tanner graphs, and hypergraphs. Hence, the three formulations transfer hardness results among one another.

\paragraph{Column-regular instances.}
For a fixed integer \(J\geq 3\), the column-regular restrictions of the
three formulations are equivalent under the above incidence transformations:
constant column weight \(J\) in the parity-check matrix corresponds to
\(J\)-left regularity of the Tanner graph and to \(J\)-uniformity of the
associated hypergraph. We denote the corresponding restricted variants by
\(\MDJ\), \(\MTJ\), and \(\MECJ\), respectively. The decision question and
the parameter \(w\) remain the same as in the unrestricted case.

\paragraph{Biregular instances.}
Similarly, for fixed integers \(J\) and \(K\), imposing constant column
weight \(J\) and constant row weight \(K\) on the parity-check matrix is
equivalent to requiring the Tanner graph to be \((J,K)\)-regular, and to
requiring the associated hypergraph to be \(J\)-uniform and \(K\)-regular.
We denote the corresponding restricted variants by \(\MDJK\), \(\MTJK\),
and \(\MECJK\), respectively. Again, the decision question and the parameter
\(w\) are unchanged.

\begin{remark}[Parameterized equivalence]\label{rem:equiv}
	Within each structural class—unrestricted, column-regular, and biregular—the matrix, Tanner graph, and hypergraph formulations are equivalent under parameter-preserving transformations. 
\end{remark}

\paragraph{Exact-weight variants.}
For the parameterized complexity results, we additionally consider the
exact-weight counterparts of the preceding problems. They are denoted by
$\MD^{=}$, $\MT^{=}$, and $\MEC^{=}$, respectively.

\begin{itemize}
	\item $\MD^{=}$: Given a binary matrix
	$H\in\mathbb{F}_2^{m\times n}$ and an integer $w>0$, decide whether
	there exists a vector $x\in\mathbb{F}_2^n$ such that
	\[
		Hx^{\top}=\boldsymbol{0}
		\quad\text{and}\quad
		\operatorname{wt}(x)=w.
	\]

	\item $\MT^{=}$: Given a Tanner graph
	$\mathcal{G}=(L,R,E)$ and an integer $w>0$, decide whether there
	exists a set $S\subseteq L$ such that
	\[
		|S|=w
		\quad\text{and}\quad
		|N(r)\cap S|\equiv0\pmod 2
		\quad\text{for every }r\in R.
	\]

	\item $\MEC^{=}$: Given a hypergraph
	$\mathcal{H}=(V,\mathcal{E})$ and an integer $w>0$, decide whether
	there exists an even cover $B\subseteq\mathcal{E}$ satisfying
	\[
		|B|=w.
	\]
\end{itemize}

The corresponding column-regular and biregular restrictions are denoted
by $\MDJ^{=}$, $\MTJ^{=}$, and $\MECJ^{=}$, and by
$\MDJK^{=}$, $\MTJK^{=}$, and $\MECJK^{=}$, respectively.
Since the standard incidence transformations preserve cardinality
exactly, the matrix, Tanner-graph, and hypergraph formulations remain
equivalent under parameter-preserving transformations.

\definecolor{myblue}{RGB}{115, 140, 245}
\definecolor{mycyan}{RGB}{160, 250, 250}
\definecolor{myorange}{RGB}{250, 225, 195}
\definecolor{mygreen}{RGB}{150, 215, 105}

\section{Intermediate Transformations} \label{sec:inter}
Although the minimum distance problem is known to be NP-complete for arbitrary linear codes, this result alone does not settle the complexity of the problem for highly structured subclasses such as regular LDPC codes. Consequently, regularity constraints require a separate hardness analysis rather than a direct appeal to the unrestricted case. Our strategy is to establish hardness by reducing from the general minimum distance problem. Transforming arbitrary instances into regular or nearly regular configurations without trivializing the underlying parity problem requires specialized algebraic graph operations. In Sections~\ref{sec:left} and~\ref{sec:regular}, we prove NP-completeness and parameterized hardness for left-regular and biregular LDPC code ensembles, respectively. To prepare for these reductions, this section introduces three transformation primitives that will be used repeatedly: \emph{hyperedge decomposition}, \emph{check node splitting}, and \emph{variable replication}. We then formalize their structural properties through three corresponding lemmas, each showing how the relevant parity structure is preserved under the transformation.

\subsection{Hyperedge Decomposition}

We first introduce a transformation that enforces the uniformity condition required by the Tanner-graph representation. When the variable node degree is fixed to three, the corresponding hypergraph must be $3$-uniform, since each variable node is represented by a hyperedge incident with exactly three check nodes. General minimum-distance instances, however, do not satisfy this restriction: their hypergraph representations may contain hyperedges of arbitrary size. Algorithm~\ref{alg:hyp} gives a hyperedge decomposition procedure that converts such hyperedges into chains of size-three components while introducing auxiliary vertices to preserve the parity structure.

\begin{algorithm}[h]
	\caption{Decomposition of Large Hyperedges into 3-Uniform Components}
	\label{alg:hyp}
	\begin{algorithmic}[1]
		\Require A Hypergraph $\mathcal{H} = (V, \mathcal{E})$
		\Ensure A hypergraph $\mathcal{H} = (V\cup V_{\mathrm{aux}}, \mathcal{E}')$ with all hyperedges of size at most 3
		\State Initialize $V_{\mathrm{aux}} \gets \emptyset,\ \mathcal{E}' \gets \emptyset$
		\For{each hyperedge $e \in \mathcal{E}$}
		\State Let $k \gets |e|$
		\LongState{Enumerate the vertices of $e$ as $v_1, v_2, \dots, v_k$, such that $e = \{v_1, v_2, \dots, v_k\}$}
		\If{$k > 3$}
		\State Introduce auxiliary vertices $x_1, x_2, \dots, x_{k-3}$
		\State Construct new hyperedges:
		\State $e^\prime_1 \gets \{v_1, v_2, x_1\}$
		\State $e^\prime_{k-2} \gets \{x_{k-3}, v_{k-1}, v_k\}$
		\For{$j \gets 2$ \textbf{to} $k-3$}
		\State $e^\prime_j \gets \{x_{j-1}, v_{j+1}, x_j\}$
		\EndFor
		\State Update vertex and edge sets:
		\State $V_{\mathrm{aux}} \gets V_{\mathrm{aux}} \cup \{x_1, \dots, x_{k-3}\}$
		\State $\mathcal{E}' \gets \mathcal{E}' \cup \left\{e^\prime_1, e^\prime_2, \dots, e^\prime_{k-2}\right\}$
        \Else
        \State $\mathcal{E}' \gets \mathcal{E}' \cup \{e\}$
		\EndIf
		\EndFor
		\State \textbf{return} $(V\cup V_{\mathrm{aux}}, \mathcal{E}')$
	\end{algorithmic}
\end{algorithm}

Fig.~\ref{fig:hyperedge_decomp} illustrates the execution of this decomposition on a hyperedge of size six. The algorithm introduces three auxiliary vertices and uses them to split the original hyperedge into four consecutive $3$-uniform components. The first component contains the first two original vertices, the last component contains the last two original vertices, and the two middle components each contain one remaining original vertex together with two auxiliary vertices. The procedure processes each hyperedge $e \in \mathcal{E}$ independently. For a hyperedge of cardinality $|e| > 3$, it introduces $|e|-3$ auxiliary vertices and constructs $|e|-2$ new $3$-uniform hyperedges. From a computational efficiency standpoint, this requires $\mathcal{O}(|e|)$ operations per decomposed hyperedge. Summing over all hyperedges, the total running time is $\mathcal{O}\left( \sum_{e \in \mathcal{E}} |e| \right)$, which is linear in the total size of the input hypergraph representation. The subsequent mathematical properties of the resulting chain structure, including how it bijectively maps between the even covers of the original and decomposed hypergraphs, are established in Lemma~\ref{lem:3.1}.

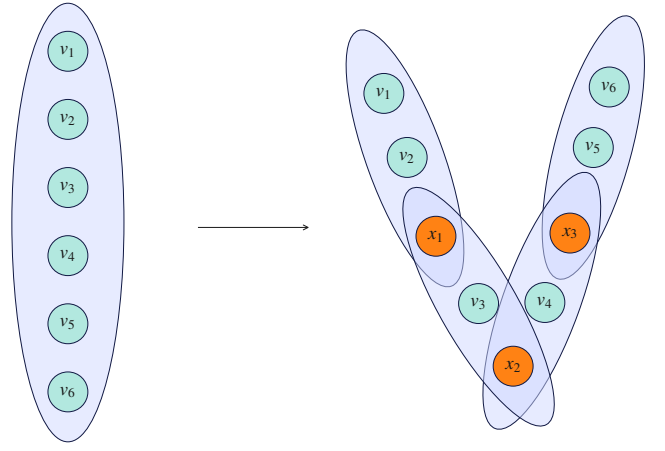
\begin{figure}[h]
	\centering
			\begin{tikzpicture}[
			x={0.0006175772\textwidth},y={-0.0006175772\textwidth},
			vertexlabel/.style={font=\scriptsize,text=textInk,inner sep=0pt}
		]
			\path[use as bounding box] (0,0) rectangle (842,595);
			\begin{scope}[scale=0.75]
				\path[fill=hyperBlue,fill opacity=0.498039,draw=outlineInk,line width=0.25pt,line cap=butt,line join=miter,miter limit=1.207107,nonzero rule] (927.713542,486.625)
				.. controls (895.0625,477.875) and (893.703125,377.083333) .. (924.671875,261.5)
				.. controls (955.645833,145.916667) and (1007.21875,59.307292) .. (1039.869792,68.057292)
				.. controls (1072.520833,76.802083) and (1073.880208,177.59375) .. (1042.90625,293.182292)
				.. controls (1011.9375,408.765625) and (960.364583,495.375) .. (927.713542,486.625)
				-- cycle;
			\end{scope}
			\begin{scope}[scale=0.75]
				\path[fill=hyperBlue,fill opacity=0.498039,draw=outlineInk,line width=0.25pt,line cap=butt,line join=miter,miter limit=1.207107,nonzero rule] (827.630208,725.161458)
				.. controls (795.864583,713.604167) and (803.296875,613.072917) .. (844.21875,500.630208)
				.. controls (885.145833,388.182292) and (944.072917,306.401042) .. (975.838542,317.963542)
				.. controls (1007.598958,329.520833) and (1000.171875,430.052083) .. (959.244792,542.494792)
				.. controls (918.317708,654.942708) and (859.390625,736.723958) .. (827.630208,725.161458)
				-- cycle;
			\end{scope}
			\begin{scope}[scale=0.75]
				\path[fill=hyperBlue,fill opacity=0.498039,draw=outlineInk,line width=0.25pt,line cap=butt,line join=miter,miter limit=1.207107,nonzero rule] (757.807292,498.151042)
				.. controls (726.041667,509.713542) and (667.114583,427.932292) .. (626.1875,315.484375)
				.. controls (585.260417,203.041667) and (577.833333,102.510417) .. (609.59375,90.953125)
				.. controls (641.359375,79.390625) and (700.286458,161.177083) .. (741.213542,273.619792)
				.. controls (782.140625,386.067708) and (789.567708,486.59375) .. (757.807292,498.151042)
				-- cycle;
			\end{scope}
			\begin{scope}[scale=0.75]
				\path[fill=hyperBlue,fill opacity=0.498039,draw=outlineInk,line width=0.25pt,line cap=butt,line join=miter,miter limit=1.207107,nonzero rule] (909.291667,718.682292)
				.. controls (880.015625,735.583333) and (807.78125,665.276042) .. (747.953125,561.645833)
				.. controls (688.119792,458.015625) and (663.348958,360.302083) .. (692.619792,343.401042)
				.. controls (721.895833,326.5) and (794.130208,396.8125) .. (853.958333,500.442708)
				.. controls (913.791667,604.072917) and (938.5625,701.78125) .. (909.291667,718.682292)
				-- cycle;
			\end{scope}
			\begin{scope}[scale=0.75]
				\path[fill=hyperBlue,fill opacity=0.498039,draw=outlineInk,line width=0.25pt,line cap=butt,line join=miter,miter limit=1.207107,nonzero rule] (146.5625,746.3125)
				.. controls (97.234375,746.3125) and (57.244792,589.854167) .. (57.244792,396.848958)
				.. controls (57.244792,203.848958) and (97.234375,47.390625) .. (146.5625,47.390625)
				.. controls (195.895833,47.390625) and (235.885417,203.848958) .. (235.885417,396.848958)
				.. controls (235.885417,589.854167) and (195.895833,746.3125) .. (146.5625,746.3125)
				-- cycle;
			\end{scope}
			\begin{scope}[scale=0.75]
				\path[fill=mintFill,draw=outlineInk,line width=0.25pt,line cap=butt,line join=miter,miter limit=1.207107,nonzero rule] (114.765625,123.947917)
				.. controls (114.765625,106.385417) and (129,92.145833) .. (146.5625,92.145833)
				.. controls (164.125,92.145833) and (178.364583,106.385417) .. (178.364583,123.947917)
				.. controls (178.364583,141.510417) and (164.125,155.75) .. (146.5625,155.75)
				.. controls (129,155.75) and (114.765625,141.510417) .. (114.765625,123.947917)
				-- cycle;
			\end{scope}
			\begin{scope}[scale=0.75]
				\path[fill=mintFill,draw=outlineInk,line width=0.25pt,line cap=butt,line join=miter,miter limit=1.207107,nonzero rule] (114.765625,232.505208)
				.. controls (114.765625,214.942708) and (129,200.708333) .. (146.5625,200.708333)
				.. controls (164.125,200.708333) and (178.364583,214.942708) .. (178.364583,232.505208)
				.. controls (178.364583,250.067708) and (164.125,264.307292) .. (146.5625,264.307292)
				.. controls (129,264.307292) and (114.765625,250.067708) .. (114.765625,232.505208)
				-- cycle;
			\end{scope}
			\begin{scope}[scale=0.75]
				\path[fill=mintFill,draw=outlineInk,line width=0.25pt,line cap=butt,line join=miter,miter limit=1.207107,nonzero rule] (114.765625,341.067708)
				.. controls (114.765625,323.505208) and (129,309.265625) .. (146.5625,309.265625)
				.. controls (164.125,309.265625) and (178.364583,323.505208) .. (178.364583,341.067708)
				.. controls (178.364583,358.630208) and (164.125,372.864583) .. (146.5625,372.864583)
				.. controls (129,372.864583) and (114.765625,358.630208) .. (114.765625,341.067708)
				-- cycle;
			\end{scope}
			\begin{scope}[scale=0.75]
				\path[fill=mintFill,draw=outlineInk,line width=0.25pt,line cap=butt,line join=miter,miter limit=1.207107,nonzero rule] (114.765625,449.625)
				.. controls (114.765625,432.0625) and (129,417.822917) .. (146.5625,417.822917)
				.. controls (164.125,417.822917) and (178.364583,432.0625) .. (178.364583,449.625)
				.. controls (178.364583,467.1875) and (164.125,481.421875) .. (146.5625,481.421875)
				.. controls (129,481.421875) and (114.765625,467.1875) .. (114.765625,449.625)
				-- cycle;
			\end{scope}
			\begin{scope}[scale=0.75]
				\path[fill=mintFill,draw=outlineInk,line width=0.25pt,line cap=butt,line join=miter,miter limit=1.207107,nonzero rule] (114.765625,558.182292)
				.. controls (114.765625,540.619792) and (129,526.385417) .. (146.5625,526.385417)
				.. controls (164.125,526.385417) and (178.364583,540.619792) .. (178.364583,558.182292)
				.. controls (178.364583,575.744792) and (164.125,589.984375) .. (146.5625,589.984375)
				.. controls (129,589.984375) and (114.765625,575.744792) .. (114.765625,558.182292)
				-- cycle;
			\end{scope}
			\begin{scope}[scale=0.75]
				\path[fill=mintFill,draw=outlineInk,line width=0.25pt,line cap=butt,line join=miter,miter limit=1.207107,nonzero rule] (114.765625,666.739583)
				.. controls (114.765625,649.177083) and (129,634.942708) .. (146.5625,634.942708)
				.. controls (164.125,634.942708) and (178.364583,649.177083) .. (178.364583,666.739583)
				.. controls (178.364583,684.302083) and (164.125,698.541667) .. (146.5625,698.541667)
				.. controls (129,698.541667) and (114.765625,684.302083) .. (114.765625,666.739583)
				-- cycle;
			\end{scope}
			\begin{scope}[scale=0.75]
				\path[fill=mintFill,draw=outlineInk,line width=0.25pt,line cap=butt,line join=miter,miter limit=1.207107,nonzero rule] (618.260417,191.114583)
				.. controls (618.260417,173.552083) and (632.494792,159.3125) .. (650.057292,159.3125)
				.. controls (667.619792,159.3125) and (681.859375,173.552083) .. (681.859375,191.114583)
				.. controls (681.859375,208.677083) and (667.619792,222.911458) .. (650.057292,222.911458)
				.. controls (632.494792,222.911458) and (618.260417,208.677083) .. (618.260417,191.114583)
				-- cycle;
			\end{scope}
			\begin{scope}[scale=0.75]
				\path[fill=mintFill,draw=outlineInk,line width=0.25pt,line cap=butt,line join=miter,miter limit=1.207107,nonzero rule] (655.385417,293.125)
				.. controls (655.385417,275.5625) and (669.625,261.328125) .. (687.1875,261.328125)
				.. controls (704.75,261.328125) and (718.989583,275.5625) .. (718.989583,293.125)
				.. controls (718.989583,310.6875) and (704.75,324.927083) .. (687.1875,324.927083)
				.. controls (669.625,324.927083) and (655.385417,310.6875) .. (655.385417,293.125)
				-- cycle;
			\end{scope}
			\begin{scope}[scale=0.75]
				\path[fill=orangeFill,draw=outlineInk,line width=0.25pt,line cap=butt,line join=miter,miter limit=1.207107,nonzero rule] (701.255208,418.713542)
				.. controls (701.255208,401.151042) and (715.489583,386.916667) .. (733.052083,386.916667)
				.. controls (750.614583,386.916667) and (764.854167,401.151042) .. (764.854167,418.713542)
				.. controls (764.854167,436.276042) and (750.614583,450.515625) .. (733.052083,450.515625)
				.. controls (715.489583,450.515625) and (701.255208,436.276042) .. (701.255208,418.713542)
				-- cycle;
			\end{scope}
			\begin{scope}[scale=0.75]
				\path[fill=mintFill,draw=outlineInk,line width=0.25pt,line cap=butt,line join=miter,miter limit=1.207107,nonzero rule] (769.15625,525.979167)
				.. controls (769.15625,508.416667) and (783.395833,494.177083) .. (800.958333,494.177083)
				.. controls (818.520833,494.177083) and (832.760417,508.416667) .. (832.760417,525.979167)
				.. controls (832.760417,543.541667) and (818.520833,557.78125) .. (800.958333,557.78125)
				.. controls (783.395833,557.78125) and (769.15625,543.541667) .. (769.15625,525.979167)
				-- cycle;
			\end{scope}
			\begin{scope}[scale=0.75]
				\path[fill=orangeFill,draw=outlineInk,line width=0.25pt,line cap=butt,line join=miter,miter limit=1.207107,nonzero rule] (824.260417,625.776042)
				.. controls (824.260417,608.213542) and (838.5,593.979167) .. (856.0625,593.979167)
				.. controls (873.625,593.979167) and (887.859375,608.213542) .. (887.859375,625.776042)
				.. controls (887.859375,643.338542) and (873.625,657.578125) .. (856.0625,657.578125)
				.. controls (838.5,657.578125) and (824.260417,643.338542) .. (824.260417,625.776042)
				-- cycle;
			\end{scope}
			\begin{scope}[scale=0.75]
				\path[fill=mintFill,draw=outlineInk,line width=0.25pt,line cap=butt,line join=miter,miter limit=1.207107,nonzero rule] (874.828125,524.291667)
				.. controls (874.828125,506.729167) and (889.0625,492.489583) .. (906.625,492.489583)
				.. controls (924.1875,492.489583) and (938.427083,506.729167) .. (938.427083,524.291667)
				.. controls (938.427083,541.854167) and (924.1875,556.088542) .. (906.625,556.088542)
				.. controls (889.0625,556.088542) and (874.828125,541.854167) .. (874.828125,524.291667)
				-- cycle;
			\end{scope}
			\begin{scope}[scale=0.75]
				\path[fill=orangeFill,draw=outlineInk,line width=0.25pt,line cap=butt,line join=miter,miter limit=1.207107,nonzero rule] (914.838542,413.651042)
				.. controls (914.838542,396.088542) and (929.072917,381.848958) .. (946.635417,381.848958)
				.. controls (964.203125,381.848958) and (978.4375,396.088542) .. (978.4375,413.651042)
				.. controls (978.4375,431.213542) and (964.203125,445.453125) .. (946.635417,445.453125)
				.. controls (929.072917,445.453125) and (914.838542,431.213542) .. (914.838542,413.651042)
				-- cycle;
			\end{scope}
			\begin{scope}[scale=0.75]
				\path[fill=mintFill,draw=outlineInk,line width=0.25pt,line cap=butt,line join=miter,miter limit=1.207107,nonzero rule] (951.989583,277.338542)
				.. controls (951.989583,259.776042) and (966.229167,245.536458) .. (983.791667,245.536458)
				.. controls (1001.354167,245.536458) and (1015.588542,259.776042) .. (1015.588542,277.338542)
				.. controls (1015.588542,294.901042) and (1001.354167,309.140625) .. (983.791667,309.140625)
				.. controls (966.229167,309.140625) and (951.989583,294.901042) .. (951.989583,277.338542)
				-- cycle;
			\end{scope}
			\begin{scope}[scale=0.75]
				\path[fill=mintFill,draw=outlineInk,line width=0.25pt,line cap=butt,line join=miter,miter limit=1.207107,nonzero rule] (977.380208,180.869792)
				.. controls (977.380208,163.307292) and (991.619792,149.067708) .. (1009.182292,149.067708)
				.. controls (1026.744792,149.067708) and (1040.979167,163.307292) .. (1040.979167,180.869792)
				.. controls (1040.979167,198.432292) and (1026.744792,212.671875) .. (1009.182292,212.671875)
				.. controls (991.619792,212.671875) and (977.380208,198.432292) .. (977.380208,180.869792)
				-- cycle;
			\end{scope}
			\begin{scope}[scale=0.75]
				\path[fill=none,draw=textInk,line width=0.25pt,line cap=butt,line join=miter,miter limit=1.207107,nonzero rule] (353.807292,404.666667)
				-- (527.692708,404.666667);
			\end{scope}
			\begin{scope}[scale=0.75]
				\path[fill=none,draw=textInk,line width=0.25pt,line cap=round,line join=round,miter limit=1.207107,nonzero rule] (521.614583,400.791667)
				-- (529.359375,404.666667)
				-- (521.614583,408.536458);
			\end{scope}
			\node[vertexlabel] at (109.922,92.961) {$v_1$};
			\node[vertexlabel] at (109.922,174.379) {$v_2$};
			\node[vertexlabel] at (109.922,255.801) {$v_3$};
			\node[vertexlabel] at (109.922,337.219) {$v_4$};
			\node[vertexlabel] at (109.922,418.637) {$v_5$};
			\node[vertexlabel] at (109.922,500.055) {$v_6$};
			\node[vertexlabel] at (487.543,143.336) {$v_1$};
			\node[vertexlabel] at (515.391,219.844) {$v_2$};
			\node[vertexlabel] at (549.789,314.035) {$x_1$};
			\node[vertexlabel] at (600.719,394.484) {$v_3$};
			\node[vertexlabel] at (642.047,469.332) {$x_2$};
			\node[vertexlabel] at (679.969,393.219) {$v_4$};
			\node[vertexlabel] at (709.977,310.238) {$x_3$};
			\node[vertexlabel] at (737.844,208.004) {$v_5$};
			\node[vertexlabel] at (756.887,135.652) {$v_6$};
		\end{tikzpicture}%
	\caption{Example of decomposing a hyperedge of size six. Numbers denote original vertices, and letters denote auxiliary vertices.}
	\label{fig:hyperedge_decomp}
\end{figure}

\begin{lemma}\label{lem:3.1}
	Let $\mathcal{H}=(V,\mathcal{E})$ be a hypergraph, and let $\mathcal{H}'=(V'=V \cup V_{\mathrm{aux}},\mathcal{E}')$ be the hypergraph obtained from $\mathcal{H}$ by applying the hyperedge decomposition procedure in Algorithm~\ref{alg:hyp}, where $V_{\mathrm{aux}}$ is the set of auxiliary vertices. Let $\mathcal{A}$ and $\mathcal{A}'$ denote the families of all even covers of $\mathcal{H}$ and $\mathcal{H}'$, respectively. Then the following statements hold.
	\begin{enumerate}
		\item[\textnormal{(i)}] \textbf{Chain constraint.} Suppose that an original hyperedge $e\in \mathcal{E}$ is decomposed into an ordered chain $\mathcal{C}(e)=(e^\prime_1,e^\prime_2,\ldots,e^\prime_t)\subseteq \mathcal{E}'$, where, for each $j=1,2,\ldots,t-1$, the two consecutive sub-hyperedges $e'_j$ and $e'_{j+1}$ intersect in the unique auxiliary vertex $x_j\in V_{\mathrm{aux}}$. Then, for every $Y'\in\mathcal{A}'$, if $e'_i\in Y'$ for some $i\in\{1,2,\ldots,t\}$, then $\{e^\prime_1,e^\prime_2,\ldots,e^\prime_t\}\subseteq Y'$.
		\item[\textnormal{(ii)}] \textbf{Bijection of even covers.} There exists a map $T:\mathcal{A}\to\mathcal{A}'$ such that $T$ is a bijection. More precisely, after regarding every undecomposed hyperedge $e$ as a chain $\mathcal{C}(e)=\{e^\prime\}$, the map
		\[
		T(Y)=\bigcup_{e\in Y}\mathcal{C}(e), \qquad Y\in\mathcal{A},
		\]
		is a bijection from $\mathcal{A}$ onto $\mathcal{A}'$.
	\end{enumerate}
\end{lemma}

\begin{proof}
	All parity computations below are taken over $\mathbb{F}_2$. For a subset $Z$ and an element $a$, write $\boldsymbol{1}_Z(a)$ for the indicator of the event $a\in Z$. Recall that $\mathcal{C}(e)$ denotes the chain produced from $e$ by Algorithm~\ref{alg:hyp}. For any hyperedge $e$ with $|e|\leq 3$, Algorithm~\ref{alg:hyp} leaves it unchanged, so $\mathcal{C}(e)=\{e\}$.

	The chains $\mathcal{C}(e)$, $e\in\mathcal{E}$, are pairwise edge-disjoint and their union equals $\mathcal{E}'$; in particular, every hyperedge in $\mathcal{E}'$ arises from exactly one original hyperedge $e\in\mathcal{E}$.

\noindent\textnormal{(i) Chain constraint.} For a hyperedge with $|e|\leq 3$ we have $\mathcal{C}(e)=\{e\}$, and the statement is immediate. It remains to consider a hyperedge $e$ with $|e|>3$. Write its decomposition chain as $\mathcal{C}(e)=(e^\prime_1,\ldots,e^\prime_t)$. For each $j=1,\ldots,t-1$, the consecutive hyperedges $e^\prime_j$ and $e^\prime_{j+1}$ share an auxiliary vertex $x_j$. By construction, $x_j$ is incident, among the hyperedges of $\mathcal{E}'$, exactly with $e^\prime_j$ and $e^\prime_{j+1}$. If $Y'\in\mathcal{A}'$, the even-cover condition at $x_j$ gives
	\[
	\boldsymbol{1}_{Y'}(e^\prime_j)+\boldsymbol{1}_{Y'}(e^\prime_{j+1})\equiv 0 \pmod{2}.
	\]
	Hence $\boldsymbol{1}_{Y'}(e^\prime_j)=\boldsymbol{1}_{Y'}(e^\prime_{j+1})$ for every $j=1,\ldots,t-1$. Induction along the chain shows that all hyperedges in $\mathcal{C}(e)$ have the same membership status in $Y'$. Therefore, once one of them belongs to $Y'$, all of them belong to $Y'$, proving~(i).

	\noindent\textnormal{(ii) Bijection of even covers.} We define the forward map $T:\mathcal{A}\to 2^{\mathcal{E}'}$ by
	\[
	T(Y)=\bigcup_{e\in Y}\mathcal{C}(e).
	\]
	We verify that $T(Y)\in\mathcal{A}'$ whenever $Y\in\mathcal{A}$. Fix any original vertex $v\in V$. For every original hyperedge $e\in\mathcal{E}$, Algorithm~\ref{alg:hyp} constructs $\mathcal{C}(e)$ so that exactly one hyperedge in $\mathcal{C}(e)$ contains $v$ if $v\in e$, and no hyperedge in $\mathcal{C}(e)$ contains $v$ if $v\notin e$. Hence the number of hyperedges in $\mathcal{C}(e)$ that contain $v$ equals $\boldsymbol{1}_e(v)$ in both cases. Consequently,
	\[
	\deg_{T(Y)}(v)\equiv \sum_{e\in Y}\boldsymbol{1}_{e}(v)=\deg_Y(v)\equiv 0 \pmod{2}.
	\]
	If $x_j\in V_{\mathrm{aux}}$ is introduced inside the chain of some hyperedge $e$, then $x_j$ is incident only with the adjacent pair $e^\prime_j,e^\prime_{j+1}$. The definition of $T$ selects either both of these hyperedges (when $e\in Y$) or neither (when $e\notin Y$). Hence $\deg_{T(Y)}(x_j)\in\{0,2\}$ is even. Thus $T(Y)\in\mathcal{A}'$.

	Define the reverse map $P:\mathcal{A}'\to 2^{\mathcal{E}}$ by
	\[
P(Y')=\{e\in\mathcal{E}:\mathcal{C}(e)\subseteq Y'\}.
	\]
	By the chain constraint, for every $Y'\in\mathcal{A}'$ and every $e\in\mathcal{E}$, either the whole chain $\mathcal{C}(e)$ is contained in $Y'$ or it is disjoint from $Y'$. Therefore $P$ is well defined and captures exactly the selected chains of $Y'$. To see that $P(Y')\in\mathcal{A}$, fix $v\in V$. Since each original occurrence of $v$ in $e$ corresponds to exactly one occurrence of $v$ in $\mathcal{C}(e)$, and each chain is either fully selected or not selected at all,
	\[
	\deg_{P(Y')}(v)\equiv \deg_{Y'}(v)\equiv 0 \pmod{2}.
	\]
	Thus $P(Y')\in\mathcal{A}$.

	Finally, $P(T(Y))=Y$ for every $Y\in\mathcal{A}$ by the definition of $T$. Conversely, for any $Y'\in\mathcal{A}'$, the chain constraint and the edge-disjoint decomposition $\mathcal{E}'=\bigcup_{e\in\mathcal{E}}\mathcal{C}(e)$ together imply that $Y'$ is precisely the union of those chains $\mathcal{C}(e)$ with $e\in P(Y')$; hence $T(P(Y'))=Y'$. Thus $P=T^{-1}$ on $\mathcal{A}'$, and $T:\mathcal{A}\to\mathcal{A}'$ is a bijection. This proves~(ii) and completes the proof.
\end{proof}

\subsection{Check Node Splitting}

Having introduced Algorithm~\ref{alg:hyp} for reducing the cardinality of every hyperedge to at most three, we now develop an analogous transformation for controlling check-node degrees in Tanner graphs. A general Tanner graph may contain check nodes of arbitrarily high degree, whereas the subsequent reductions require check-node degrees to be bounded before the desired regularity conditions can be imposed. Algorithm~\ref{alg:chec} therefore replaces each check node of degree greater than three with a chain of degree-three check nodes connected through auxiliary variable nodes. This construction is the Tanner-graph counterpart of the chain decomposition in Algorithm~\ref{alg:hyp} and preserves the original parity constraint through a collection of local parity constraints.

\begin{algorithm}[htbp]
	\caption{Check Node Splitting for Degree Reduction to Three}
	\label{alg:chec}
	\begin{algorithmic}[1]
		\Require Bipartite graph $\mathcal{G} = (L, R, E)$, where $L$ and $R$ denote variable and check nodes, respectively
		\Ensure A bipartite graph $\mathcal{G} = (L\cup L_{\mathrm{aux}}, R', E')$ with $\deg(r) \leq 3$ for all $r \in R'$
		\State Initialize $L_{\mathrm{aux}} \gets \emptyset$, $R' \gets \emptyset$, $E' \gets \emptyset$
		\For{each check node $r \in R$}
		\State Let $k \gets \deg(r)$
		\LongState{Enumerate neighbors of $r$ as $u_1, u_2, \dots, u_k$, i.e., $N(r) = \{u_1, \dots, u_k\}$}
		\If{$k > 3$}
		\State Introduce auxiliary variable nodes $y_1, y_2, \dots, y_{k-3}$
		\LongState{Construct split check nodes $r_1, \dots, r_{k-2}$ with updated neighborhoods:}
		\State $N(r_1) \gets \{u_1, u_2, y_1\}$
		\State $N(r_{k-2}) \gets \{y_{k-3}, u_{k-1}, u_k\}$
		\For{$j \gets 2$ \textbf{to} $k-3$}
		\State $N(r_j) \gets \{y_{j-1}, u_{j+1}, y_j\}$
		\EndFor
		\State Update graph components:
		\State $R' \gets R' \cup \{r_1, \dots, r_{k-2}\}$
		\State $L_{\mathrm{aux}} \gets L_{\mathrm{aux}} \cup \{y_1, \dots, y_{k-3}\}$
		\State $E' \gets E'\cup \left\{\{u, r_j\} \mid 1\leq j\leq k-2,\ u \in N(r_j)\right\}$
        \Else
        \State $R' \gets R'\cup \{r\}$
        \State $E' \gets E' \cup \left\{\{u, r\} \mid u \in N(r)\right\}$
		\EndIf
		\EndFor
		\State \textbf{return} $(L\cup L_{\mathrm{aux}}, R', E')$
	\end{algorithmic}
\end{algorithm}

Fig.~\ref{fig:split} illustrates the check node splitting operation at the Tanner-graph level for a check node $r$ of degree five. The original node $r$, adjacent to $u_1^r,\ldots,u_5^r$, is replaced by a chain of three check nodes $r_1^r,r_2^r,r_3^r$. The auxiliary variables $y_1^r$ and $y_2^r$ connect consecutive check nodes in the chain: each original variable is adjacent to exactly one split check node, whereas each auxiliary variable is adjacent to exactly two consecutive split check nodes. Thus, the local even-parity conditions along the chain reproduce the parity condition imposed by the original check node. Computationally, Algorithm~\ref{alg:chec} iterates through each check node $r \in R$. For any check node with $\deg(r) > 3$, it introduces $\deg(r)-3$ auxiliary variable nodes and constructs $\deg(r)-2$ new split check nodes, completing the local update in $\mathcal{O}(\deg(r))$ operations. The total running time is therefore $\mathcal{O}\left( \sum_{r \in R} \deg(r) \right) = \mathcal{O}(|E|)$, where $|E|$ is the number of edges in the initial bipartite graph. The mapping between the $(a,0)$-trapping sets of the original and split graphs is formalized in Lemma~\ref{lem:3.2}.

\begin{figure*}[h]
	\centering
	\begin{tikzpicture}[
		>=stealth,
		scale=0.95,
		every node/.style={font=\small},
		var/.style={
			circle,
			draw=black,
			fill=blue!12,
			minimum size=0.72cm,
			inner sep=0pt,
			line width=0.8pt
		},
		aux/.style={
			circle,
			draw=black,
			fill=orange!18,
			minimum size=0.72cm,
			inner sep=0pt,
			line width=0.8pt
		},
		chk/.style={
			rectangle,
			rounded corners=2pt,
			draw=black,
			fill=green!15,
			minimum size=0.72cm,
			inner sep=0pt,
			line width=0.8pt
		},
		edge/.style={draw=black, line width=0.9pt}
	]

		% ================= Left: original check node =================
		\node[var] (u1) at (-2.4,1.6) {$u_1^r$};
		\node[var] (u2) at (-1.2,1.6) {$u_2^r$};
		\node[var] (u3) at ( 0.0,1.6) {$u_3^r$};
		\node[var] (u4) at ( 1.2,1.6) {$u_4^r$};
		\node[var] (u5) at ( 2.4,1.6) {$u_5^r$};

		\node[chk] (r) at (0,0) {$r$};

		\foreach \x in {u1,u2,u3,u4,u5}{
			\draw[edge] (\x)--(r);
		}

		% label
		\node[font=\normalsize] at (0,2.55) {Original};

		% ================= Arrow =================
		\node[font=\normalsize] at (4.8,0.8) {Split};
		\draw[->, line width=1pt] (3.7,0.45) -- (5.9,0.45);

		% ================= Right: split chain =================
		% top row: original + auxiliary variables
		\node[var] (v1) at (7.0,1.8) {$u_1^r$};
		\node[var] (v2) at (8.2,1.8) {$u_2^r$};
		\node[aux] (y1) at (9.4,1.8) {$y_1^r$};
		\node[var] (v3) at (10.6,1.8) {$u_3^r$};
		\node[aux] (y2) at (11.8,1.8) {$y_2^r$};
		\node[var] (v4) at (13.0,1.8) {$u_4^r$};
		\node[var] (v5) at (14.2,1.8) {$u_5^r$};

		% bottom row: split checks
		\node[chk] (r1) at (8.2,0) {$r_1^r$};
		\node[chk] (r2) at (10.6,0) {$r_2^r$};
		\node[chk] (r3) at (13.0,0) {$r_3^r$};

		% edges
		\draw[edge] (v1)--(r1);
		\draw[edge] (v2)--(r1);
		\draw[edge] (y1)--(r1);

		\draw[edge] (y1)--(r2);
		\draw[edge] (v3)--(r2);
		\draw[edge] (y2)--(r2);

		\draw[edge] (y2)--(r3);
		\draw[edge] (v4)--(r3);
		\draw[edge] (v5)--(r3);

		% label
		\node[font=\normalsize] at (10.6,2.75) {After splitting};

	\end{tikzpicture}
	\caption{Check node splitting operation for a check node $r$ of degree $5$. The original neighbors
	$u_1^r,\ldots,u_5^r$ are replaced by the split check nodes
	$r_1^r,r_2^r,r_3^r$ and auxiliary variables $y_1^r,y_2^r$ according to Algorithm~\ref{alg:chec}.}
	\label{fig:split}
\end{figure*}
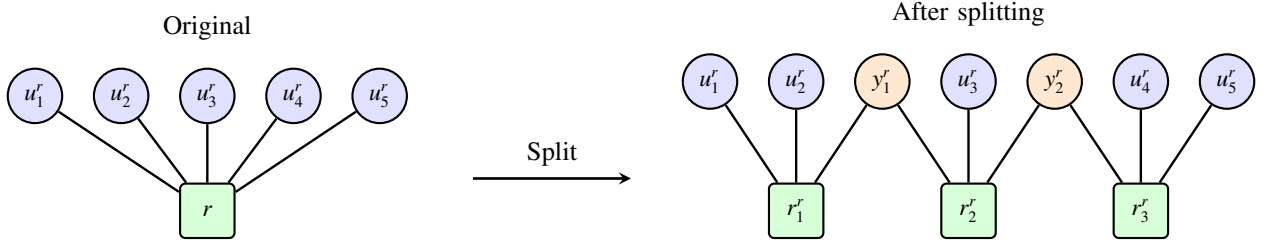

\begin{lemma}\label{lem:3.2}
	Let $\mathcal{G}=(L,R,E)$ be a bipartite graph. Let
	$\mathcal{G}'=(L'=L\cup L_{\mathrm{aux}},R',E')$ be the
	bipartite graph obtained from $\mathcal{G}$ by applying the check
	node splitting procedure in Algorithm~\ref{alg:chec}, where
	$L_{\mathrm{aux}}$ is the set of auxiliary variable nodes,
	and $R'$ {\color{blue}is} the resulting set of split check nodes.
	Let $\mathcal{A}$ and $\mathcal{A}'$ denote the families of all
	$(a,0)$-TSs of $\mathcal{G}$ and $\mathcal{G}'$, respectively.
	Then the following statements hold.
	\begin{enumerate}

		\item[\textnormal{(i)}]
		\textbf{Nonempty projection onto the original variables.}
		For every $S'\in\mathcal{A}'$, one has
		\[
			S'\cap L\neq\varnothing.
		\]

		\item[\textnormal{(ii)}]
		\textbf{Bijection of $(a,0)$-TSs.}
		For $S\in\mathcal{A}$, define
		\[
		Q(S)=S\cup
		\left\{
			y_j^r\ \middle|\
			\begin{array}{l}
				k_r=\deg_{\mathcal G}(r)>3,\ 1\leq j\leq k_r-3\\
				r\in R,\
				\displaystyle
				\sum_{i=1}^{j+1}\boldsymbol{1}_S(u_i^r)
				\equiv 1\pmod 2
			\end{array}
		\right\}.
		\]
		Then $Q:\mathcal{A}\to\mathcal{A}'$ is a bijection, whose
		inverse is the projection
		\[
			Q^{-1}(S')=S'\cap L.
		\]
	\end{enumerate}
\end{lemma}

\begin{proof}
	In the parity formulation of $(a,0)$-TSs adopted in this paper,
	a subset $S$ of variable nodes belongs to the $(a,0)$-TS family
	if and only if $S$ is nonempty and every check node is adjacent
	to an even number of variable nodes in $S$.

	For each check node $r\in R$ with
	$k_r=\deg_{\mathcal G}(r)>3$, let its neighbors be ordered as
	\[
		N_{\mathcal G}(r)
		=\{u_1^r,u_2^r,\ldots,u_{k_r}^r\}.
	\]
	Algorithm~\ref{alg:chec} splits $r$ into check nodes
	$r_1^r,\ldots,r_{k_r-2}^r$ and introduces auxiliary variable
	nodes $y_1^r,\ldots,y_{k_r-3}^r$, where
	\[
		N_{\mathcal G'}(r_1^r)
		=\{u_1^r,u_2^r,y_1^r\},
	\]
	\[
		N_{\mathcal G'}(r_j^r)
		=\{y_{j-1}^r,u_{j+1}^r,y_j^r\},
		\qquad 2\leq j\leq k_r-3,
	\]
	and
	\[
		N_{\mathcal G'}(r_{k_r-2}^r)
		=\{y_{k_r-3}^r,u_{k_r-1}^r,u_{k_r}^r\}.
	\]

	\noindent\textnormal{(i)} Let $S'\in\mathcal A'$. Suppose, for contradiction, that
	\[
		S'\cap L=\varnothing.
	\]
	Consider any split chain corresponding to an original check node
	$r$ with $k_r>3$. Since neither $u_1^r$ nor $u_2^r$ belongs to
	$S'$, the even-neighbor condition at the first split check node
	$r_1^r$ gives
	\[
		\boldsymbol{1}_{S'}(y_1^r)
		\equiv
		\boldsymbol{1}_{S'}(u_1^r)
		+\boldsymbol{1}_{S'}(u_2^r)
		\equiv 0\pmod 2.
	\]
	Thus $y_1^r\notin S'$. Next, for every
	$2\leq j\leq k_r-3$, the parity condition at $r_j^r$ gives
	\[
		\boldsymbol{1}_{S'}(y_j^r)
		\equiv
		\boldsymbol{1}_{S'}(y_{j-1}^r)
		+\boldsymbol{1}_{S'}(u_{j+1}^r)
		\pmod 2.
	\]
	Because $u_{j+1}^r\notin S'$, induction on $j$ shows that
	$y_j^r\notin S'$ for every $1\leq j\leq k_r-3$. This argument
	also covers the case $k_r=4$, in which the only auxiliary
	variable $y_1^r$ is already forced to be absent by the first
	split check node.

	Every auxiliary variable belongs to one of these split chains.
	Hence no auxiliary variable belongs to $S'$. Together with
	$S'\cap L=\varnothing$, this implies $S'=\varnothing$, contradicting
	the convention that an $(a,0)$-TS is nonempty. Therefore
	\[
		S'\cap L\neq\varnothing.
	\]
	This proves statement~\textnormal{(i)}.

	\medskip
	\noindent\textnormal{(ii)} For $S\in\mathcal{A}$, define
	\begin{align*}
	Q(S)=S\cup\bigl\{y_j^r \mid {}&
	  \sum_{i=1}^{j+1}\boldsymbol{1}_S(u_i^r)
	  \equiv 1 \pmod{2},\\
	& 1\leq j\leq k_r-3,\quad r\in R,\quad k_r>3
	\bigr\}.
	\end{align*}
	Thus $Q(S)$ retains all original variable nodes of $S$, so that
	\[
		Q(S)\cap L=S,
	\]
	and each auxiliary node $y_j^r$ is included if and only if the
	prefix sum
	$\sum_{i=1}^{j+1}\boldsymbol{1}_S(u_i^r)$ is odd. Equivalently,
	\[
	\boldsymbol{1}_{Q(S)}(y_j^r)
	\equiv
	\sum_{i=1}^{j+1}\boldsymbol{1}_S(u_i^r)
	\pmod 2.
	\]

	We first prove that $Q(S)\in\mathcal{A}'$. Since
	$S\in\mathcal{A}$, every check node in $\mathcal{G}$ is adjacent
	to an even number of variable nodes in $S$. For every unsplit
	check node, its neighborhood in $\mathcal{G}'$ is identical to
	that in $\mathcal{G}$. Since $Q(S)\cap L=S$, such a check node
	remains adjacent to an even number of selected variables in
	$Q(S)$.

	Consider a split chain associated with some $r\in R$ satisfying
	$\deg_{\mathcal G}(r)>3$. For the first split check node $r_1^r$,
	the number of adjacent variables belonging to $Q(S)$ has parity
	\begin{align*}
	&\boldsymbol{1}_{Q(S)}(u_1^r)
	+\boldsymbol{1}_{Q(S)}(u_2^r)
	+\boldsymbol{1}_{Q(S)}(y_1^r)\\
	\equiv{}&
	\boldsymbol{1}_S(u_1^r)
	+\boldsymbol{1}_S(u_2^r)
	+\sum_{i=1}^{2}\boldsymbol{1}_S(u_i^r)\\
	\equiv{}&0\pmod 2.
	\end{align*}

	For an intermediate split check node $r_j^r$,
	$2\leq j\leq k_r-3$, the corresponding parity is
	\begin{align*}
	&\boldsymbol{1}_{Q(S)}(y_{j-1}^r)
	+\boldsymbol{1}_{Q(S)}(u_{j+1}^r)
	+\boldsymbol{1}_{Q(S)}(y_j^r)\\
	\equiv{}&
	\sum_{i=1}^{j}\boldsymbol{1}_S(u_i^r)
	+\boldsymbol{1}_S(u_{j+1}^r)
	+\sum_{i=1}^{j+1}\boldsymbol{1}_S(u_i^r)\\
	\equiv{}&0\pmod 2.
	\end{align*}

	For the last split check node $r_{k_r-2}^r$, the corresponding
	parity is
	\begin{align*}
	&\boldsymbol{1}_{Q(S)}(y_{k_r-3}^r)
	+\boldsymbol{1}_{Q(S)}(u_{k_r-1}^r)
	+\boldsymbol{1}_{Q(S)}(u_{k_r}^r)\\
	\equiv{}&
	\sum_{i=1}^{k_r-2}\boldsymbol{1}_S(u_i^r)
	+\boldsymbol{1}_S(u_{k_r-1}^r)
	+\boldsymbol{1}_S(u_{k_r}^r)\\
	\equiv{}&
	\sum_{i=1}^{k_r}\boldsymbol{1}_S(u_i^r)\\
	\equiv{}&0\pmod 2,
	\end{align*}
	where the last congruence holds because $S\in\mathcal A$
	satisfies the parity condition at the original check node $r$.
	Hence every check node in $\mathcal G'$ has even intersection
	with $Q(S)$, and therefore $Q(S)\in\mathcal A'$.
	Thus $Q$ is a well-defined map from
	$\mathcal A$ to $\mathcal A'$.

	Define
	\[
		P:\mathcal A'\longrightarrow \mathcal A,
		\qquad
		P(S')=S'\cap L.
	\]
	By statement~\textnormal{(i)}, $P(S')$ is nonempty for every
	$S'\in\mathcal A'$.
	We now show that $P(S')\in\mathcal A$. For an unsplit check node,
	this follows immediately because its neighborhood contains only
	original variable nodes and is unchanged by the transformation.

	For an original check node $r$ that is split, the parity
	constraints in $\mathcal G'$ are
	\begin{align*}
		\boldsymbol{1}_{S'}(u_1^r)
		+\boldsymbol{1}_{S'}(u_2^r)
		+\boldsymbol{1}_{S'}(y_1^r)
		&\equiv 0\pmod 2,\\
		\boldsymbol{1}_{S'}(y_{j-1}^r)
		+\boldsymbol{1}_{S'}(u_{j+1}^r)
		+\boldsymbol{1}_{S'}(y_j^r)
		&\equiv 0\pmod 2,
		&&2\leq j\leq k_r-3,\\
		\boldsymbol{1}_{S'}(y_{k_r-3}^r)
		+\boldsymbol{1}_{S'}(u_{k_r-1}^r)
		+\boldsymbol{1}_{S'}(u_{k_r}^r)
		&\equiv 0\pmod 2.
	\end{align*}
	The middle family is empty when $k_r=4$. Adding all applicable
	congruences cancels every auxiliary variable $y_j^r$, because
	each appears exactly twice, and yields
	\[
		\sum_{i=1}^{k_r}\boldsymbol{1}_{S'}(u_i^r)
		\equiv 0\pmod 2.
	\]
	Since $u_i^r\in L$, this is exactly the parity condition for
	$P(S')$ at the original check node $r$. Hence
	$P(S')\in\mathcal A$.

	Clearly,
	\[
		P(Q(S))=S
	\]
	for every $S\in\mathcal A$. Conversely, let
	$S'\in\mathcal A'$. The first congruence in the split chain gives
	\[
	\boldsymbol{1}_{S'}(y_1^r)
	\equiv
	\boldsymbol{1}_{S'}(u_1^r)
	+\boldsymbol{1}_{S'}(u_2^r)
	\pmod 2,
	\]
	and the $j$th intermediate congruence gives recursively
	\[
	\boldsymbol{1}_{S'}(y_j^r)
	\equiv
	\boldsymbol{1}_{S'}(y_{j-1}^r)
	+\boldsymbol{1}_{S'}(u_{j+1}^r)
	\pmod 2,
	\qquad 2\leq j\leq k_r-3.
	\]
	Therefore
	\[
	\boldsymbol{1}_{S'}(y_j^r)
	\equiv
	\sum_{i=1}^{j+1}\boldsymbol{1}_{S'}(u_i^r)
	\overset{(a)}{=}
	\sum_{i=1}^{j+1}\boldsymbol{1}_{P(S')}(u_i^r)
	\equiv
	\boldsymbol{1}_{Q(P(S'))}(y_j^r)
	\pmod 2.
	\]

	Equality~$(a)$ holds because $u_i^r\in L$ and
	$P(S')=S'\cap L$, so
	\[
		\boldsymbol{1}_{S'}(u_i^r)
		=
		\boldsymbol{1}_{P(S')}(u_i^r)
	\]
	for every $i$.
	Thus every auxiliary variable has the same membership status in
	$Q(P(S'))$ as in $S'$, while the original variables agree by
	the definition of $P$. Hence
	\[
		Q(P(S'))=S'.
	\]
	Consequently, $P=Q^{-1}$ on $\mathcal A'$, and $Q$ is a
	bijection. This proves statement~\textnormal{(ii)}.
\end{proof}

\subsection{Variable Replication}

% 这里还要再润色一下，第一句话可以不用改，后面需要说一下前两个算法的问题：都会产生变量数增加的现象，并且这会导致两个原来的非零码字的相对大小与转换后这两个的相对大小不一致。为了避免这种情况的发生我们才引入这个算法，对应于每个原始变量，我们复制足够多份来消除多余变量的影响。
The preceding two transformations control hyperedge cardinalities and check node degrees. However, both transformations increase the number of variable-side objects: hyperedge decomposition replaces a single hyperedge by several component hyperedges, while check node splitting introduces auxiliary variable nodes. Although the corresponding parity structures are preserved bijectively, Hamming weights are not necessarily preserved. In particular, because the number of additional selected variables may depend on the structure of a codeword, two nonzero codewords can have one weight ordering before the transformation and a different ordering after it. This creates a difficulty for minimum-distance reductions, where the relative weights of feasible nonzero codewords must be controlled.

To overcome this difficulty, we use variable replication as a uniform weight-amplification mechanism. For each original variable, we create a sufficiently large, common even number of replication variables and force all of them to have exactly the same membership status as the original variable in every $(a,0)$-TS. Consequently, selecting an original variable contributes the same large baseline amount to the transformed weight, whereas the nonuniform contribution of the auxiliary variables introduced by the preceding transformations remains bounded. By choosing the replication count sufficiently large, this auxiliary contribution can no longer change the weight ordering induced by the original variables.

Algorithm~\ref{alg:expan} implements this mechanism locally while preserving $3$-left regularity. Given a target variable node $v$, it reroutes two edges incident with $v$ and connects the new variables $x_1,x_2,\ldots,x_m$ through auxiliary degree-two check nodes. The resulting parity constraints force $v,x_1,x_2,\ldots,x_m$ to be either all selected or all unselected in every $(a,0)$-TS. Lemma~\ref{lem:3.3} formalizes this equivalent-variable property and establishes the induced bijection between the $(a,0)$-TS families before and after replication.

\begin{algorithm}[h]
	\caption{Variable Replication with Degree-Bounded Check Node Insertion}
	\label{alg:expan}
	\begin{algorithmic}[1]
		\Require 3-left regular bipartite graph $\mathcal{G} = (L, R, E)$, target variable $v \in L$, replication count $m = 2k$ ($k \in \mathbb{Z}^+$)
		\Ensure A 3-left regular bipartite graph $\mathcal{G} = (L\cup\{x_1,x_2,\ldots,x_m\}, 
        R\cup \{y_1,\ldots, y_{m+1}, y_{3,4}, y_{5,6} \ldots, y_{m-1,m}\}, E')$ containing $m$ mutually equivalent variables
		\State Initialize $E' \gets E$
		\State Identify neighbors of $v$: $N(v) = \{r_1, r_2, r_3\}$
		\State Introduce auxiliary variables $x_1, x_2$
		\State Update edges: $E' \gets E' \setminus \{\{v, r_1\}, \{v, r_2\}\} \cup \{\{x_1, r_1\}, \{x_2, r_2\}\}$
		\State Introduce auxiliary check nodes $y_1, y_2$
		\State Update edges: $E' \gets E' \cup \{\{x_1, y_1\}, \{x_2, y_2\}, \{v, y_1\}, \{v, y_2\}\}$
		\For{$i \gets 4,6,\ldots,m$}
		\LongState{Introduce variables $x_{i-1}, x_i$ and check nodes $y_{i-1}, y_i, y_{i-1,i} \in R'$}
		\LongState{Update edges: $E' \gets E' \cup \left\{
        \begin{array}{l}
        (x_{i-3}, y_{i-1}), (x_{i-1}, y_{i-1}),\\
        (x_{i-2}, y_i), (x_i, y_i),\\
        (x_{i-1}, y_{i-1,i}), (x_i, y_{i-1,i})
        \end{array}
        \right\}$}
		\EndFor
		\State Introduce terminal check nodes $y_{m+1}$
		\State Update edges: $E' \gets E' \cup \{\{x_{m-1}, y_{m+1}\}, \{x_m, y_{m+1}\}\}$
		\State \textbf{return} $(L\cup\{x_1,x_2,\ldots,x_m\}, R\cup \{y_1,\ldots, y_{m+1}, y_{3,4}, y_{5,6} \ldots, y_{m-1,m}\}, E')$
	\end{algorithmic}
\end{algorithm}

Fig.~\ref{fig:variable_expand} depicts the resulting topological structure. The process introduces auxiliary check nodes that bridge pairs of replication variables, creating an extended local network stemming from the target node $v$. The procedure is designed to execute in a single pass. For a requested replication factor $m$, the loop iterates $\mathcal{O}(m)$ times, performing a constant number of edge, variable, and check node insertions at each step. Consequently, the running time is $\mathcal{O}(m)$, linear in the number of clones generated. The corresponding correctness properties, demonstrating that these replicated variables are deterministically bound to the state of the original node under $(a,0)$-TS conditions, are verified in Lemma~\ref{lem:3.3}.

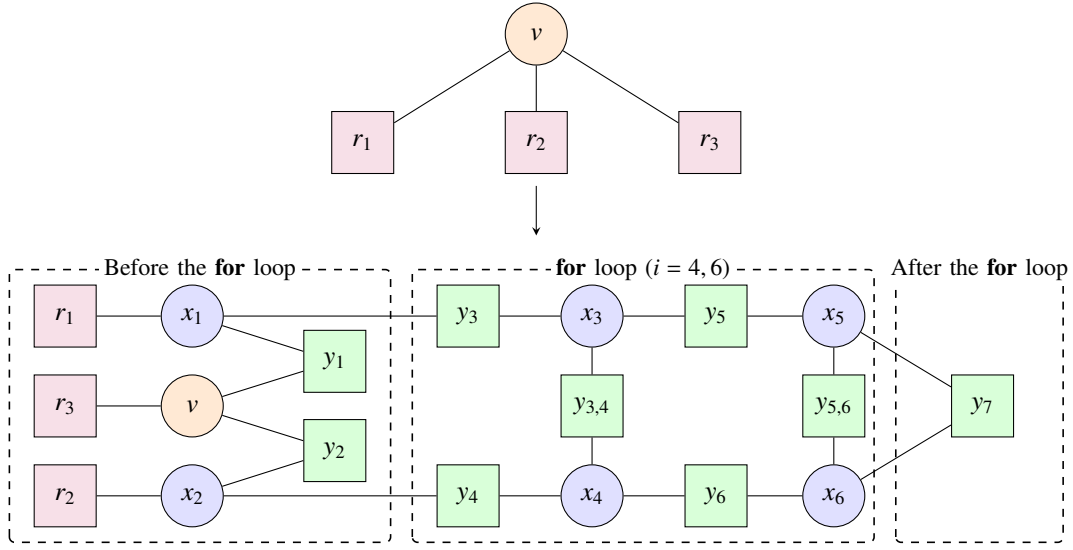
\begin{figure*}[h]
	\centering
	\begin{tikzpicture}[
		>=stealth,
		xnode/.style={
			circle,
			draw=black,
			fill=blue!12,
			minimum size=0.82cm,
			inner sep=0pt
		},
		vnode/.style={
			circle,
			draw=black,
			fill=orange!17,
			minimum size=0.82cm,
			inner sep=0pt
		},
		ynode/.style={
			rectangle,
			draw=black,
			fill=green!15,
			minimum size=0.82cm,
			inner sep=0pt
		},
		rnode/.style={
			rectangle,
			draw=black,
			fill=purple!12,
			minimum size=0.82cm,
			inner sep=0pt
		},
		group/.style={
			draw=black,
			dashed,
			rounded corners=2pt,
			line width=0.6pt
		},
		scale=0.82
	]
		
		\node[vnode] (v-top) at (-1.65,6.0) {$v$};
		\node[rnode] (r1-top) at (-4.45,4.25) {$r_1$};
		\node[rnode] (r2-top) at (-1.65,4.25) {$r_2$};
		\node[rnode] (r3-top) at ( 1.15,4.25) {$r_3$};
		
		\draw (v-top) -- (r1-top);
		\draw (v-top) -- (r2-top);
		\draw (v-top) -- (r3-top);
		
		\draw[->] (-1.65,3.55) -- (-1.65,2.75);
		
		\draw[group] (-10.15,-2.20) rectangle (-4.00,2.20);
		\draw[group] ( -3.65,-2.20) rectangle ( 3.80,2.20);
		\draw[group] (  4.15,-2.20) rectangle ( 6.85,2.20);
		
		\node[fill=white,inner sep=2pt,font=\small]
			at (-7.075,2.20)
			{Before the \textbf{for} loop};
		
		\node[fill=white,inner sep=2pt,font=\small]
			at (0.075,2.20)
			{\textbf{for} loop ($i=4,6$)};
		
		\node[fill=white,inner sep=2pt,font=\small]
			at (5.50,2.20)
			{After the \textbf{for} loop};
		
		\node[rnode] (r1) at (-9.25, 1.45) {$r_1$};
		\node[rnode] (r3) at (-9.25, 0.00) {$r_3$};
		\node[rnode] (r2) at (-9.25,-1.45) {$r_2$};
		
		\node[xnode] (x1) at (-7.20, 1.45) {$x_1$};
		\node[vnode] (v)  at (-7.20, 0.00) {$v$};
		\node[xnode] (x2) at (-7.20,-1.45) {$x_2$};
		
		\node[ynode] (y1) at (-4.90, 0.72) {$y_1$};
		\node[ynode] (y2) at (-4.90,-0.72) {$y_2$};
		
		\draw (r1) -- (x1);
		\draw (r3) -- (v);
		\draw (r2) -- (x2);
		
		\draw (x1) -- (y1);
		\draw (v)  -- (y1);
		\draw (v)  -- (y2);
		\draw (x2) -- (y2);
		
		\node[ynode] (y3)  at (-2.75, 1.45) {$y_3$};
		\node[ynode] (y4)  at (-2.75,-1.45) {$y_4$};
		\node[xnode] (x3)  at (-0.75, 1.45) {$x_3$};
		\node[xnode] (x4)  at (-0.75,-1.45) {$x_4$};
		\node[ynode] (y34) at (-0.75, 0.00) {$y_{3,4}$};
		
		\draw (x1) -- (y3);
		\draw (y3) -- (x3);
		\draw (x2) -- (y4);
		\draw (y4) -- (x4);
		\draw (x3) -- (y34);
		\draw (y34) -- (x4);

		\node[ynode] (y5) at (1.25, 1.45) {$y_5$};
		\node[ynode] (y6) at (1.25,-1.45) {$y_6$};
		\node[xnode] (x5) at (3.15, 1.45) {$x_5$};
		\node[xnode] (x6) at (3.15,-1.45) {$x_6$};
		\node[ynode] (y56) at (3.15, 0.00) {$y_{5,6}$};
		
		\draw (x3) -- (y5);
		\draw (y5) -- (x5);
		\draw (x4) -- (y6);
		\draw (y6) -- (x6);
		\draw (x5) -- (y56);
		\draw (y56) -- (x6);
		
		\node[ynode] (y7) at (5.55,0) {$y_7$};
		
		\draw (x5) -- (y7);
		\draw (x6) -- (y7);
		
	\end{tikzpicture}
	
	\caption{Example of applying Algorithm~\ref{alg:expan} to a target
		variable node $v$ with replication count $m=6$. Circles represent
		variable nodes, squares represent check nodes, and the three dashed
		boxes contain the operations before, during, and after the
		\textbf{for} loop, respectively.}
	\label{fig:variable_expand}
\end{figure*}

\begin{lemma}\label{lem:3.3}
	Let $\mathcal{G}=(L,R,E)$ be a $3$-left regular bipartite graph. Fix $v\in L$, and apply the variable replication procedure in Algorithm~\ref{alg:expan} to $v$ with an even integer $m\ge 2$, producing replication variable nodes $x_1,x_2,\ldots,x_m$. Let $\mathcal{G}'=(L',R',E')$ be the resulting $3$-left regular bipartite graph, where $L'=L\cup\{x_1,x_2,\ldots,x_m\}$. Let $\mathcal{A}$ and $\mathcal{A}'$ denote the families of all $(a,0)$-TSs of $\mathcal{G}$ and $\mathcal{G}'$, respectively. Then the following statements hold.
	\begin{enumerate}
		\item[\textnormal{(i)}] \textbf{Equivalent-variable constraint.} For every $S'\in\mathcal{A}'$,
		\[
		v\in S'\iff x_1\in S'\iff\cdots\iff x_m\in S'.
		\]
		Equivalently, $v$ and all replication variables are either selected simultaneously or excluded simultaneously in every $(a,0)$-TS of $\mathcal{G}'$.
		\item[\textnormal{(ii)}] \textbf{Bijection of $(a,0)$-TSs.} The map $F_{v,m}:\mathcal{A}\to\mathcal{A}'$ defined by
		\[
		F_{v,m}(S)=
		\begin{cases}
			S\cup\{x_1,x_2,\ldots,x_m\}, & \text{if } v\in S,\\
			S, & \text{if } v\notin S,
		\end{cases}
		\]
		is a bijection.
	\end{enumerate}
\end{lemma}

\begin{proof}
	As in the preceding lemma, $(a,0)$-TSs are interpreted through the even-neighbor parity constraints at check nodes, and all congruences are over $\mathbb{F}_2$. Let $N_\mathcal{G}(v)=\{r_1,r_2,r_3\}$ be the three original neighbors of $v$, labelled as in Algorithm~\ref{alg:expan}.
	
	We first prove the equivalent-variable constraint. Let $S'\in\mathcal{A}'$. The auxiliary check nodes $y_1$ and $y_2$ have neighborhoods $\{v,x_1\}$ and $\{v,x_2\}$, respectively. Hence
	\[
	\boldsymbol{1}_{S'}(v)+\boldsymbol{1}_{S'}(x_1)\equiv 0,
	\qquad
	\boldsymbol{1}_{S'}(v)+\boldsymbol{1}_{S'}(x_2)\equiv 0,
	\]
	so $x_1$ and $x_2$ have the same membership status as $v$. Now consider any iteration of Algorithm~\ref{alg:expan} that introduces $x_{i-1}$ and $x_i$, where $i=4,6,\ldots,m$. The auxiliary check nodes $y_{i-1}$ and $y_i$ have neighborhoods $\{x_{i-3},x_{i-1}\}$ and $\{x_{i-2},x_i\}$, respectively. Their parity constraints imply
	\[
	\boldsymbol{1}_{S'}(x_{i-3})=\boldsymbol{1}_{S'}(x_{i-1}),
	\qquad
	\boldsymbol{1}_{S'}(x_{i-2})=\boldsymbol{1}_{S'}(x_i).
	\]
	Induction over $i=4,6,\ldots,m$ therefore yields
	\[
	\boldsymbol{1}_{S'}(x_1)=\boldsymbol{1}_{S'}(x_2)=\cdots=\boldsymbol{1}_{S'}(x_m)=\boldsymbol{1}_{S'}(v).
	\]
	The additional check nodes with neighborhoods $\{x_{i-1},x_i\}$, including the terminal check nodes, are then automatically satisfied and do not alter this conclusion. This proves (i).
	
	We next show that $F_{v,m}$ is well defined. Let $S\in\mathcal{A}$. The map $F_{v,m}$ preserves every original variable in $L$ and assigns each $x_i$ the same membership status as $v$. For every original check node other than possibly $r_1$ and $r_2$, the neighborhood contribution is unchanged. At $r_1$ and $r_2$, Algorithm~\ref{alg:expan} replaces the incidences of $v$ by incidences of $x_1$ and $x_2$, respectively; since $\boldsymbol{1}_{F_{v,m}(S)}(x_1)=\boldsymbol{1}_{F_{v,m}(S)}(x_2)=\boldsymbol{1}_{S}(v)$, the parity at these two check nodes is exactly the same as in $G$. All newly introduced auxiliary check nodes are adjacent to pairs of variables that have identical membership statuses under $F_{v,m}(S)$, and therefore each such check node is incident with either zero or two selected variables. Hence every check node of $\mathcal{G}'$ has even intersection with $F_{v,m}(S)$, and $F_{v,m}(S)\in\mathcal{A}'$.
	
	Define $P:\mathcal{A}'\to 2^L$ by $P(S')=S'\cap L$. We prove that $P(S')\in\mathcal{A}$ for every $S'\in\mathcal{A}'$. Let $S'\in\mathcal{A}'$. For check nodes whose neighborhoods are not modified by the replication, the parity condition for $P(S')$ is inherited directly from $S'$. For $r_1$ and $r_2$, the graph $G'$ contains $x_1$ and $x_2$ in place of $v$, respectively. By (i), $\boldsymbol{1}_{S'}(x_1)=\boldsymbol{1}_{S'}(x_2)=\boldsymbol{1}_{S'}(v)$. Replacing $x_1$ and $x_2$ by $v$ therefore preserves parity at $r_1$ and $r_2$ when passing from $G'$ back to $G$. The third neighbor $r_3$ of $v$ is unchanged. Thus every check node of $G$ has even intersection with $P(S')$, and $P(S')\in\mathcal{A}$.
	
	Finally, $P(F_{v,m}(S))=S$ for every $S\in\mathcal{A}$, since $P$ simply deletes the replication variables. Conversely, for any $S'\in\mathcal{A}'$, statement (i) shows that the memberships of $x_1,\ldots,x_m$ are uniquely determined by the membership of $v$. Therefore applying $F_{v,m}$ to $P(S')$ restores exactly the original statuses of all replication variables, while all original variables are preserved by $P$. Hence $F_{v,m}(P(S'))=S'$. Thus $P=F_{v,m}^{-1}$ on $\mathcal{A}'$, so $F_{v,m}$ is a bijection. This proves (ii) and completes the proof.
\end{proof}

The three algorithmic primitives detailed in this section comprise the core operational framework for the hardness reductions that follow. By operating locally and independently on individual structural elements without requiring global coordination, these algorithms can be efficiently composed in polynomial time. In Section~\ref{sec:left}, Algorithm~\ref{alg:hyp} and Algorithm~\ref{alg:expan} will be applied in tandem to transform arbitrary hypergraphs into specifically constrained uniform topologies. Section~\ref{sec:regular} will then incorporate Algorithm~\ref{alg:chec} alongside variable replication to strictly meet $(J,K)$-regularity conditions on both sides of the bipartition.

\section{Left Regular Minimum Distance Problem}
\label{sec:left}

This section investigates the MDP for Tanner graphs with a constant left degree $J$, which is equivalent to analyzing parity-check matrices with a constant column weight. Although the MDP is known to be $\mathrm{NP}$-complete for arbitrary linear codes, this result does not immediately determine its complexity under a fixed left-degree constraint. Indeed, reductions for the unrestricted problem may produce columns of nonuniform or unbounded weight and therefore do not necessarily yield instances belonging to any fixed left-regular class. A separate analysis is thus required to determine whether computational hardness persists when every variable node has the same prescribed degree and, if it does, to identify the smallest degree at which intractability arises.

For $J \leq 2$, the problem is solvable in polynomial time. When $J=1$, the associated hypergraph is $1$-uniform; finding a nonempty even cover reduces to checking whether a vertex is incident to two identical singleton hyperedges, which is natively impossible in simple hypergraphs. When $J=2$, the hypergraph degenerates into an ordinary graph (potentially with parallel edges). In this configuration, nonempty even covers correspond exclusively to nonempty unions of cycles. The minimum distance is therefore the length of the shortest cycle in the graph, where a pair of parallel edges counts as a cycle of length two. Since the shortest cycle can be found in polynomial time, $J=3$ represents the strict lower bound at which computational hardness can emerge.

Through the correspondence established in Proposition~\ref{prop:code-even}, codewords in a $J$-left regular Tanner graph correspond exactly to even covers in the associated $J$-uniform hypergraph, with the Hamming weight of a codeword equal to the number of hyperedges in the corresponding even cover. Using this correspondence, the following subsections show that the boundary at $J=3$ precisely separates the tractable and intractable cases for the standard at-most-weight problem. We first prove $\mathrm{NP}$-completeness for $3$-left regular Tanner graphs by reducing the unrestricted minimum even cover problem to its $3$-uniform restriction. The construction combines the intermediate transformations introduced in Section~\ref{sec:inter}, namely hyperedge decomposition, structural padding, and variable replication. Hyperedge decomposition and structural padding enforce $3$-uniformity, whereas variable replication ensures that the nonuniform auxiliary contributions introduced during the transformation cannot alter the ordering of the original solution sizes. Separately, the $\mathrm{W}[1]$-completeness of the exact-weight variant follows from the homogeneous exact-weight instances of Arvind et al.~\cite{Arvind16}, in which every variable occurs exactly three times. Finally, we extend both results from left degree three to every fixed $J>3$ by adding $J-3$ redundant global check nodes. Since every feasible support in a $3$-left regular Tanner graph has even cardinality, these additional checks impose no new constraint and preserve support sizes exactly.

\subsection{Intractability of the $3$-Left Regular Case}

This subsection establishes two complementary hardness results for Tanner graphs of left degree exactly three. The standard at-most-weight problem remains $\mathrm{NP}$-complete, while its exact-weight counterpart is $\mathrm{W}[1]$-complete when parameterized by the prescribed weight.

\begin{theorem}\label{l3}
	The decision problem \MDt\ is $\mathrm{NP}$-complete. Moreover, its exact-weight variant $\MDt^{=}$ is $\mathrm{W}[1]$-complete when parameterized by the prescribed weight $w$.
\end{theorem}

\begin{proof}
	We first prove $\mathrm{NP}$-completeness of the standard at-most-weight
	problem. The reduction has three stages. Starting from a hypergraph
	$\mathcal{H}$, Algorithm~\ref{alg:hyp} produces a hypergraph
	$\mathcal{H}_0$ whose hyperedges have size at most three. We then pad the
	hyperedges of sizes one and two to obtain a $3$-uniform hypergraph
	$\mathcal{H}_1$. Finally, after passing to the Tanner-graph representation,
	we replicate one designated variable from each decomposition chain and
	obtain a $3$-left regular Tanner graph $\mathcal{G}_2$. We construct
	bijections between the feasible families in consecutive stages and choose
	the replication counts so that an original cover of size $t$ is transformed
	into a feasible support whose size lies between $tB$ and $t(B+1)+3$.
	This interval separation yields the required Karp reduction.

	Let $(\mathcal{H},w)$ be an instance of \MEC, where
	$\mathcal{H}=(V,\mathcal{E})$. Replace $w$ by
	$\min\{w,|\mathcal{E}|\}$. Next, introduce a private vertex $q_\star$ and
	a dummy singleton hyperedge $e_\star=\{q_\star\}$. Since $q_\star$ is
	incident with no other hyperedge, $e_\star$ cannot belong to an even
	cover. Thus the feasible nonempty even covers are unchanged, while we may
	assume that $w<|\mathcal{E}|$. We continue to denote the padded instance
	by $(\mathcal{H},w)$.
    
	Apply Algorithm~\ref{alg:hyp} to $\mathcal{H}$ and let
	$\mathcal{H}_0=(V_0,\mathcal{E}_0)$ be the resulting hypergraph. For each
	original hyperedge $e\in\mathcal{E}$, let $\mathcal{C}(e)$ be the chain
	produced from $e$; an undecomposed hyperedge is regarded as a chain of
	length one. We write
	\[
	T_{\mathrm{hyp}}:
	\mathcal{A}(\mathcal{H})\longrightarrow
	\mathcal{A}(\mathcal{H}_0)
	\]
	for the instance of the map $T$ from Lemma~\ref{lem:3.1}. Explicitly,
	\[
	T_{\mathrm{hyp}}(X)=\bigcup_{e\in X}\mathcal{C}(e).
	\]
	Thus $T_{\mathrm{hyp}}$ replaces every selected original hyperedge by its
	complete decomposition chain. By Lemma~\ref{lem:3.1}, it is a bijection,
	and its inverse contracts every selected chain to its unique original
	hyperedge. The subscript ``$\mathrm{hyp}$'' merely distinguishes this
	instance of $T$ from the other maps used below.

	The hypergraph $\mathcal{H}_0$ contains only hyperedges of sizes at most
	three. We next construct a $3$-uniform hypergraph
	$\mathcal{H}_1=(V_1,\mathcal{E}_1)$. Introduce four fresh vertices
	$y_1,y_2,z_1,z_2$ and put
	\[
	V_1=V_0\cup\{y_1,y_2,z_1,z_2\}.
	\]
	For every $f\in\mathcal{E}_0$, define its padded copy by
	\[
	\rho(f)=
	\begin{cases}
		f, & |f|=3,\\
		f\cup\{z_1\}, & |f|=2,\\
		f\cup\{y_1,y_2\}, & |f|=1.
	\end{cases}
	\]
	Introduce three additional hyperedges
	\[
	\begin{aligned}
		a_1&=\{y_1,y_2,z_1\},\\
		a_2&=\{y_1,z_1,z_2\},\\
		a_3&=\{y_2,z_1,z_2\},
	\end{aligned}
	\qquad
	A=\{a_1,a_2,a_3\},
	\]
	and define the indexed hyperedge family
	\[
	\mathcal{E}_1=\rho(\mathcal{E}_0)\mathbin{\dot\cup}A.
	\]

	We now define the padding map
	\[
	U_{\mathrm{pad}}:
	\mathcal{A}(\mathcal{H}_0)\longrightarrow
	\mathcal{A}(\mathcal{H}_1).
	\]
	Let $Y_0\in\mathcal{A}(\mathcal{H}_0)$ and set
	\begin{align*}
	&\alpha(Y_0)
	\equiv
	\bigl|\{f\in Y_0:|f|=1\}\bigr|
	\pmod 2,\\
	\quad
	&\beta(Y_0)
	\equiv
	\bigl|\{f\in Y_0:|f|=2\}\bigr|
	\pmod 2,
	\end{align*}
	where $\alpha(Y_0),\beta(Y_0)\in\{0,1\}$. The primary padded family
	$\rho(Y_0)=\{\rho(f):f\in Y_0\}$ induces, in the coordinate order
	$(y_1,y_2,z_1,z_2)$, the parity vector
	\[
	\bigl(\alpha(Y_0),\alpha(Y_0),\beta(Y_0),0\bigr).
	\]
	For $(\alpha,\beta)\in\{0,1\}^2$, define
	\[
	R_{\alpha,\beta}=
	\begin{cases}
		\emptyset, & (\alpha,\beta)=(0,0),\\
		A, & (\alpha,\beta)=(0,1),\\
		\{a_2,a_3\}, & (\alpha,\beta)=(1,0),\\
		\{a_1\}, & (\alpha,\beta)=(1,1).
	\end{cases}
	\]
	We set
	\[
	U_{\mathrm{pad}}(Y_0)
	=
	\rho(Y_0)\cup R_{\alpha(Y_0),\beta(Y_0)}.
	\]

	Padding leaves every incidence with a vertex of $V_0$ unchanged. Hence
	$\rho(Y_0)$ already has even degree at every vertex of $V_0$. The four
	correction sets in the preceding display have auxiliary parity vectors
	\[
	(0,0,0,0),\quad
	(0,0,1,0),\quad
	(1,1,0,0),\quad
	(1,1,1,0),
	\]
	respectively, so the chosen correction cancels exactly the auxiliary
	parity vector of $\rho(Y_0)$. Therefore $U_{\mathrm{pad}}(Y_0)$ is an
	even cover of $\mathcal{H}_1$.

	Conversely, let $Y_1$ be an even cover of $\mathcal{H}_1$. Delete the
	hyperedges in $A$ and remove the padding vertices from the remaining
	hyperedges. Since the hyperedges in $A$ do not meet $V_0$, the resulting
	indexed family is an even cover $Y_0$ of $\mathcal{H}_0$. Moreover, the
	incidence vectors of $a_1,a_2,a_3$ on
	$\{y_1,y_2,z_1,z_2\}$ are linearly independent over $\mathbb{F}_2$.
	Thus the parity conditions at the four auxiliary vertices uniquely
	determine the selected subset of $A$, and that subset must be
	$R_{\alpha(Y_0),\beta(Y_0)}$. Hence $U_{\mathrm{pad}}$ is a bijection.
	It maps the empty cover only to the empty cover and therefore restricts to
	a bijection between the nonempty even covers. In addition,
	\[
	\bigl|U_{\mathrm{pad}}(Y_0)\cap A\bigr|\leq3.
	\]

	Let $\mathcal{G}_1$ be the Tanner graph associated with
	$\mathcal{H}_1$. By Proposition~\ref{prop:code-even}, every hyperedge of
	$\mathcal{H}_1$ corresponds to one variable node of $\mathcal{G}_1$, and
	$\mathcal{G}_1$ is $3$-left regular. Enumerate the original indexed
	hyperedges as $\mathcal{E}=(e_1,\ldots,e_N)$ and put
	\[
	s_i=|\mathcal{C}(e_i)|,
	\qquad
	D_i=\{\rho(f):f\in\mathcal{C}(e_i)\}.
	\]
	The family $D_i$ consists of the $s_i$ primary hyperedges representing
	$e_i$. By the definitions of $T_{\mathrm{hyp}}$ and
	$U_{\mathrm{pad}}$, all hyperedges in $D_i$ are selected if
	$e_i$ is selected, and none is selected otherwise. Choose one fixed
	representative hyperedge $f_i^{\circ}\in D_i$ and let $v_i$ be its
	corresponding variable node in $\mathcal{G}_1$. Thus $v_i$ is not an
	arbitrary variable: its membership records the common selection status of
	the entire chain associated with $e_i$. None of the three auxiliary
	hyperedges in $A$ is chosen as a representative.

	Set
	\[
	B=
	\max\left\{
		w+4,
		2+\max_{1\leq i\leq N}s_i
	\right\}.
	\]
	For each $i$, exactly one of the consecutive integers $B-s_i$ and
	$B+1-s_i$ is even. Let $m_i$ be that even integer. Since
	$B\geq s_i+2$, we have $m_i\geq2$, and
	\[
	s_i+m_i\in\{B,B+1\}.
	\]

	Starting from $\mathcal{G}_1$, apply Algorithm~\ref{alg:expan}
	successively to $v_1,\ldots,v_N$ with replication counts
	$m_1,\ldots,m_N$, and let $\mathcal{G}_2$ be the resulting Tanner graph.
	Every application preserves $3$-left regularity. Let
	\[
	F_{\mathrm{exp}}:
	\mathcal{A}(\mathcal{G}_1)\longrightarrow
	\mathcal{A}(\mathcal{G}_2)
	\]
	denote the composition
	\[
	F_{\mathrm{exp}}
	=
	F_{v_N,m_N}\circ\cdots\circ F_{v_1,m_1}.
	\]
	By Lemma~\ref{lem:3.3}, $F_{\mathrm{exp}}$ is a bijection. Identifying
	even covers of $\mathcal{H}_1$ with the corresponding $(a,0)$-TSs of
	$\mathcal{G}_1$, define
	\[
	\Psi
	\coloneqq
	F_{\mathrm{exp}}\circ U_{\mathrm{pad}}\circ T_{\mathrm{hyp}}:
	\mathcal{A}(\mathcal{H})\longrightarrow
	\mathcal{A}(\mathcal{G}_2).
	\]
	All three component maps are bijections and preserve nonemptiness, so
	$\Psi$ is a bijection between the nonempty feasible families.

	Let $X$ be a nonempty even cover of $\mathcal{H}$ and put $t=|X|$.
	Whenever $e_i\in X$, its $s_i$ primary variables and the $m_i$
	replication variables associated with $v_i$ are all selected. The padding
	correction contributes at most three additional variables, and those
	auxiliary variables are not replicated. Consequently,
	\[
	|\Psi(X)|
	=
	\sum_{e_i\in X}(s_i+m_i)
	+
	\bigl|U_{\mathrm{pad}}(T_{\mathrm{hyp}}(X))\cap A\bigr|,
	\]
	which gives
	\[
	tB\leq |\Psi(X)|\leq t(B+1)+3.
	\]
	Set
	\[
	\widehat{w}=w(B+1)+3.
	\]
	If $|X|\leq w$, then $|\Psi(X)|\leq\widehat{w}$. Conversely, suppose
	that $Y$ is a nonempty $(a,0)$-TS of $\mathcal{G}_2$ satisfying
	$|Y|\leq\widehat{w}$. Since $\Psi$ is bijective, $Y=\Psi(X)$ for a
	unique nonempty even cover $X$ of $\mathcal{H}$. If $|X|\geq w+1$,
	then
	\[
	|Y|
	\geq
	(w+1)B
	>
	w(B+1)+3
	=
	\widehat{w},
	\]
	where the strict inequality follows from $B\geq w+4$. This is a
	contradiction. Therefore $|X|\leq w$, and we have proved
	\[
	(\mathcal{H},w)\in\MEC
	\quad\Longleftrightarrow\quad
	(\mathcal{G}_2,\widehat{w})\in\MDt.
	\]

	The construction is polynomial. Indeed,
	$s_i\leq\max\{1,|e_i|-2\}$, and the preliminary bound
	$w<|\mathcal{E}|$ implies that $B$ is polynomially bounded by the input
	representation size. There are $N$ replications, each of size
	$\mathcal{O}(B)$, so the output graph has polynomial size. Membership in
	$\mathrm{NP}$ is immediate because a proposed nonempty support of size at
	most $\widehat{w}$ can be checked against all parity constraints in
	polynomial time. Hence \MDt\ is $\mathrm{NP}$-complete.

	We now turn to the exact-weight variant $\MDt^{=}$. To prove membership
	in $\mathrm{W}[1]$, let $(\mathcal{G},w)$ be an instance, where
	$\mathcal{G}=(L,R,E)$ is $3$-left regular. Construct a nondeterministic
	multitape Turing machine $M_{\mathcal{G}}$ whose finite control contains
	the incidence relation of $\mathcal{G}$. The machine guesses exactly $w$
	pairwise distinct variable nodes and records the three neighboring check
	nodes of each guessed variable. Thus at most $3w$ check-node labels are
	recorded. It accepts precisely when every recorded check-node label occurs
	an even number of times. Pairwise distinctness and all parity conditions
	can be verified with $\mathcal{O}(w^2)$ comparisons. Consequently, for
	some computable function $f$,
	\[
	M_{\mathcal{G}}
	\text{ accepts within }f(w)\text{ steps}
	\ \Longleftrightarrow\ 
	(\mathcal{G},w)\in\MDt^{=}.
	\]
	The description of $M_{\mathcal{G}}$ has size polynomial in
	$|\mathcal{G}|$, whereas $f(w)$ depends only on $w$. This gives an FPT
	reduction to \textsc{Nondet TM Acceptance}, and hence
	$\MDt^{=}\in\mathrm{W}[1]$.

	For hardness, Arvind et al.\ proved that the homogeneous exact-weight
	linear-equation problem is $\mathrm{W}[1]$-hard even when every variable
	occurs exactly three times \cite{Arvind16}. Let $(M,\boldsymbol{0},w)$ be
	one of these instances. Regard every equation as a hypergraph vertex and
	every variable as the indexed hyperedge consisting of the equations in
	which that variable occurs. Since every variable occurs exactly three
	times, the resulting hypergraph is already $3$-uniform; equivalently, its
	Tanner graph $\mathcal{G}_M$ is already $3$-left regular. Moreover, the
	incidence transformation preserves the number of selected variables
	exactly. Therefore, $M\boldsymbol{x}^{\top}=\boldsymbol{0}
	\text{ has a solution of weight exactly }w
	\quad\Longleftrightarrow\quad
	(\mathcal{G}_M,w)\in\MDt^{=}$.
	The parameter remains $w$, so this is a parameter-preserving FPT
	reduction. Thus $\MDt^{=}$ is $\mathrm{W}[1]$-hard. Combining hardness
	with membership proves that $\MDt^{=}$ is $\mathrm{W}[1]$-complete.
\end{proof}

Theorem~\ref{l3} identifies left degree three as the first
intractable case. For the standard at-most-weight problem, the replication
step separates the size intervals corresponding to different original cover
sizes and thereby yields a Karp reduction. For the exact-weight parameterized
problem, no such transformation is needed: the sparse homogeneous instances
of Arvind et al.~\cite{Arvind16} already have column weight exactly three.

\subsection{Intractability of Left Degree at Least Three}

We next lift both conclusions from left degree three to any fixed
left degree $J>3$. Given a $3$-left regular Tanner graph, we add $J-3$
global check nodes and join each of them to every variable node. Every
feasible support in a $3$-left regular Tanner graph has even cardinality, so
the new checks are redundant. Consequently, this construction preserves not
only feasibility but also the cardinality of every feasible support.

\begin{theorem}\label{l}
	For every fixed integer $J\geq3$, the problem \MTJ\ is
	$\mathrm{NP}$-complete. Moreover, its exact-weight variant $\MTJ^{=}$ is
	$\mathrm{W}[1]$-complete when parameterized by the prescribed weight $w$.
\end{theorem}

\begin{proof}
	The case $J=3$ is Theorem~\ref{l3}. Fix $J>3$, and let
	$\mathcal{G}=(L,R,E)$ be a $3$-left regular Tanner graph. Construct
	$\mathcal{G}'=(L,R',E')$ by setting
	\[
	R'=R\mathbin{\dot\cup}\{r_1^*,\ldots,r_{J-3}^*\}
	\]
	and joining every new check node $r_i^*$ to every variable node in $L$.
	All original edges are retained. Thus every variable node gains exactly
	$J-3$ neighbors, and $\mathcal{G}'$ is $J$-left regular. Since $J$ is
	fixed, the construction is polynomial.

	We claim that the identity map on subsets of $L$ induces a
	cardinality-preserving bijection
	\[
	\Gamma_J:
	\mathcal{A}(\mathcal{G})\longrightarrow
	\mathcal{A}(\mathcal{G}'),
	\qquad
	\Gamma_J(S)=S.
	\]
	Indeed, if $S\in\mathcal{A}(\mathcal{G})$, then every original check node
	has even degree in the subgraph induced by $S$. Summing these degrees gives
	\[
	3|S|
	=
	\sum_{r\in R}|N_{\mathcal{G}}(r)\cap S|
	\equiv0\pmod2.
	\]
	Because $3$ is odd, $|S|$ is even. Each new global check node has exactly
	$|S|$ selected neighbors and hence also has even induced degree. Therefore
	$S\in\mathcal{A}(\mathcal{G}')$. Conversely, if
	$S\in\mathcal{A}(\mathcal{G}')$, then all parity constraints at the
	original check nodes in $R$ hold, and hence
	$S\in\mathcal{A}(\mathcal{G})$. This proves the claim; in particular,
	\[
	|\Gamma_J(S)|=|S|
	\qquad
	\text{for every }S\in\mathcal{A}(\mathcal{G}).
	\]

	For the standard at-most-weight problems, use the same bound $w$ in the
	output instance. The preceding bijection gives
	\[
	(\mathcal{G},w)\in\MDt
	\quad\Longleftrightarrow\quad
	(\mathcal{G}',w)\in\MTJ.
	\]
	This is a Karp reduction from \MDt\ to \MTJ. Theorem~\ref{l3} supplies
	$\mathrm{NP}$-hardness, while membership in $\mathrm{NP}$ follows by
	checking a proposed support and all incident parities. Thus \MTJ\ is
	$\mathrm{NP}$-complete.

	For the exact-weight problems, the same bijection gives
	\[
	(\mathcal{G},w)\in\MDt^{=}
	\quad\Longleftrightarrow\quad
	(\mathcal{G}',w)\in\MTJ^{=}.
	\]
	The parameter is unchanged, so Theorem~\ref{l3} transfers
	$\mathrm{W}[1]$-hardness to $\MTJ^{=}$. To prove membership, use the
	nondeterministic-machine construction from that theorem: guess exactly
	$w$ distinct variable nodes and check the parity of their incident check
	nodes. Since $J$ is fixed, at most $Jw$ check-node labels are recorded, and
	the verification takes a number of steps depending only on $w$ (and the
	fixed constant $J$). Hence $\MTJ^{=}\in\mathrm{W}[1]$, completing the
	proof of $\mathrm{W}[1]$-completeness.
\end{proof}

Theorems~\ref{l3} and~\ref{l} jointly delineate a sharp
classical-complexity boundary for left regular Tanner graphs. Left degrees at
most two admit polynomial-time algorithms, whereas the standard at-most-weight
problem is $\mathrm{NP}$-complete for every fixed $J\geq3$. In the
parameterized setting, the corresponding exact-weight problem is
$\mathrm{W}[1]$-complete for every such $J$. Thus left regularity alone does
not remove the intrinsic difficulty of the relevant minimum-distance and
exact-weight codeword questions. This motivates the biregular setting
considered next, where both variable-node and check-node degrees are
simultaneously constrained.

\section{Biregular Minimum Distance Problem} \label{sec:regular}

We now impose regularity on both sides of the Tanner graph. The hardness results established in Section~\ref{sec:left} for left-regular instances do not, by themselves, determine the complexity of the fully biregular case, because the check-node degrees in those instances may remain nonuniform or unbounded. Requiring every check node to have the same prescribed degree is a substantial additional restriction that may exclude the hard instances produced by the earlier reductions. It is therefore necessary to develop a separate regularization procedure that enforces exact degrees on both sides while preserving the relevant parity constraints and the ordering of nonzero support sizes.

The low-right-degree regimes remain algorithmically tractable. If every check node has degree one, each parity check forces its unique neighboring variable to be zero. Since every variable node has positive degree, the resulting code contains no nonzero codeword. If every check node has degree two, each check imposes an equality constraint between its two neighboring variables. Construct an auxiliary graph whose vertices are the variable nodes and whose edges correspond to the degree-two checks. All variables within the same connected component must then take the same value. Consequently, every nonzero codeword is the indicator vector of a nonempty union of connected components, and the minimum distance is the number of vertices in a smallest connected component. This quantity can be computed in polynomial time. Hence $K=3$ is the smallest right degree at which computational hardness can arise.

In matrix terminology, a $(J,K)$-regular Tanner graph corresponds to a parity-check matrix with constant column weight $J$ and constant row weight $K$. Equivalently, under the hypergraph formulation, it corresponds to a $J$-uniform, $K$-regular incidence structure. The remainder of this section proves that the tractability boundary identified above is tight and establishes hardness under full biregularity. We first consider the base case $(3,3)$, starting from the $3$-left regular problem proved hard in Section~\ref{sec:left}. Check node splitting and the subsequent degree-completion operations are used to make every check node have degree three, while the map $Q$ from Lemma~\ref{lem:3.2} describes the induced correspondence between feasible supports. Since these operations introduce auxiliary variables whose contribution may depend on the selected support, the replication maps $F_{v,m}$ from Lemma~\ref{lem:3.3} are then applied with a sufficiently large common replication factor to ensure that the original support size dominates all auxiliary contributions. This yields hardness for the $(3,3)$-regular case. Next, for every fixed $K\geq 3$, a single-port completion gadget raises the right degree from three to $K$ while preserving the left degree and the minimum-support correspondence, giving a polynomial-time Karp reduction from the $(3,3)$-regular problem to the $(3,K)$-regular problem. Finally, a replica-and-global-check construction increases the left degree one unit at a time and provides both a Karp reduction and an FPT reduction from the $(J-1,K)$-regular case to the $(J,K)$-regular case.

\subsection{Intractability of the $(3,3)$-Regular Case}

The base case $(3,3)$ requires simultaneous structural control over both sides of the bipartition. Starting from the $3$-left regular instances analyzed in Section \ref{sec:left}, the reduction utilizes the check node splitting procedure (Algorithm~\ref{alg:chec}) formalized via the bijection $Q$ (from Lemma~\ref{lem:3.2}) to restrict check node degrees to at most three. To address the auxiliary variable nodes of degree two introduced during this split, the construction duplicates the graph, couples the corresponding auxiliary variables, and applies the variable replication procedure (Algorithm~\ref{alg:expan} and Lemma~\ref{lem:3.3}) to amplify non-auxiliary variables. The structural regularization is then completed by appending specifically configured suffix variable gadgets across disjoint copies of the graph to eliminate any remaining legacy check nodes of degree two. The formal properties and threshold preservation of this transformation are established in Theorem~\ref{thm:np33}.

\begin{theorem}\label{thm:np33}
	The problem \MDtt\ is $\mathrm{NP}$-complete.
\end{theorem}

\begin{proof}
	Membership in $\mathrm{NP}$ is immediate: a nonempty set of at
	most $w$ variable nodes is a polynomially verifiable certificate.
	We prove $\mathrm{NP}$-hardness by a polynomial-time reduction
	from \MTt, which is $\mathrm{NP}$-complete by
	Theorem~\ref{l3} and Proposition~\ref{prop:trap-codeword}.

	Let $(\mathcal{G},w)$ be an instance of \MTt, where
	$\mathcal{G}=(L,R,E)$ is $3$-left regular. First delete isolated
	check nodes. We then repeatedly delete every degree-one check
	node together with its unique neighboring variable node. Indeed,
	if a degree-one check node is adjacent only to a variable $v$,
	then the even-parity condition forces $v$ to be absent from every
	$(a,0)$-TS. Consequently, restriction to the remaining variables,
	and extension by assigning zero to the deleted variables, preserve
	all feasible supports. After each deletion we again remove isolated
	check nodes and continue until every remaining check node has degree
	at least two.

	If no variable node remains, the input is a no-instance, and we may
	output any fixed no-instance of \MDtt. Hence assume that $L\neq
	\varnothing$. Put
	\[
		n=|L|
		\qquad\text{and}\qquad
		w\gets\min\{w,n\}.
	\]

	\medskip
	\noindent\textbf{Phase I: reduction to check degrees two and three.}
	Apply the check-node splitting procedure of
	Algorithm~\ref{alg:chec} to obtain
	\[
		\mathcal{G}_1=(L_1,R_1,E_1),
		\qquad
		L_1=L\mathbin{\dot\cup}L_{\mathrm{aux}},
	\]
	in which every check node has degree two or three. Let
	\[
		Q_{\mathrm{split}}:
		\mathcal{A}(\mathcal{G})
		\longrightarrow
		\mathcal{A}(\mathcal{G}_1)
	\]
	be the map from Lemma~\ref{lem:3.2}. Its inverse is the projection
	onto $L$. Every original variable in $L$ still has degree three,
	whereas every variable in $L_{\mathrm{aux}}$ has degree two.

	Take two disjoint copies
	$\mathcal{G}_1^{(0)}$ and $\mathcal{G}_1^{(1)}$ of
	$\mathcal{G}_1$. For every $u\in L_1$, add a degree-two coupling
	check $z_u$ adjacent to the two copies $u^{(0)}$ and $u^{(1)}$.
	The parity equation at $z_u$ forces
	\[
		\boldsymbol{1}_{S}(u^{(0)})
		=
		\boldsymbol{1}_{S}(u^{(1)})
	\]
	in every $(a,0)$-TS $S$. Thus every auxiliary variable now has
	degree three. Each original variable, however, has degree four.

	We reduce these degree-four original variables to degree three by
	the variable-side dual of Algorithm~\ref{alg:chec}. More precisely,
	for every $v\in L$ and $\ell\in\{0,1\}$, enumerate the four check
	nodes adjacent to $v^{(\ell)}$ as
	$r_1,r_2,r_3,r_4$. Replace $v^{(\ell)}$ by two variables
	$v_1^{(\ell)}$ and $v_2^{(\ell)}$, set
	\[
		N(v_1^{(\ell)})=\{r_1,r_2,e_v^{(\ell)}\},
		\qquad
		N(v_2^{(\ell)})=\{r_3,r_4,e_v^{(\ell)}\},
	\]
	and introduce the degree-two check node $e_v^{(\ell)}$.
	The parity equation at $e_v^{(\ell)}$ forces
	\[
		\boldsymbol{1}_{S}(v_1^{(\ell)})
		=
		\boldsymbol{1}_{S}(v_2^{(\ell)}).
	\]
	Hence all four variables $v_1^{(0)},\ v_2^{(0)},\
		v_1^{(1)},\ v_2^{(1)}$
	have the same membership status in every $(a,0)$-TS. Denote the
	resulting graph by $\mathcal{G}_2$. Every variable node of
	$\mathcal{G}_2$ has degree three, and every check node has degree
	two or three.

	For $S_1\in\mathcal{A}(\mathcal{G}_1)$, let
	$D(S_1)$ be its canonical lift to $\mathcal{G}_2$: both copies of
	every selected auxiliary variable are selected, and all four split
	copies of every selected original variable are selected. The
	coupling checks and the splitting checks show conversely that every
	$(a,0)$-TS of $\mathcal{G}_2$ is obtained in this way.

	Let $M=|L_{\mathrm{aux}}|,
		\
		m=2\max\{1,M\},
		\
		\lambda=m+4$. Then $m$ is a positive even integer and $\lambda>2M$.
	For every original variable $v\in L$, choose one representative,
	say $p_v=v_1^{(0)}$, among its four equivalent copies and apply
	Algorithm~\ref{alg:expan} to $p_v$ with replication parameter $m$.
	Let
	\[
		\Phi_{\mathrm{amp}}
		=
		\mathop{\circ}_{v\in L} F_{p_v,m}
	\]
	denote the resulting composition, and let $\mathcal{G}_3$ be the
	final graph of this phase. By Lemma~\ref{lem:3.3}, the replication
	variables associated with $p_v$ are selected if and only if $v$ is
	selected.

	Define
	\[
		\Theta
		=
		\Phi_{\mathrm{amp}}
		\circ D
		\circ Q_{\mathrm{split}}.
	\]
	For $X\in\mathcal{A}(\mathcal{G})$, put
	\[
		a(X)
		=
		\bigl|
			Q_{\mathrm{split}}(X)\cap L_{\mathrm{aux}}
		\bigr|.
	\]
	Since $0\leq a(X)\leq M$, the construction gives the exact weight
	formula
	\begin{equation}\label{eq:np33-weight}
		|\Theta(X)|
		=
		\lambda |X|+2a(X),
		\qquad
		0\leq a(X)\leq M.
	\end{equation}
	Indeed, a selected original variable contributes four split copies
	and $m$ replication variables, whereas a selected splitting
	auxiliary variable contributes its two coupled copies.

	Conversely, the coupling checks, the variable-splitting checks,
	Lemma~\ref{lem:3.3}, and Lemma~\ref{lem:3.2} provide a projection
	\[
		\Pi:
		\mathcal{A}(\mathcal{G}_3)
		\longrightarrow
		\mathcal{A}(\mathcal{G})
	\]
	that collapses all equivalent variables and then deletes the
	check-splitting auxiliary variables. Every
	$Y\in\mathcal{A}(\mathcal{G}_3)$ satisfies
	\[
		|Y|
		=
		\lambda|\Pi(Y)|+2a(\Pi(Y)).
	\]

	Set $\widehat{w}=\lambda w+2M$. If $\mathcal{G}$ contains an $(a,0)$-TS $X$ with $|X|\leq w$, then
	by~\eqref{eq:np33-weight},
	\[
		|\Theta(X)|
		\leq
		\lambda w+2M
		=
		\widehat{w}.
	\]
	Conversely, suppose that
	$Y\in\mathcal{A}(\mathcal{G}_3)$ satisfies
	$|Y|\leq\widehat{w}$. If
	$|\Pi(Y)|\geq w+1$, then
	\[
		|Y|
		\geq
		\lambda(w+1)
		>
		\lambda w+2M
		=
		\widehat{w},
	\]
	where the strict inequality follows from $\lambda>2M$. This is a
	contradiction. Hence $|\Pi(Y)|\leq w$. We have therefore proved 	\begin{equation}\label{eq:np33-phase1}
	\begin{aligned}
		&(\mathcal{G},w)\in\MTt\\
		\Longleftrightarrow\quad
		&(\mathcal{G}_3,\widehat{w})
		\text{ contains a nonempty $(a,0)$-TS}.
	\end{aligned}
    \end{equation}

	We also record explicitly the correspondence between minimum
	solutions. Suppose that $\mathcal{G}$ has a nonempty $(a,0)$-TS,
	and let
	\[
		d=\min\{|X|:X\in\mathcal{A}(\mathcal{G})\}.
	\]
	Among all $X$ with $|X|=d$, choose $X^\star$ minimizing $a(X)$.
	We claim that $\Theta(X^\star)$ is a minimum $(a,0)$-TS of
	$\mathcal{G}_3$. Otherwise, let $Y$ be a strictly smaller
	$(a,0)$-TS and put $X=\Pi(Y)$. If $|X|\geq d+1$, then
	\[
		|Y|
		\geq
		\lambda(d+1)
		>
		\lambda d+2M
		\geq
		|\Theta(X^\star)|,
	\]
	a contradiction. If $|X|=d$, then the choice of $X^\star$ gives
	$a(X)\geq a(X^\star)$ and hence
	\[
		|Y|
		=
		\lambda d+2a(X)
		\geq
		\lambda d+2a(X^\star)
		=
		|\Theta(X^\star)|,
	\]
	again a contradiction. Thus $\Theta(X^\star)$ is minimum.
	The same argument shows that the projection of every minimum
	$(a,0)$-TS of $\mathcal{G}_3$ is a minimum $(a,0)$-TS of
	$\mathcal{G}$.

	\medskip
	\noindent\textbf{Phase II: elimination of degree-two check nodes.}
	Write $\mathcal{G}_3=(L_3,R_3,E_3),\ N=|L_3|$.
	By construction, $N=\lambda n+2M$, and since $w\leq n$, we have $\widehat{w}\leq N$.

	If $\mathcal{G}_3$ has no degree-two check node, then it is already
	$(3,3)$-regular. Otherwise, let $q>0$ be the number of its
	degree-two check nodes. Choose the smallest positive integer $c$
	such that
	\[
		3\mid c
		\qquad\text{and}\qquad
		h\coloneqq cq\geq 3N+3.
	\]
	Take $c$ disjoint copies of $\mathcal{G}_3$. Across these copies
	there are exactly $h$ degree-two check nodes. Let $C_0$ be their
	set and put
	\[
		t=\frac{h}{3}.
	\]

	We now perform the iterative completion depicted in
	Fig.~\ref{fig:np33_reduction}. At stage $j$, where
	$0\leq j\leq t-1$, the current set $C_j$ contains
	\[
		n_j=h-3j
	\]
	degree-two check nodes. Order them as
	$r_{j,1},\ldots,r_{j,n_j}$. Conceptually, we introduce one
	variable adjacent to all checks in $C_j$ and immediately apply the
	variable-side version of the splitting operation.

	If $n_j>3$, introduce the suffix variables $U_j=
		\{u_{j,1},\ldots,u_{j,n_j-2}\}$
	and the new degree-two check nodes $C_{j+1}
		=
		\{s_{j,1},\ldots,s_{j,n_j-3}\}$,
	with neighborhoods
	\[
		N(u_{j,1})
		=
		\{r_{j,1},r_{j,2},s_{j,1}\},
	\]
	\[
		N(u_{j,\ell})
		=
		\{s_{j,\ell-1},r_{j,\ell+1},s_{j,\ell}\},
		\qquad
		2\leq\ell\leq n_j-3,
	\]
	and
	\[
		N(u_{j,n_j-2})
		=
		\{s_{j,n_j-3},
		  r_{j,n_j-1},r_{j,n_j}\}.
	\]
	Thus every check in $C_j$ gains one neighbor and becomes
	degree three, every variable in $U_j$ has degree three, and the
	only newly created degree-two checks are those in $C_{j+1}$.
	In particular,
	\[
		|C_{j+1}|=|C_j|-3.
	\]

	At the terminal stage, $n_{t-1}=3$. We introduce a single variable
	$u_{t-1,1}$ adjacent to all three checks in $C_{t-1}$ and create
	no further check nodes. Let $\mathcal{G}_4$ be the resulting graph.
	Every variable node and every check node of $\mathcal{G}_4$ has
	degree three.

	It remains to prove that a small $(a,0)$-TS cannot contain suffix
	variables. Let
	\[
		U_{\mathrm{suf}}
		=
		\mathop{\dot\bigcup}_{j=0}^{t-1}U_j
	\]
	be the suffix-variable set, and suppose that
	$Z\in\mathcal{A}(\mathcal{G}_4)$ satisfies
	$Z\cap U_{\mathrm{suf}}\neq\varnothing$. Let
	\[
		k=\max\{j:Z\cap U_j\neq\varnothing\}
	\]
	be the last occupied stage.

	If $k<t-1$, then no variable in $U_{k+1}$ is selected. Every check
	in $C_{k+1}$ is adjacent to two consecutive variables in $U_k$
	and one variable in $U_{k+1}$. Hence its parity equation forces
	the two consecutive variables in $U_k$ to have the same membership
	status. Induction along the chain shows that all variables in
	$U_k$ have the same status. Since $Z\cap U_k$ is nonempty, every
	variable in $U_k$ is selected. The same conclusion is immediate
	when $k=t-1$, because $|U_{t-1}|=1$. Therefore
	\[
		|Z\cap U_k|
		=
		|U_k|
		=
		h-3k-2.
	\]

	Moreover, for every $0\leq j<k$,
	\[
		Z\cap U_{j+1}\neq\varnothing
		\quad\Longrightarrow\quad
		Z\cap U_j\neq\varnothing.
	\]
	Indeed, if no variable in $U_j$ were selected, then the parity
	equations at the checks in $C_{j+1}$ would force every adjacent
	variable in $U_{j+1}$ to be unselected. Since every variable in
	$U_{j+1}$ is adjacent to at least one check in $C_{j+1}$, this
	would imply $Z\cap U_{j+1}=\varnothing$, a contradiction.

	Consequently, each of the preceding sets
	$U_0,U_1,\ldots,U_{k-1}$ contains at least one selected variable.
	It follows that
	\begin{align*}
		|Z\cap U_{\mathrm{suf}}|
		&\geq
		(h-3k-2)+k\\
		&=
		h-2-2k\\
		&\geq
		h-2-2(t-1)\\
		&=
		\frac{h}{3}\\
		&>
		N.
	\end{align*}
	Thus every $(a,0)$-TS intersecting the suffix part has cardinality
	strictly greater than $N$, and therefore strictly greater than
	$\widehat{w}$.

	Any $(a,0)$-TS of $\mathcal{G}_3$ can be embedded into one of its
	$c$ copies in $\mathcal{G}_4$, with all suffix variables
	unselected. Conversely, if
	$Z\in\mathcal{A}(\mathcal{G}_4)$ satisfies
	$|Z|\leq\widehat{w}$, then the preceding bound gives
	$Z\cap U_{\mathrm{suf}}=\varnothing$. At every completed
	degree-two check, the suffix neighbor is therefore unselected, so
	the parity equation reduces exactly to the original parity equation
	in the corresponding copy of $\mathcal{G}_3$. Hence the restriction
	of $Z$ to each copy is either empty or an $(a,0)$-TS of
	$\mathcal{G}_3$. At least one restriction is nonempty and has
	cardinality at most $|Z|$. Therefore
	\begin{equation}\label{eq:np33-phase2}
		(\mathcal{G}_3,\widehat{w})
		\text{ is a yes-instance}
		\quad\Longleftrightarrow\quad
		(\mathcal{G}_4,\widehat{w})
		\text{ is a yes-instance}.
	\end{equation}

	For completeness, the minimum-to-minimum correspondence in this
	second phase follows directly by contradiction. Let $Y^\star$ be
	a minimum $(a,0)$-TS of $\mathcal{G}_3$, and let $d_3=|Y^\star|$.
	Embedding $Y^\star$ into one copy gives an $(a,0)$-TS of
	$\mathcal{G}_4$ of size $d_3\leq N$. Hence a minimum
	$(a,0)$-TS $Z^\star$ of $\mathcal{G}_4$ cannot contain a suffix
	variable. It therefore decomposes into $(a,0)$-TSs lying in the
	disjoint copies of $\mathcal{G}_3$. If
	$|Z^\star|<d_3$, then one nonempty component of $Z^\star$ would
	have size strictly smaller than $d_3$, contradicting the minimality
	of $Y^\star$. Thus
	\[
		|Z^\star|=d_3.
	\]
	Furthermore, $Z^\star$ has exactly one nonempty component, and that
	component is a minimum $(a,0)$-TS of $\mathcal{G}_3$. Combining
	this with the first-phase projection proves that every minimum
	$(a,0)$-TS of $\mathcal{G}_4$ projects to a minimum
	$(a,0)$-TS of the original graph, while the canonical lift of a
	suitably chosen minimum support of the original graph is minimum
	in $\mathcal{G}_4$.

	The entire construction is polynomial in the input size. In
	particular, $M$ and $m$ are polynomially bounded, the graph
	$\mathcal{G}_3$ has polynomial size, and the suffix construction
	introduces
	\[
		\sum_{j=0}^{t-1}(h-3j-2)=\mathcal{O}(h^2)
	\]
	variable nodes and a comparable number of check nodes. Since
	$\mathcal{G}_4$ is $(3,3)$-regular,
	Equations~\eqref{eq:np33-phase1} and~\eqref{eq:np33-phase2}
	give a polynomial-time Karp reduction from \MTt\ to \MDtt.
	Therefore \MDtt\ is $\mathrm{NP}$-complete.
\end{proof}

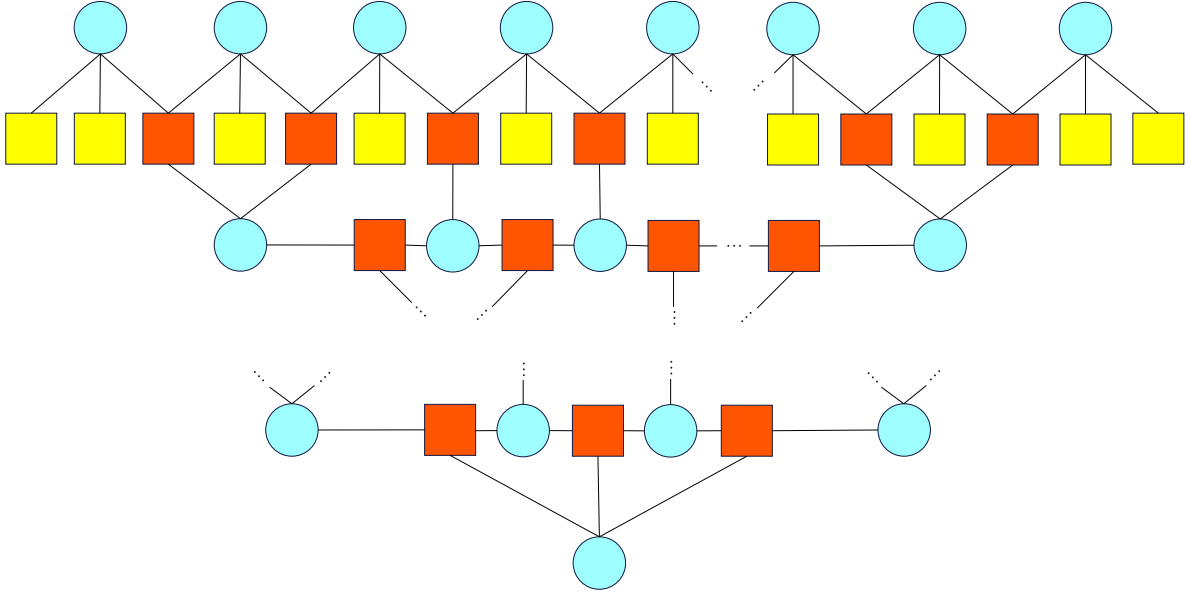
\begin{figure*}[h]
	\centering
		\resizebox{0.9\textwidth}{!}{%
		\begingroup
\definecolor{reduceCyan}{RGB}{157,253,255}
\definecolor{reduceOrange}{RGB}{255,85,0}
\definecolor{reduceYellow}{RGB}{255,255,0}
\definecolor{reduceFrame}{RGB}{16,24,67}
\definecolor{reduceDots}{RGB}{25,25,25}

\begin{tikzpicture}[x=0.74918bp,y=-0.74918bp]
  % Preserve the CropBox and proportions of the source PDF.
  \useasboundingbox
    (336.461557,7.161336) rectangle (1917.288594,793.828288);

  \tikzset{
    reduce edge/.style={draw=black,line width=0.74918bp,
                        line cap=butt,line join=miter},
    reduce shape/.style={draw=reduceFrame,line width=0.74918bp,
                         line cap=butt,line join=miter}
  }

  \def\ReduceSquare#1#2#3#4#5{%
    \path[reduce shape,fill=#5] (#1,#2) rectangle (#3,#4);%
  }
  \def\ReduceCircle#1#2{%
    \path[reduce shape,fill=reduceCyan] (#1,#2) circle[radius=34.1984];%
  }
  \def\ReduceEdge#1#2#3#4{%
    \draw[reduce edge] (#1,#2) -- (#3,#4);%
  }
  \def\ReduceDot#1#2{%
    \path[fill=reduceDots,draw=none] (#1,#2) circle[radius=1.25];%
  }

  % The drawing order follows the source PDF.
  \ReduceSquare{533.218151}{159.265078}{599.723130}{225.775272}{reduceOrange}
  \ReduceSquare{443.828784}{159.265078}{510.338978}{225.775272}{reduceYellow}
  \ReduceSquare{354.439417}{159.265078}{420.949611}{225.775272}{reduceYellow}
  \ReduceCircle{477.834701}{47.804714}
  \ReduceCircle{660.450962}{47.804714}
  \ReduceSquare{626.252126}{159.265078}{692.762319}{225.775272}{reduceYellow}
  \ReduceSquare{719.291315}{159.265078}{785.801509}{225.775272}{reduceOrange}
  \ReduceSquare{808.638969}{159.265078}{875.143949}{225.775272}{reduceYellow}
  \ReduceSquare{904.217392}{159.265078}{970.722372}{225.775272}{reduceOrange}
  \ReduceSquare{999.801029}{159.265078}{1066.306009}{225.775272}{reduceYellow}
  \ReduceSquare{1190.957874}{159.265078}{1257.468068}{225.775272}{reduceYellow}
  \ReduceCircle{842.055701}{47.804714}
  \ReduceCircle{1033.520175}{47.804714}
  \ReduceSquare{1095.374238}{159.265078}{1161.884431}{225.775272}{reduceOrange}
  \ReduceSquare{1347.738606}{160.375667}{1414.243586}{226.880647}{reduceYellow}
  \ReduceSquare{1443.322243}{160.375667}{1509.832437}{226.880647}{reduceOrange}
  \ReduceCircle{1224.212971}{47.804714}
  \ReduceCircle{1380.988489}{48.910089}

  \ReduceEdge{477.834701}{81.998336}{387.694514}{159.265078}
  \ReduceEdge{477.834701}{81.998336}{477.083881}{159.265078}
  \ReduceEdge{477.834701}{81.998336}{566.473247}{159.265078}
  \ReduceEdge{660.450962}{81.998336}{659.507223}{159.265078}
  \ReduceEdge{660.450962}{81.998336}{566.473247}{159.265078}
  \ReduceEdge{660.450962}{81.998336}{752.546412}{159.265078}
  \ReduceEdge{842.055701}{81.998336}{752.546412}{159.265078}
  \ReduceEdge{842.055701}{81.998336}{841.888852}{159.265078}
  \ReduceEdge{842.055701}{81.998336}{937.472489}{159.265078}
  \ReduceEdge{1033.520175}{81.998336}{937.472489}{159.265078}
  \ReduceEdge{1033.520175}{81.998336}{1033.056126}{159.265078}
  \ReduceEdge{1033.520175}{81.998336}{1128.629334}{159.265078}
  \ReduceEdge{1224.212971}{81.998336}{1128.629334}{159.265078}
  \ReduceEdge{1224.212971}{81.998336}{1224.212971}{159.265078}
  \ReduceEdge{1380.988489}{83.108925}{1380.988489}{160.375667}
  \ReduceEdge{1380.988489}{83.108925}{1476.577340}{160.375667}

  \ReduceCircle{660.450962}{331.103935}
  \ReduceCircle{937.305640}{332.209310}
  \ReduceSquare{808.800604}{298.015687}{875.310798}{364.525881}{reduceOrange}
  \ReduceEdge{566.473247}{225.775272}{660.450962}{296.905098}
  \ReduceEdge{752.546412}{225.775272}{660.450962}{296.905098}
  \ReduceEdge{694.644585}{331.103935}{808.800604}{331.270784}
  \ReduceEdge{875.310798}{331.270784}{903.106803}{332.209310}
  \ReduceEdge{937.472489}{225.775272}{937.305640}{298.015687}

  \ReduceCircle{1129.573074}{331.437633}
  \ReduceSquare{1001.542516}{298.015687}{1068.047495}{364.525881}{reduceOrange}
  \ReduceEdge{971.499262}{332.209310}{1001.542516}{331.270784}
  \ReduceEdge{1095.374238}{331.437633}{1068.047495}{331.270784}
  \ReduceEdge{1128.629334}{225.775272}{1129.573074}{297.238796}

  \ReduceCircle{1572.150549}{48.910089}
  \ReduceSquare{1538.728603}{160.375667}{1605.238797}{226.880647}{reduceYellow}
  \ReduceEdge{1572.150549}{83.108925}{1476.577340}{160.375667}
  \ReduceEdge{1572.150549}{83.108925}{1571.983700}{160.375667}
  \ReduceSquare{1634.150605}{160.375667}{1700.655584}{226.880647}{reduceOrange}
  \ReduceEdge{1572.150549}{83.108925}{1667.405702}{160.375667}

  % Horizontal ellipsis in the middle chain.
  \ReduceDot{1296.625617}{332.240460}
  \ReduceDot{1304.300675}{332.240460}
  \ReduceDot{1311.954377}{332.240460}

  \ReduceEdge{1250.559476}{108.344841}{1224.212971}{81.998336}
  \ReduceDot{1263.028639}{119.178928}
  \ReduceDot{1268.451902}{124.602190}
  \ReduceDot{1273.873830}{130.024118}
  \ReduceEdge{1356.565963}{106.905768}{1380.988489}{83.108925}
  \ReduceDot{1341.293382}{117.973841}
  \ReduceDot{1335.870119}{123.397103}
  \ReduceDot{1330.448191}{128.819031}

  \ReduceSquare{1191.901614}{299.126276}{1258.411808}{365.631256}{reduceOrange}
  \ReduceEdge{1191.901614}{332.376159}{1163.766697}{331.437633}
  \ReduceEdge{1258.411808}{332.376159}{1281.979233}{332.376159}
  \ReduceSquare{1348.843981}{299.126276}{1415.354175}{365.631256}{reduceOrange}
  \ReduceEdge{1348.843981}{332.376159}{1324.165967}{332.376159}
  \ReduceCircle{1572.927439}{331.437633}
  \ReduceEdge{1538.728603}{331.437633}{1415.354175}{332.376159}
  \ReduceEdge{1476.577340}{226.880647}{1572.927439}{297.238796}

  \ReduceCircle{1762.368869}{48.910089}
  \ReduceSquare{1729.113772}{160.375667}{1795.623965}{226.880647}{reduceYellow}
  \ReduceSquare{1824.087367}{159.265078}{1890.597560}{225.775272}{reduceYellow}
  \ReduceEdge{1762.368869}{83.108925}{1667.405702}{160.375667}
  \ReduceEdge{1762.368869}{83.108925}{1762.368869}{160.375667}
  \ReduceEdge{1762.368869}{83.108925}{1857.342464}{159.265078}
  \ReduceEdge{1667.405702}{226.880647}{1572.927439}{297.238796}

  \ReduceEdge{842.055701}{364.525881}{883.877453}{406.347632}
  \ReduceEdge{1034.792398}{364.525881}{988.648214}{411.707657}
  \ReduceEdge{1225.156711}{365.631256}{1225.156711}{411.707657}
  \ReduceEdge{1382.099078}{365.631256}{1334.369828}{413.370934}

  \ReduceDot{1222.278404}{483.466841}
  \ReduceDot{1222.278404}{491.141899}
  \ReduceDot{1222.278404}{498.795601}
  \ReduceDot{1325.648893}{419.722904}
  \ReduceDot{1320.225631}{425.146166}
  \ReduceDot{1314.803703}{430.568094}

  \ReduceCircle{727.581626}{572.398913}
  \ReduceEdge{698.362190}{516.499275}{727.581626}{538.205291}
  \ReduceEdge{757.025265}{514.403234}{727.581626}{538.205291}
  \ReduceSquare{900.604068}{538.976967}{967.114261}{605.487161}{reduceOrange}
  \ReduceEdge{900.604068}{572.232064}{761.780463}{572.398913}
  \ReduceCircle{1029.109103}{573.175804}
  \ReduceEdge{967.114261}{573.342653}{994.910267}{573.175804}
  \ReduceCircle{1221.376538}{573.509502}
  \ReduceSquare{1093.345979}{540.087556}{1159.850959}{606.597750}{reduceOrange}
  \ReduceEdge{1063.302726}{573.175804}{1093.345979}{573.342653}
  \ReduceEdge{1187.177701}{573.509502}{1159.850959}{573.342653}
  \ReduceSquare{1287.453967}{540.087556}{1353.964161}{606.597750}{reduceOrange}
  \ReduceCircle{1525.949010}{572.398913}
  \ReduceEdge{1491.750173}{572.398913}{1353.964161}{573.342653}
  \ReduceEdge{1255.570160}{573.509502}{1287.453967}{573.342653}
  \ReduceEdge{1496.145602}{516.499275}{1525.370253}{538.205291}
  \ReduceEdge{1554.808678}{514.403234}{1525.370253}{538.205291}

  \ReduceDot{679.170954}{496.877499}
  \ReduceDot{684.594216}{502.300762}
  \ReduceDot{690.016144}{507.722690}
  \ReduceEdge{1029.109103}{538.976967}{1029.109103}{507.046235}
  \ReduceEdge{1221.376538}{539.310666}{1221.376538}{505.998214}
  \ReduceDot{980.713938}{416.885607}
  \ReduceDot{975.290675}{422.308870}
  \ReduceDot{969.868747}{427.730798}
  \ReduceDot{901.010075}{423.292133}
  \ReduceDot{895.586812}{417.868870}
  \ReduceDot{890.164884}{412.446942}
  \ReduceDot{776.305543}{497.903221}
  \ReduceDot{770.882281}{503.326484}
  \ReduceDot{765.460353}{508.748412}
  \ReduceDot{1226.380786}{418.572698}
  \ReduceDot{1226.380786}{426.247756}
  \ReduceDot{1226.380786}{433.901458}
  \ReduceDot{1030.009735}{485.684346}
  \ReduceDot{1030.009735}{493.359404}
  \ReduceDot{1030.009735}{501.013105}
  \ReduceDot{1571.514643}{495.685718}
  \ReduceDot{1566.091380}{501.108980}
  \ReduceDot{1560.669452}{506.530908}
  \ReduceDot{1479.556811}{498.365445}
  \ReduceDot{1484.980074}{503.788708}
  \ReduceDot{1490.402002}{509.210636}

  \ReduceCircle{1128.462485}{745.895832}
  \ReduceEdge{933.859165}{605.487161}{1128.462485}{711.702210}
  \ReduceEdge{1126.601076}{606.597750}{1128.462485}{711.702210}
  \ReduceEdge{1320.709064}{606.597750}{1128.462485}{711.702210}
\end{tikzpicture}
\endgroup%
	}
	\caption{Second phase of the reduction in the proof of Theorem~\ref{thm:np33}. Yellow nodes denote original check nodes of degree two, blue nodes denote variable nodes generated during the iterative construction, and red nodes denote auxiliary check nodes.}
	\label{fig:np33_reduction}
\end{figure*}

Theorem~\ref{thm:np33} demonstrates that the smallest nontrivial biregular class already captures the classical hardness of the MDP. The remaining results in this section record how the two degree parameters can be amplified while tracking the decision threshold.

\subsection{Degree Amplification for Fixed Degrees $J,K\geq 3$}

This subsection extends the framework to instances where the degrees on either side of the bipartition are scaled to arbitrary fixed integers $J, K \ge 3$. The key new point is the right-degree amplification. A degree-two auxiliary block only forces auxiliary supports to contain long cycles, so a Karp reduction would require girth larger than the target threshold. Instead, we use a degree-three single-port gadget. Any nonzero relative support inside the gadget induces an essentially cubic graph, and the cubic Moore bound then makes its weight exponential in the girth. Hence logarithmic girth suffices, and the gadget remains polynomial in the input size.

\begin{lemma}\label{lem:2lr}
	For any $2$-left regular Tanner graph, the cardinality of a minimum $(a,0)$-TS is exactly half of the graph girth.
\end{lemma}

\begin{proof}
	Let $\mathcal{Z}=(L_Z,R_Z,E_Z)$ be a $2$-left regular Tanner graph with girth $g$, and let $a_{\min}$ be the minimum trapping-set cardinality. Since $\mathcal{Z}$ is bipartite, $g$ is even.
	
	For the upper bound, take a shortest cycle $C$ of length $g$. The cycle alternates between $L_Z$ and $R_Z$, and therefore contains exactly $g/2$ variable nodes. Let $S$ be this set of variable nodes. Every check node on $C$ has exactly two neighbors in $S$, and every check node outside $C$ has zero neighbors in $S$. Hence $S$ is a $(a,0)$-TS and $a_{\min}\leq g/2$.
	
	For the lower bound, let $S\subseteq L_Z$ be any nonempty $(a,0)$-TS and consider the subgraph induced by $S\cup N(S)$. Every variable node has degree two in this induced subgraph, and every check node has even degree by the trapping-set condition. Thus all nodes in the induced subgraph have even degree, so the subgraph decomposes into edge-disjoint cycles. Each such cycle has length at least $g$ and contains at least $g/2$ variable nodes. Therefore $|S|\geq g/2$. Combining the two bounds gives $a_{\min}=g/2$.
\end{proof}

The deterministic relationship between girth and trapping-set cardinality in $2$-left regular graphs, outlined by Lemma~\ref{lem:2lr}, is sufficient for the completion block in the proof of Theorem~\ref{thm:np33}. For a polynomial-time right-degree amplification, however, we need an auxiliary block whose forbidden nonempty relative supports have size polynomially larger than the input threshold while the block itself remains polynomial in the input size. The next lemma supplies high-girth regular Tanner graphs of logarithmic size in the required threshold.

\begin{lemma}[Algorithmic high-girth supply]\label{lem:high-girth-supply}
	For every fixed integer $K\geq 3$, there is a constant $A_K>1$ and a deterministic algorithm which, on input $h\geq 3$, outputs a simple $(3,K)$-regular Tanner graph $\mathcal{B}_h$ such that
	\[
	\operatorname{girth}(\mathcal{B}_h)\geq h,
	\qquad
	|V(\mathcal{B}_h)|+|C(\mathcal{B}_h)|\leq A_K^h .
	\]
	The running time is polynomial in the output size.
\end{lemma}

\begin{proof}
	Let $m=\operatorname{lcm}(3,K)$. By Dirichlet's theorem on primes in arithmetic progressions~\cite[Ch.~7]{Apostol76}, choose once and for all an odd prime $q$ with $q\equiv -1\pmod m$, and put $R=q+1$. Since $K$ is fixed, both $q$ and $R$ are constants hard-wired into the reduction, and $R$ is divisible by both $3$ and $K$.
	
	On input $h$, set $d=2h$. Deterministically construct a monic irreducible polynomial $g_h\in\mathbb{F}_q[x]$ of degree $d$ by enumerating monic degree-$d$ polynomials and applying the standard irreducibility test over finite fields. There are $q^d$ candidates, and $q$ is constant, so this search is polynomial in the size of the graph constructed below.
	
	Apply Morgenstern's explicit construction~\cite[Theorem~4.13]{Morgenstern94} with the fixed prime power $q$ and the polynomial $g_h$. It outputs a simple $R$-regular Cayley graph $X_h$ whose full adjacency list is computable in time polynomial in its order. The estimates in the construction give
	\[
		|V(X_h)|\leq q^{3d}=q^{6h},
		\qquad
		\operatorname{girth}(X_h)>2d>h .
	\]
	Thus the family is selected deterministically for every requested $h$, and there is no gap between the infinite explicit family and the algorithmic quantifier used here.
	
	Take the bipartite double cover of $X_h$. On one side, partition the $R$ incident edges at every vertex into blocks of three and split the vertex into $R/3$ vertices, one for each block. On the other side, partition the incident edges into blocks of $K$ and split the vertex into $R/K$ vertices. The resulting bipartite graph is $(3,K)$-regular and simple.
	
	Vertex splitting does not decrease the girth: a cycle in the split graph projects to a nonbacktracking closed walk in the graph before splitting, and every nonbacktracking closed walk contains a cycle no longer than itself. The bipartite double cover also does not decrease girth. Since $R$ and $K$ are fixed, the size and running-time bounds follow after increasing the constant $A_K$.
\end{proof}

Fix $K\geq 4$ and a nonnegative target weight $w$. Set
\[
	g_w=4\left\lceil \log_2(w+1)\right\rceil+8 .
\]
We construct a single-port gadget $\mathcal{P}_{K,w}=(U,W,E_P)$ with a distinguished check node $\rho\in W$ such that every variable node in $U$ has degree three, $\deg_{\mathcal{P}_{K,w}}(\rho)=K-3$, every check node in $W\setminus\{\rho\}$ has degree $K$, and every nonempty variable set satisfying all checks except possibly $\rho$ has cardinality greater than $w$.

By Lemma~\ref{lem:high-girth-supply}, construct a $(3,K)$-regular Tanner graph $\mathcal{B}=(U_B,W_B,E_B)$ with
\[
	\operatorname{girth}(\mathcal{B})\geq 2g_w+8.
\]
Choose a check node $\rho\in W_B$. Let $L$ be the smallest odd integer with $L\geq g_w+2$. Starting at $\rho$, follow a nonbacktracking path of length $L$ and let its variable endpoint be $v$. This path is simple because $L<\operatorname{girth}(\mathcal{B})$. Moreover,
\[
	\operatorname{dist}_{\mathcal{B}}(\rho,v)\geq g_w+2,
\]
because a shorter $\rho$--$v$ path of length at most $g_w+1$, together with the chosen path of length at most $g_w+3$, would contain a cycle of length at most $2g_w+4$.

Write $N_{\mathcal{B}}(v)=\{a_1,a_2,a_3\}$ and choose three distinct variables $u_1,u_2,u_3\in N_{\mathcal{B}}(\rho)$. Delete $v$, delete the six edges $\{v,a_i\}$ and $\rho u_i$, and insert the three edges $a_i u_i$. Formally,
\begin{align*}
	U=&U_B\setminus\{v\},\quad W=W_B,\nonumber\\
	E_P
	=&\left(E_B\setminus
	\bigl(\{\{v,a_i\}:1\leq i\leq 3\}
	\cup\{\rho u_i:1\leq i\leq 3\}\bigr)\right)\\
	&\cup\{\{a_i, u_i\}:1\leq i\leq 3\}.
\end{align*}
The distance condition implies that no $a_i$ was adjacent to any $u_j$ in $\mathcal{B}$, so the resulting Tanner graph is simple. Every $a_i$ loses $v$ and gains $u_i$, while every $u_i$ loses $\rho$ and gains $a_i$. Hence all variable degrees remain three, all check degrees other than that of $\rho$ remain $K$, and $\deg_{\mathcal{P}_{K,w}}(\rho)=K-3$.

\begin{lemma}[Girth after the defect switch]\label{lem:defect-girth}
	The graph $\mathcal{P}_{K,w}$ satisfies
	\[
	\operatorname{girth}(\mathcal{P}_{K,w})\geq g_w .
	\]
\end{lemma}

\begin{proof}
	Suppose that $\mathcal{P}_{K,w}$ contains a cycle $Z$ of length less than $g_w$. The cycle must use at least one new edge $a_i u_i$, because all other edges form a subgraph of $\mathcal{B}$. Delete all new edges from $Z$. Each resulting component is a path whose endpoints lie in $A\cup U_0$, where $A=\{a_1,a_2,a_3\}$ and $U_0=\{u_1,u_2,u_3\}$. Let $R$ be one of these paths.
	
	If both endpoints of $R$ lie in $A$, adjoining the two-edge path through $v$ gives a cycle in $\mathcal{B}$ shorter than $g_w+2$. If both endpoints lie in $U_0$ and $R$ avoids $\rho$, adjoining the two-edge path through $\rho$ gives such a cycle. If both endpoints lie in $U_0$ and $R$ contains $\rho$, then the subpath from either endpoint $u_i$ to $\rho$, together with the deleted edge $\rho u_i$, gives an even shorter cycle in $\mathcal{B}$.
	
	It remains that $R$ has one endpoint in $A$ and one endpoint $u_i$ in $U_0$. If $R$ contains $\rho$, the $u_i$--$\rho$ subpath together with the deleted edge $\rho u_i$ is a cycle in $\mathcal{B}$ shorter than $g_w+1$. If $R$ avoids $\rho$, adjoining the edge from the $A$-endpoint to $v$ and the edge $u_i\rho$ gives a $v$--$\rho$ path of length less than $g_w+2$. These alternatives contradict either the girth of $\mathcal{B}$ or the distance lower bound between $\rho$ and $v$. Therefore no such cycle $Z$ exists.
\end{proof}

Define the relative distance of the port gadget by
\[
	d_{\mathrm{rel}}(\mathcal{P}_{K,w},\rho)
	=
	\min\left\{
	|Y|:
	\begin{array}{l}
		\emptyset\neq Y\subseteq U,\\
		|N_{\mathcal{P}_{K,w}}(c)\cap Y|\equiv 0\pmod 2\\
		\forall c\in W\setminus\{\rho\}
	\end{array}
	\right\}.
\]
If no such nonempty set exists, the minimum is interpreted as $+\infty$.

\begin{lemma}[Relative-distance bound]\label{lem:relative-distance}
	The gadget $\mathcal{P}_{K,w}$ satisfies
	\[
	d_{\mathrm{rel}}(\mathcal{P}_{K,w},\rho)>w.
	\]
	Moreover, its order is $(w+1)^{O_K(1)}$, and it can be constructed in deterministic time $(w+1)^{O_K(1)}$.
\end{lemma}

\begin{proof}
	Let $Y\subseteq U$ be nonempty and satisfy all parity constraints in $W\setminus\{\rho\}$. Construct a multigraph $Q_Y$ with vertex set $Y$. For every $c\in W\setminus\{\rho\}$, pair the variables in $N_{\mathcal{P}_{K,w}}(c)\cap Y$ arbitrarily and add one edge to $Q_Y$ for each pair. At the port $\rho$, pair as many variables as possible. If $|N_{\mathcal{P}_{K,w}}(\rho)\cap Y|$ is odd, leave exactly one variable unpaired.
	
	No loop is created, because each paired edge at a check node joins two distinct variables. If two parallel edges were created, then the two corresponding Tanner checks would be adjacent to the same two variables, producing a $4$-cycle in $\mathcal{P}_{K,w}$; this is impossible since $g_w\geq 8$ and Lemma~\ref{lem:defect-girth} gives $\operatorname{girth}(\mathcal{P}_{K,w})\geq g_w$. Hence $Q_Y$ is a simple graph.
	
	Every variable in $\mathcal{P}_{K,w}$ has degree three. Consequently, every vertex of $Q_Y$ has degree three, except that there may be one vertex of degree two. A cycle of length $q$ in $Q_Y$ gives a nonbacktracking closed walk of length $2q$ in $\mathcal{P}_{K,w}$, and such a walk contains a cycle of length at most $2q$. By Lemma~\ref{lem:defect-girth},
	\[
		\operatorname{girth}(Q_Y)\geq \frac{g_w}{2}.
	\]
	
	If $Q_Y$ has a cubic component, use such a component. Otherwise, the unique noncubic component contains the possible degree-two vertex. Take two disjoint copies of this component and join their degree-two vertices by one edge. The added edge is a bridge, and the resulting graph $\widehat Q$ is cubic. In either case we obtain a simple cubic graph $\widehat Q$ such that
	\[
		|V(\widehat Q)|\leq 2|Y|,
		\qquad
		\operatorname{girth}(\widehat Q)\geq \frac{g_w}{2}.
	\]
	The cubic Moore bound~\cite{Biggs93} gives
	\[
		2|Y|
		\geq |V(\widehat Q)|
		\geq 2^{g_w/4-1},
	\]
	and hence
	\[
		|Y|\geq 2^{g_w/4-2}
		=2^{\lceil\log_2(w+1)\rceil}
		\geq w+1.
	\]
	This proves the relative-distance claim. The size and construction-time bounds follow from Lemma~\ref{lem:high-girth-supply} and $g_w=O(\log(w+1))$.
\end{proof}

\begin{theorem}\label{thm:k3}
	For every fixed integer $K\geq 3$, the problem \MDtK\ is $\mathrm{NP}$-complete.
\end{theorem}

\begin{proof}
	Membership in $\mathrm{NP}$ is immediate, since a nonempty set of at most $w$ variable nodes is a polynomially verifiable certificate. The case $K=3$ is Theorem~\ref{thm:np33}. Fix $K\geq 4$. We reduce from \MDtt. Let $(\mathcal{G},w)$ be an instance, where $\mathcal{G}=(L,R,E)$ is $(3,3)$-regular. Replacing $w$ by $\min\{w,|L|\}$ does not change the answer, so assume $w\leq |L|$.
	
	For every source check $r\in R$, take a fresh copy $(\mathcal{P}^{r}_{K,w},\rho_r)$ of the gadget from Lemma~\ref{lem:relative-distance}, and identify its port $\rho_r$ with $r$. Denote the resulting Tanner graph by $\mathcal{G}^{\star}$. Every original variable and every auxiliary variable has degree three. Each identified check has its three original neighbors and $K-3$ auxiliary neighbors, and therefore has degree $K$. Every other auxiliary check also has degree $K$. Hence $\mathcal{G}^{\star}$ is $(3,K)$-regular.
	
	Because a $(3,3)$-regular bipartite graph satisfies $|L|=|R|$, the output size is bounded by
	\[
		|\mathcal{G}^{\star}|
		\leq |\mathcal{G}|(w+1)^{O_K(1)}
		\leq |\mathcal{G}|^{O_K(1)}.
	\]
	Thus, for fixed $K$, the construction is a deterministic polynomial-time transformation.
	
	Suppose first that $S\subseteq L$ is a nonempty $(a,0)$-TS of $\mathcal{G}$ with $|S|\leq w$. Select the same original variables in $\mathcal{G}^{\star}$ and select no auxiliary variable. The parity at each identified check is the original parity plus zero, and every auxiliary internal check sees zero selected variables. Hence $S$ is a $(a,0)$-TS of $\mathcal{G}^{\star}$ of the same size.
	
	Conversely, let $Z$ be a nonempty $(a,0)$-TS of $\mathcal{G}^{\star}$ with $|Z|\leq w$. Write
	\[
		Z=S\mathbin{\dot\cup}\bigcup_{r\in R}Y_r,
	\]
	where $S=Z\cap L$ and $Y_r$ consists of the selected auxiliary variables in the copy $\mathcal{P}^{r}_{K,w}$. All internal checks of $\mathcal{P}^{r}_{K,w}$ are checks of $\mathcal{G}^{\star}$, so $Y_r$ satisfies every gadget parity constraint except possibly the port constraint. By Lemma~\ref{lem:relative-distance}, either $Y_r=\emptyset$ or $|Y_r|>w$. The latter alternative contradicts $|Z|\leq w$. Therefore $Y_r=\emptyset$ for every $r\in R$. It follows that $Z=S$, and parity at each identified check reduces exactly to the original parity constraint at $r$. Since $Z$ is nonempty, $S$ is a nonempty $(a,0)$-TS of $\mathcal{G}$, and $|S|\leq w$.
	
	We have proved the threshold-preserving equivalence $(\mathcal{G},w)\in \MDtt\ \Longleftrightarrow\
		(\mathcal{G}^{\star},w)\in \MDtK$.
	This is a polynomial-time Karp reduction from \MDtt\ to \MDtK, completing the proof.
\end{proof}

Building upon the $(3,K)$-regular NP-completeness established in Theorem~\ref{thm:k3}, the concluding stage propagates the left-degree scaling to arbitrary left degrees $J > 3$. This is engineered via the explicit introduction of global parity-check matrices unifying multiple independent base instances.

\begin{theorem}\label{thm:JK}
	For fixed integers $J\geq 4$ and $K\geq 3$, there is a Karp reduction from $\textsc{MinDistance}_{(J-1,K)}$ to \MDJK. The same transformation is also an FPT reduction when parameterized by the target set size $w$. Consequently, for every fixed $J,K\geq 3$, \MDtK\ admits both a Karp reduction and an FPT reduction to \MDJK.
\end{theorem}

\begin{proof}
	Fix $J\geq 4$ and $K\geq 3$, and let $\mathcal{G}=(L,R,E)$ be a $(J-1,K)$-regular instance with $L=\{v_1,\ldots,v_n\}$.
	
	Construct $K$ disjoint copies $\mathcal{G}_i=(L_i,R_i,E_i)$ for $i=1,\ldots,K$, and write $v_j^{(i)}$ for the copy of $v_j$ in $L_i$. Add global check nodes $c_1,\ldots,c_n$, where each $c_j$ is adjacent to $v_j^{(1)},\ldots,v_j^{(K)}$. Denote the resulting graph by $\mathcal{G}'$. Every copied variable gains exactly one new incident edge, so its degree becomes $J$; each original copied check still has degree $K$, and each global check has degree $K$. Thus $\mathcal{G}'$ is $(J,K)$-regular. This is the construction illustrated in Fig.~\ref{fig:gc_reduction}.\footnote{The construction is analogous to the global-check component of GC-LDPC codes~\cite{Juane16}.}
	
	Define the lifting map constructed in the present proof by
	\begin{align*}
	&\Lambda_J:\mathcal{A}(\mathcal{G})\longrightarrow \mathcal{A}(\mathcal{G}'),\\
	&\Lambda_J(S)=\{v_j^{(1)}:v_j\in S\}\cup\{v_j^{(2)}:v_j\in S\}.
	\end{align*}
	If an instance has no nonempty feasible set, we interpret its optimum as $\infty$ and use the convention $2\infty=\infty$.
	For local check nodes in the first two copies, the parity constraints are inherited from $S$; all local check nodes in the remaining copies see no selected variables. Each global check $c_j$ sees either two selected copies of $v_j$ or none. Hence $\Lambda_J(S)$ is a $(a,0)$-TS of $\mathcal{G}'$ and $|\Lambda_J(S)|=2|S|$. Therefore, the optimum in $\mathcal{G}'$ is at most twice the optimum in $\mathcal{G}$.
	
	Conversely, let $Y\in\mathcal{A}(\mathcal{G}')$ be nonempty. For each copy define the projection
	\[
	\pi_i(Y)=\{v_j\in L:v_j^{(i)}\in Y\},
	\qquad 1\leq i\leq K.
	\]
	Each nonempty $\pi_i(Y)$ is a $(a,0)$-TS of $\mathcal{G}$, since the local checks in copy $i$ impose exactly the original parity constraints. Moreover, the global check $c_j$ enforces that the number of copies containing $v_j$ is even. Thus every selected variable position contributes to at least two copies. If $\mathcal{G}$ has no nonempty feasible set, then no nonempty projection can occur, contradicting $Y\neq\emptyset$; hence the no-solution case maps to the no-solution case. Otherwise, if $X_{\min}$ is a minimum $(a,0)$-TS of $\mathcal{G}$, then for any nonempty projection $\pi_i(Y)$ we have $|\pi_i(Y)|\geq |X_{\min}|$, and the even-copy condition gives
	\[
	|Y|\geq 2|\pi_i(Y)|\geq 2|X_{\min}|.
	\]
	Hence, the optimum in $\mathcal{G}'$ is at least twice the optimum in $\mathcal{G}$.
    
\usetikzlibrary{decorations.pathreplacing, calc}
\definecolor{nodeblue}{RGB}{115, 140, 255}
\definecolor{nodeorange}{RGB}{252, 224, 180}
\definecolor{boxborder}{RGB}{50, 60, 90}
\begin{figure*}[h]
	\centering
	\begin{tikzpicture}[
		>=stealth,
		box/.style={rectangle, draw=boxborder, minimum width=4.6cm, minimum height=1.2cm, fill=white, font=\normalsize},
		circ/.style={circle, draw=black!80, fill=nodeblue, minimum size=0.85cm, inner sep=0pt},
		sq/.style={rectangle, draw=black!80, fill=nodeorange, minimum size=0.85cm, inner sep=0pt}
		]
		\node[box] (b1) at (-5, 3) {$3$-regular checks};
		\node[box] (b2) at (0, 3) {$3$-regular checks};
		\node[box] (b3) at (5, 3) {$3$-regular checks};
		
		\draw [decorate, decoration={brace, amplitude=8pt}]
		($(b1.north west) + (0, 0.25)$) -- ($(b3.north east) + (0, 0.25)$)
		node [midway, above=10pt] {Copies};
		
		\foreach \gx in {-5, 0, 5} {
			\node[circ] (c\gx_1) at (\gx - 1.65, 0.5) {};
			\node[circ] (c\gx_2) at (\gx - 0.55, 0.5) {};
			\node at (\gx + 0.55, 0.5) {$\cdots$};
			\node[circ] (c\gx_3) at (\gx + 1.65, 0.5) {};
			
			\coordinate (block_bot_A) at (\gx - 2.0, 2.4);
			\coordinate (block_bot_B) at (\gx - 0.8, 2.4);
			\coordinate (block_bot_C) at (\gx + 0.8, 2.4);
			\coordinate (block_bot_D) at (\gx + 2.0, 2.4);
			
			\draw (c\gx_1.north) -- (block_bot_A);
			\draw (c\gx_1.north) -- (block_bot_B);
			\draw (c\gx_1.north) -- (block_bot_C);
			\draw (c\gx_2.north) -- (block_bot_A);
			\draw (c\gx_2.north) -- (block_bot_B);
			\draw (c\gx_2.north) -- (block_bot_D);
			\draw (c\gx_3.north) -- (block_bot_B);
			\draw (c\gx_3.north) -- (block_bot_C);
			\draw (c\gx_3.north) -- (block_bot_D);
		}
		
		\node[sq] (s1) at (-1.65, -3.5) {};
		\node[sq] (s2) at (-0.55, -3.5) {};
		\node at (0.55, -3.5) {$\cdots$};
		\node[sq] (s3) at (1.65, -3.5) {};
		
		\draw [decorate, decoration={brace, amplitude=6pt, mirror}]
		($(s1.south west) + (0, -0.25)$) -- ($(s3.south east) + (0, -0.25)$)
		node [midway, below=8pt] {Auxiliary checks};
		
		\foreach \gx in {-5, 0, 5} {
			\draw (c\gx_1.south) -- (s1.north);
			\draw (c\gx_2.south) -- (s2.north);
			\draw (c\gx_3.south) -- (s3.north);
		}
	\end{tikzpicture}
	\caption{Schematic reduction from the $(3,3)$-regular case to the $(4,3)$-regular case using copies and global check nodes.}
	\label{fig:gc_reduction}
\end{figure*}
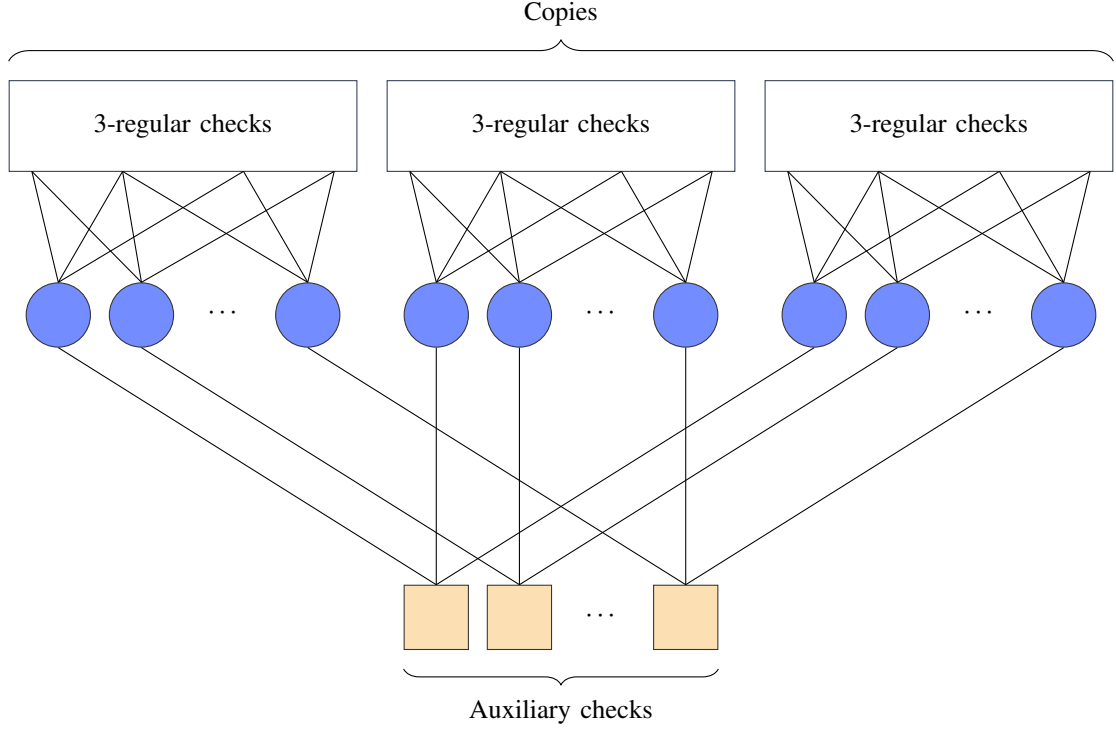
	
	The optimum is therefore scaled by exactly a factor of two. For the decision problem, map the target $w$ to $2w$. The transformation is polynomial, and the new parameter is bounded by a computable function of the old parameter. Hence it is both a Karp reduction and an FPT reduction for fixed $J$ and $K$. Iterating over the left degree gives the reduction from \MDtK\ to \MDJK.
\end{proof}

\begin{corollary}\label{cor:jk-np}
	For every fixed pair of integers $J,K\geq 3$, the problem \MDJK\ is $\mathrm{NP}$-complete.
\end{corollary}

\begin{proof}
	The case $J=3$ is Theorem~\ref{thm:k3}. For $J>3$, iterate the Karp reduction of Theorem~\ref{thm:JK}.
\end{proof}

Combining Theorems~\ref{thm:np33}, \ref{thm:k3}, and~\ref{thm:JK}, we obtain the full classical biregular hardness result: for every fixed $J,K\geq 3$, the MDP remains $\mathrm{NP}$-complete on $(J,K)$-regular Tanner graphs. The proof separates the right-degree amplification, handled by the single-port high-girth gadget, from the left-degree amplification, handled by the replica-and-global-check construction.

\section{Discussion}
\label{sec:discussion}

The intractability results established in this work have direct implications for the design and analysis of LDPC codes. Since minimum distance computation remains NP-complete even for every fixed biregular degree pair $J,K\geq3$, practical code design cannot rely on exact minimum distance evaluation except for highly restricted families. A more viable approach is to combine heuristic search, provable lower bounds, and additional algebraic or combinatorial structures. In particular, regular LDPC codes arising from finite geometries, incidence designs, or other highly symmetric constructions may contain exploitable features that are absent from arbitrary regular Tanner graphs.

These findings simultaneously contextualize the error-floor phenomenon in iterative LDPC decoding. Low-weight codewords and minimal even-parity topological structures are primary contributors to high-SNR error floors. The demonstrated intractability confirms that the exhaustive identification of these structures in arbitrary regular graphs is fundamentally hard. Effective error-floor analysis and mitigation must therefore rely on structure-aware partial enumeration, probabilistic approximations, and algorithmic interventions that bypass the need to solve the exact minimum distance problem as a subroutine.

A prominent open direction involves isolating the structural impact of short cycles. In LDPC code design, it is well established that short cycles---particularly 4-cycles---severely degrade iterative decoding performance by compromising the independence of passed messages and dominating the error floor \cite{MacKayGood, Richardson03, Hu05}. High-performance codes are thus conventionally constructed to be 4-cycle-free. In the hypergraph formulation, a 4-cycle-free Tanner graph exactly corresponds to a \emph{linear hypergraph}, defined by the property that any two hyperedges intersect in at most one vertex. This practical design requirement motivates a fundamental theoretical inquiry: does the computational hardness of the minimum distance problem persist when instances are strictly restricted to 4-cycle-free configurations? This structural consideration leads to the following problem formulation.
\begin{itemize}	
	\item \MLECJK: Given a $J$-uniform, $K$-regular, linear hypergraph $\mathcal{H}=(V,\mathcal{E})$ and an integer $w>0$, decide whether there exists a nonempty set $Y\subseteq\mathcal{E}$ with $|Y|\leq w$ such that
	\[
	\deg_Y(v)\equiv 0\pmod 2
	\qquad\text{for every } v\in V.
	\]
\end{itemize}

\begin{conj}\label{conj:main}
	For any fixed integers $J,K\geq3$, the problem \MLECJK\ is $\mathrm{NP}$-complete.
\end{conj}

Resolving Conjecture~\ref{conj:main} would formalize the complexity classification of the minimum distance problem for regular LDPC codes under realistic decoding constraints. Mathematically, a structural proof for this hypothesis would demand a two-stage approach leveraging linearity-preserving reductions. The first stage necessitates an operation transforming a general $J$-uniform, $K$-regular hypergraph into a strict linear hypergraph. As modeled in our intermediate transformations, analyzing such an operation requires establishing an explicit bijection between families of feasible even covers, mirroring the functionality of the mapping constructed in Lemma~\ref{lem:3.1}. Any requisite hyperedge subdivision would need to enforce deterministic parity propagation bounds without inadvertently introducing smaller alternative minimal covers. 

The second stage requires an induction mechanism to lift the regularity parameter $K$ strictly within the linear domain. While Theorem~\ref{thm:k3} uses a high-girth single-port gadget and Theorem~\ref{thm:JK} uses global parity constraints over multiple graph copies, a linear-hypergraph analog must natively achieve proper parity bounding while strictly preventing multiple intersections between any pair of hyperedges. Formulating explicit and invertible solution mappings under this strict pairwise intersection constraint constitutes the primary technical and algorithmic obstacle for future research.

\section{Conclusion}
\label{sec:conclusion}

This work provides a complexity-theoretic characterization of the MDP for regular LDPC codes. The main results are as follows:
\begin{itemize}
	\item[(i)] For every fixed $J\geq3$, the standard
	at-most-weight MDP for $J$-left regular Tanner graphs is
	$\mathrm{NP}$-complete, while its exact-weight variant is
	$\mathrm{W}[1]$-complete when parameterized by the prescribed weight;
	\item[(ii)] For every fixed $J,K\geq3$, the MDP for $(J,K)$-regular Tanner graphs is $\mathrm{NP}$-complete.
\end{itemize}

The reductions are organized around a degree-preserving transformation framework consisting of hyperedge decomposition, check node splitting, and controlled variable replication. The key technical point is not merely that these transformations preserve feasibility, but that their proofs construct explicit maps between solution spaces: the even-cover bijection $T$ in Lemma~\ref{lem:3.1}, the trapping-set bijection $Q$ in Lemma~\ref{lem:3.2}, and the replication bijections $F_{v,m}$ in Lemma~\ref{lem:3.3}. For the right-degree amplification to $(3,K)$, the single-port gadget supplies the additional ingredient needed for a polynomial-time Karp reduction: relative supports inside the gadget are forced to be larger than the target threshold while the gadget size remains polynomial. By unifying minimum distance, even covers, and trapping sets within a single parity-based framework, the paper clarifies why computing the minimum distance exactly remains intrinsically difficult for regular LDPC codes under natural regularity constraints.

\bibliographystyle{IEEEtran}
\bibliography{main}

@article{Vardy97,
	author={Vardy, Alexander},
	journal={IEEE Transactions on Information Theory}, 
	title={The intractability of computing the minimum distance of a code}, 
	year={1997},
	volume={43},
	number={6},
	pages={1757-1766},
	doi={10.1109/18.641542}}

@article{Hsieh25,
	author = {Hsieh, Jun-Ting and Kothari, Pravesh K and Mohanty, Sidhanth and Correia, David Munhá and Sudakov, Benny},
	title = {Small Even Covers, Locally Decodable Codes and Restricted Subgraphs of Edge-Colored Kikuchi Graphs},
	journal = {International Mathematics Research Notices},
	volume = {2025},
	number = {5},
	pages = {rnaf045},
	year = {2025},
	month = {03}}

@ARTICLE{Berlekamp78,
	author={Berlekamp, E. and McEliece, R. and van Tilborg, H.},
	journal={IEEE Transactions on Information Theory}, 
	title={On the inherent intractability of certain coding problems (Corresp.)}, 
	year={1978},
	volume={24},
	number={3},
	pages={384-386}}

@ARTICLE{Dehghan20,
author={Dehghan, Ali and Banihashemi, Amir H.},
journal={IEEE Transactions on Information Theory}, 
title={On Finding Bipartite Graphs With a Small Number of Short Cycles and Large Girth}, 
year={2020},
volume={66},
number={10},
pages={6024-6036},
doi={10.1109/TIT.2020.3017127}}

@ARTICLE{Dehghan19,
	author={Dehghan, Ali and Banihashemi, Amir H.},
	journal={IEEE Transactions on Information Theory}, 
	title={Hardness Results on Finding Leafless Elementary Trapping Sets and Elementary Absorbing Sets of LDPC Codes}, 
	year={2019},
	volume={65},
	number={7},
	pages={4307-4315}}

@ARTICLE{Dumer03,
	author={Dumer, I. and Micciancio, D. and Sudan, M.},
	journal={IEEE Transactions on Information Theory}, 
	title={Hardness of approximating the minimum distance of a linear code}, 
	year={2003},
	volume={49},
	number={1},
	pages={22-37}}

@Inbook{Feige08,
	author="Feige, Uriel",
	title="Small Linear Dependencies for Binary Vectors of Low Weight",
	bookTitle="Building Bridges: Between Mathematics and Computer Science",
	year="2008",
	publisher="Springer Berlin Heidelberg",
	address="Berlin, Heidelberg",
	pages="283--307"}

@inproceedings{Hsieh23,
	author    = {Hsieh, Jun-Ting and Kothari, Pravesh K. and Mohanty, Sidhanth},
	title     = {A Simple and Sharper Proof of the Hypergraph Moore Bound},
	booktitle = {Proc. of the 2023 Annu. ACM-SIAM Symp. on Discrete Algorithms (SODA)},
	year      = {2023},
	pages     = {2324--2344},
	address   = {Philadelphia, PA, USA},
	publisher = {SIAM}
}

@book{Garey79,
	asin = {0716710455},
	author = {Garey, M. R. and Johnson, D. S.},
	dewey = {519.4},
	ean = {9780716710455},
	edition = {First Edition},
	publisher = {W. H. Freeman},
	title = {Computers and Intractability: A Guide to the Theory of NP-Completeness (Series of Books in the Mathematical Sciences)},
	year = 1979
}

@article{Arora97,
	author = {Arora, Sanjeev and Babai, László and Stern, Jacques and Sweedyk, Z.},
	journal = {J. Comput. Syst. Sci.},
	number = 2,
	pages = {317-331},
	title = {The Hardness of Approximate Optima in Lattices, Codes, and Systems of Linear Equations.},
	volume = 54,
	year = 1997
}

@article{Gallager62,
	author={Gallager, R.},
	journal={IRE Transactions on Information Theory}, 
	title={Low-density parity-check codes}, 
	year={1962},
	volume={8},
	number={1},
	pages={21-28}
}

@article{Downey99,
	author = {Downey, Rod G. and Fellows, Michael R. and Vardy, Alexander and Whittle, Geoff},
	title = {The Parametrized Complexity of Some Fundamental Problems in Coding Theory},
	journal = {SIAM Journal on Computing},
	volume = {29},
	number = {2},
	pages = {545-570},
	year = {1999}
}

@book{Flum06,
	title={Parameterized complexity},
	author={J\"org Flum, Martin Grohe},
	year={2006},
	publisher={Springer Berlin, Heidelberg},
	edition = {First Edition},
	doi={https://doi.org/10.1007/3-540-29953-X}
}

@book{Rodney13,
	title={Fundamentals of parameterized complexity},
	author={Rodney G. Downey and Michael R. Fellows},
	year={2013},
	publisher={Springer London},
	edition = {First Edition}
}

@book{Erik26,
  author    = {Erik D. Demaine and William Gasarch and Mohammad T. Hajiaghayi},
  title     = {{Computational Intractability}: A Guide to Algorithmic Lower Bounds},
  publisher = {The MIT Press},
  year      = {2026},
  month     = sep,
  isbn      = {978-0-262-55077-2},
  url       = {https://mitpress.mit.edu/9780262550772/computational-intractability/},
  note      = {Forthcoming}
}

@inproceedings{Richardson03,
	title={Error floors of LDPC codes},
	author={Richardson, Tom},
	booktitle={Proceedings of the Annual Allerton Conference on Communication, Control, and Computing},
	volume={41},
	number={3},
	pages={1426--1435},
	year={2003}
}

@article{Arvind16,
	title={Solving linear equations parameterized by hamming weight},
	author={Arvind, Vikraman and K{\"o}bler, Johannes and Kuhnert, Sebastian and Tor{\'a}n, Jacobo},
	journal={Algorithmica},
	volume={75},
	number={2},
	pages={322--338},
	year={2016},
	publisher={Springer}
}

@article{Bhattacharyya21,
	author = {Bhattacharyya, Arnab and Bonnet, \'{E}douard and Egri, L\'{a}szl\'{o} and Ghoshal, Suprovat and S., Karthik C. and Lin, Bingkai and Manurangsi, Pasin and Marx, D\'{a}niel},
	title = {Parameterized Intractability of Even Set and Shortest Vector Problem},
	year = {2021},
	publisher = {Association for Computing Machinery},
	address = {New York, NY, USA},
	volume = {68},
	number = {3},
	issn = {0004-5411},
	doi = {10.1145/3444942},
	journal = {J. ACM},
	month = mar,
	articleno = {16},
	numpages = {40}
}

@INPROCEEDINGS{Juane16,
	author={Li, Juane and Lin, Shu and Abdel-Ghaffar, Khaled and Ryan, William E. and Costello, Daniel J.},
	booktitle={2016 Information Theory and Applications Workshop (ITA)}, 
	title={Globally coupled LDPC codes}, 
	year={2016},
	volume={},
	number={},
	pages={1-10},
	doi={10.1109/ITA.2016.7888167}}

@ARTICLE{MacKayGood,
  author={MacKay, D.J.C.},
  journal={IEEE Transactions on Information Theory}, 
  title={Good error-correcting codes based on very sparse matrices}, 
  year={1999},
  volume={45},
  number={2},
  pages={399-431},
  doi={10.1109/18.748992}}

@ARTICLE{Hu05,
  author={Xiao-Yu Hu and Eleftheriou, E. and Arnold, D.M.},
  journal={IEEE Transactions on Information Theory}, 
  title={Regular and irregular progressive edge-growth tanner graphs},
  year={2005},
  volume={51},
  number={1},
  pages={386-398},
  doi={10.1109/TIT.2004.839541}}

@book{Biggs93,
	author    = {Norman Biggs},
	title     = {Algebraic Graph Theory},
	edition   = {2},
	series    = {Cambridge Mathematical Library},
	publisher = {Cambridge University Press},
	address   = {Cambridge},
	year      = {1993},
	doi       = {10.1017/CBO9780511608704}
}

@book{Apostol76,
	author    = {Tom M. Apostol},
	title     = {Introduction to Analytic Number Theory},
	series    = {Undergraduate Texts in Mathematics},
	publisher = {Springer},
	address   = {New York},
	year      = {1976},
	doi       = {10.1007/978-1-4757-5579-4}
}

@article{Morgenstern94,
	author  = {Moshe Morgenstern},
	title   = {Existence and Explicit Constructions of {$q+1$}-Regular {Ramanujan} Graphs for Every Prime Power {$q$}},
	journal = {Journal of Combinatorial Theory, Series B},
	volume  = {62},
	number  = {1},
	pages   = {44--62},
	year    = {1994},
	doi     = {10.1006/jctb.1994.1058}
}

%\begin{IEEEbiography}[{\includegraphics[width=1in,height=1.25in,clip,keepaspectratio]{author1-photo}}]{Michael Shell}
%	Use $\backslash${\tt{begin\{IEEEbiography\}}} and then for the 1st argument use $\backslash${\tt{includegraphics}} to declare and link the author photo.
%	Use the author name as the 3rd argument followed by the biography text.
%\end{IEEEbiography}

\begin{IEEEbiographynophoto}{Chenyuan Jia}
    received the B.Sc. degree in mathematics from Ocean University of China, in 2024. He is currently pursuing the Ph.D. degree in mathematics with the School of Mathematics, Shandong University, China. His research interests include coding theory, combinatorial optimization and computational complexity theory.
\end{IEEEbiographynophoto}

\begin{IEEEbiographynophoto}{Qingqing Peng}
	received the B.Sc. degree in mathematics from Shaanxi University of Science \& Technology, China, in 2021. She is currently pursuing the Ph.D. degree in mathematics with the School of Mathematics, Shandong University, China. Her research interests include coding theory.
\end{IEEEbiographynophoto}

\begin{IEEEbiographynophoto}{Ke Liu}
    received the B.Sc. degree in mathematics from Ocean University of China, and the Ph.D. degree from Tsinghua University, China, in 2017 and 2021, respectively. He joined Huawei Technologies Company Ltd. in 2021. His research interests include coding theory and graph theory and its applications.
\end{IEEEbiographynophoto}

\begin{IEEEbiographynophoto}{Guanghui Wang}
received the B.S. degree in mathematics from Shandong University, China, and the Ph.D. degree from the Universit\'{e} Paris-Sud, France, in 2001 and 2007, respectively. He is currently a Professor with the School of Mathematics, Shandong University, Shandong, China. His research interests include applied mathematics and their applications.
\end{IEEEbiographynophoto}

\begin{IEEEbiographynophoto}{Guiying Yan}
received the B.S., M.S., and Ph.D. degrees in mathematics from Shandong University. She is currently a Professor with the Academy of Mathematics and Systems Science, Chinese Academy of Sciences, Beijing, China. Her research interests include applied mathematics and their applications.
\end{IEEEbiographynophoto}

\vfill

\end{document}